\documentclass[a4paper,UKenglish,cleveref, autoref, thm-restate]{lipics-v2021}

\usepackage{xspace}

\nolinenumbers 

\title{Efficient Algorithms for Minimum Covering of Orthogonal Polygons with Squares} 

\titlerunning{Minimum Square Covering of Orthogonal Polygons} 

\author{Anubhav Dhar}{Indian Institute of Technology Kharagpur, India}{anubhavdhar@kgpian.iitkgp.ac.in}{}{}

\author{Subham Ghosh}{Indian Institute of Technology Kharagpur, India}{subham.g@kgpian.iitkgp.ac.in}{}{}

\author{Sudeshna Kolay}{Indian Institute of Technology Kharagpur, India}{skolay@cse.iitkgp.ac.in}{}{}

\authorrunning{A.~Dhar, S.~Ghosh, S.~Kolay} 

\Copyright{Anubhav Dhar, Subham Ghosh, Sudeshna Kolay} 



\category{} 

\relatedversion{}

\hideLIPIcs  

\newcommand{\defparproblem}[4]{
  \vspace{1mm}
\noindent\fbox{
  \begin{minipage}{0.96\textwidth}
  \begin{tabular*}{\textwidth}{@{\extracolsep{\fill}}lr} #1  & {\bf{Parameter:}} #3 \\ \end{tabular*}
  {\bf{Input:}} #2  \\
  {\bf{Question:}} #4
  \end{minipage}
  }
  \vspace{1mm}
}

\newcommand{\defproblem}[3]{
  \vspace{1mm}
\noindent\fbox{
  \begin{minipage}{0.96\textwidth}
  \begin{tabular*}{\textwidth}{@{\extracolsep{\fill}}lr} #1 \\ \end{tabular*}
  {\bf{Input:}} #2  \\
  {\bf{Question:}} #3
  \end{minipage}
  }
  \vspace{1mm}
}

\usepackage{todonotes}
\usepackage{algorithm}
\usepackage{algpseudocode}

\newcommand{\pnp}{$\mbox{P} = \mbox{NP}$\xspace}

\newcommand{\OO}{\mathcal{O}}

\newcommand{\OPCS}{\textsc{OPCS}\xspace}

\newcommand{\OPCSH}{\textsc{OPCSH}\xspace}
\newcommand{\Ex}{\mathsf{ext}}
\newcommand{\GS}{G^\mathscr{R}(P)}

 \definecolor{babyblue}{rgb}{0.54, 0.81, 0.94}
 \definecolor{b1}{rgb}{0.63, 0.79, 0.95}
 \definecolor{b2}{rgb}{0.74, 0.83, 0.9}
 \definecolor{b3}{rgb}{0.67, 0.9, 0.93}
 \definecolor{gentlegreen}{rgb}{0.00, 0.51, 0.00}

\begin{document}

\maketitle

\begin{abstract}
Let $P$ be an orthogonal polygon (polygon having axis-parallel edges) of $n$ vertices, without holes. The {\sc Orthogonal Polygon Covering with Squares (OPCS)} problem takes as input such an orthogonal polygon $P$ with integral vertex coordinates, and asks to find the minimum number of axis-parallel squares whose union is $P$ itself. Aupperle et. al~\cite{aupperle1988covering} provide an $\mathcal O(N^{1.5})$-time algorithm to solve {\sc OPCS} for orthogonal polygons without holes, where $N$ is the number of integral lattice points lying in the interior or on the boundary of $P$. In their paper, designing algorithms for {\sc OPCS} with a running time polynomial in $n$, the number of vertices of $P$, was stated as an open question; $N$ can be arbitrarily larger than $n$. Output sensitive algorithms were known due to Bar-Yehuda and Ben-Chanoch~\cite{bar1994}, but these fail to address the open question, as the output can be arbitrarily larger than $n$. We address this open question by designing a polynomial-time exact algorithm for {\sc OPCS} with a worst-case running time of $\mathcal O(n^{10})$. 

We also consider the following structural parameterized version of the problem. Let a knob be a polygon edge whose both endpoints are convex polygon vertices. Given an input orthogonal polygon without holes that has $n$ vertices and at most $k$ knobs, we design an algorithm for {\sc OPCS} with a worst-case running time $\mathcal O(n^2 + k^{10} \cdot n)$. This algorithm is more efficient than the former, whenever $k = o(n^{9/10})$. 

The problem of {\sc Orthogonal Polygon with Holes Covering with Squares (OPCSH)} is also studied by Aupperle et. al~\cite{aupperle1988covering}. Here, the input is an orthogonal polygon which could have holes and the objective is to find a minimum square covering it. They claim a proof that OPCSH is NP-complete even when all lattice points inside the polygon constitute the input. We think there is an error in their proof, where an incorrect reduction from \textsc{Planar 3-CNF} is shown. We provide a correct reduction with a novel construction of one of the gadgets, and show how this leads to a correct proof of NP-completeness of OPCSH.
\end{abstract}
\newpage
\section{Introduction}\label{sec:intro}

An \emph{orthogonal polygon} is a simple polygon such that every polygon-edge is parallel to either the x-axis or the y-axis (Figure~\ref{fig:orthogonal-polygon}). We consider the problem of covering a given orthogonal polygon with the minimum number of (possibly overlapping) axis-parallel squares. 

Formally, given an orthogonal polygon $P$, the problem is to find the minimum number of squares such that every square lies inside the region defined by $P$, and the union of the squares is the entire polygon $P$ itself (Figure~\ref{fig:example-covering} and Figure~\ref{fig:example-non-covering}). Equivalently, the problem is to find the minimum number of axis-parallel squares whose union is exactly $P$.

\begin{figure} [ht]
    \hfill 
    \begin{subfigure}[t]{0.26\textwidth} 
        \centering 
        \includegraphics[height=4cm]{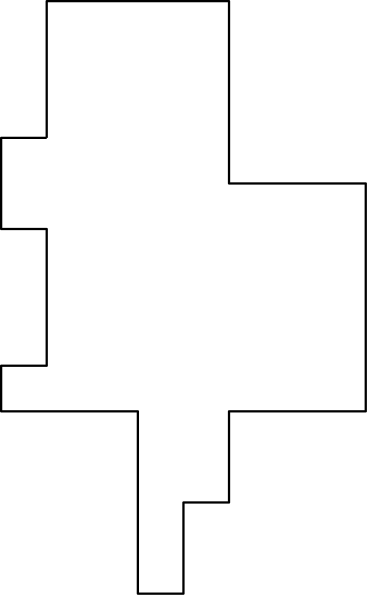}
        \caption{Orthogonal Polygon} \label{fig:orthogonal-polygon}
    \end{subfigure} \hfill
    \begin{subfigure}[t]{0.26\textwidth}
        \centering
        \includegraphics[height=4cm]{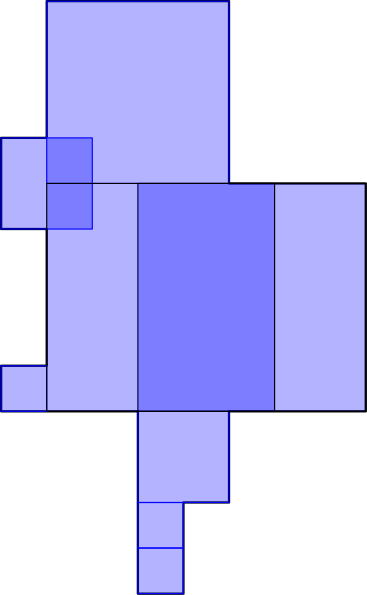}
        \caption{Valid Covering} \label{fig:example-covering}
    \end{subfigure} \hfill
    \begin{subfigure}[t]{0.38\textwidth}
        \centering
        \includegraphics[height=4cm]{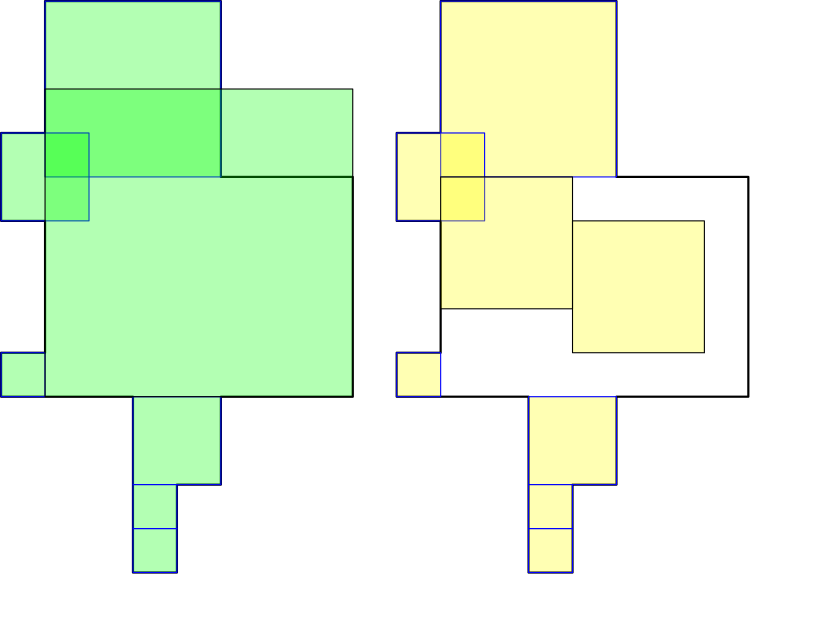}
        \caption{Invalid Coverings. Left: A chosen square not contained in polygon (largest square in the diagram). Right: squares not covering entire polygon} \label{fig:example-non-covering}
    \end{subfigure} \hfill \hfill \hfill
    \caption{Orthogonal Polygons and Covering with Squares}
\end{figure}

For most of the paper, we deal with orthogonal polygons without holes. We formally define the problem of \textsc{Orthogonal Polygon Covering with Squares (OPCS)}.

\defproblem{\textsc{Orthogonal Polygon Covering with Squares (OPCS)}}{An orthogonal polygon $P$ where its $n$ vertices have integral coordinates and $P$ does not have any holes}{Find the minimum number of axis-parallel squares contained in $P$, such that $P$ is entirely covered by these squares}

We also study a parameterized version of {\sc OPCS}, called $p$-{\sc OPCS}. Let a \emph{knob} be
an edge where both of its endpoints are convex vertices of $P$ (the formal definition of knobs is given in Section~\ref{sec:prelims}). The parameterized problem of $p$-OPCS takes the number of knobs in the input orthogonal polygon to be a structural parameter, to be utilized when designing parameterized algorithms. We define $p$-{\sc OPCS} as follows.

\defparproblem{$p$-\textsc{Orthogonal Polygon Covering with Squares ($p$-OPCS)}}{An orthogonal polygon $P$ where its $n$ vertices have integral coordinates and $P$ does not have any holes, a non-negative integer $k$ such that $P$ has at most $k$ knobs}{$k$}{Find the minimum number of axis-parallel squares contained in $P$, such that $P$ is entirely covered by these squares}

The final variant that we consider in this paper is the following variant of {\sc OPCS} where the input orthogonal polygon is allowed to have holes.

\defproblem{\textsc{Square Covering of Orthogonal Polygons with Holes (OPCSH)}}{An orthogonal polygon $P$ (possibly with holes) where all boundary vertices have integral coordinates}{Find the minimum number of axis-parallel squares contained in $P$, such that $P$ is entirely covered by these squares}

Note that for these problems, there are a couple of ways to represent the polygon $P$. One way is to just provide the coordinates of the $n$ vertices, in clockwise/counter-clockwise order. Another way is to list down all lattice points lying inside $P$. If there are $N$ lattice points lying inside $P$, $N$ can be even exponential in $n$.

\paragraph*{Previous results.}

The \OPCS problem originates from image processing and further finds its application in VLSI mask generation, incremental update of raster displays, and image compression \cite{DBLP:journals/algorithmica/Moitra91}. Moitra~\cite{DBLP:journals/algorithmica/Moitra91} considers the problem of a minimal covering of squares (i.e. no subset of the cover is also a cover) for covering a binary image $\sqrt N \times \sqrt N$ pixels; and presents a parallel algorithm running on EREW PRAM which takes time $T \in [\log N, N]$ with $(N/T)$ processors. 

In the works of Aupperle, Conn, Keil and O'Rourke \cite{aupperle1988covering}, exact polynomial time solutions of \OPCS are discussed, where the input is considered to be the set of all $N$ lattice points inside $P$. The algorithms crucially use the notion of associated graph $G(P)$.

\begin{definition}[Associated graph $G(P)$] \label{def:G}
    A unit square region inside an orthogonal polygon $P$ with corners on lattice points is called a \emph{block}. An \emph{associated graph} $G(P)$ is constructed with the set of nodes as the set of blocks inside $P$; two blocks $p_1$ and $p_2$ are connected by an undirected edge if there is an axis parallel square lying inside $P$ that simultaneously covers $p_1$ and $p_2$.
\end{definition}

By definition, any set of blocks covered by a single square shall correspond to a clique in the associated graph $G(P)$. Further, Albertson and O'Keefe \cite{doi:10.1137/0602026/albertson} show the converse as well. 

\begin{proposition} \label{prop:sq-clique}
    A subset of blocks inside an orthogonal polygon $P$ can be covered by a square if and only if they induce a clique in $G(P)$.
\end{proposition}

In the works of Aupperle et. al~\cite{aupperle1988covering}, $G(P)$ is shown to be a chordal graph, if $P$ does not have holes (please refer to~\cref{def:chordal} for the definition of a chordal graph). They further show that, using the polynomial time algorithm for minimum clique covering in chordal graphs (due to Gavril~\cite{Gavril1972AlgorithmsFM}), \OPCS can be solved in polynomial time with respect to the number of lattice points, $N$, inside $P$. An algorithm with running time $\mathcal O (N^{2.5})$ is presented and techniques are mentioned to speed it up to $\mathcal O(N^{1.5})$, using specific structural properties of the associated graph $G(P)$.

Aupperle et. al~\cite{aupperle1988covering} also claim a proof that \OPCSH is \textsc{NP}-hard if the input orthogonal polygon $P$ (possibly with holes) is described in terms of the $N$ lattice points inside $P$. The NP-complete problem \textsc{Planar 3-CNF}~\cite{doi:10.1137/0211025} is reduced to \OPCSH in order to prove its \textsc{NP}-hardness.

The works of Bar-Yehuda and Ben-Chanoch~\cite{bar1994} consider the problem of \OPCS where the input is the $n$ vertices of the polygon $P$. They provide an exact algorithm solving \OPCS in $\OO(n + \textsc{OPT})$ time and $\OO(n)$ space, where {\sc OPT} is the minimum number of squares to cover $P$. This is particularly interesting as this eliminates the dependence on $N$, which could be arbitrarily larger than $n$. They crucially use the concept of \emph{essential squares} and build up the solution one square at a time, to finally get an optimal set of squares. Since they considered arbitrary placement of vertices (not necessarily at lattice points), the concept of associated graphs was irrelevant to their approach. However, in the worst case {\sc OPT} can be arbitrarily larger than $n$; for example, a $1 \times t$ rectangle has $n = 4$ and $\text{OPT} = t$. Hence this algorithm, or any other output-sensitive algorithm, can never provide a worst-case running time guarantee of a polynomial in $n$. Later, in another paper, Bar-Yehuda~\cite{BarYehuda2014CoveringPW} describes an $\OO(n \log n + \textsc{OPT})$ $2$-factor approximation algorithm to solve this problem where $n$ is the number of vertices and \textsc{OPT} is the minimum number of squares to cover $P$. Incidentally the same paper states that no polynomial time exact algorithm with respect to $n$ and {\sc OPT} were known.

The works of Culberson and Reckhow~\cite{Culberson21976} prove that the problem of minimum covering of orthogonal polygons with rectangles is NP-hard. Kumar and Ramesh~\cite{10.1145/301250.301369Kumar} provide a polynomial time approximation algorithm with an approximation factor of $\OO(\sqrt{\log N})$ for the problem of covering orthogonal polygons (with or without holes) with rectangles. Most of the other geometric minimum cover problems for orthogonal polygons are NP-complete as well~\cite{o1987art}. Another closely related problem is the tiling problem, where the squares additionally need to be non-overlapping~\cite{aamand2023}. This is already non-trivial when the input is a rectangle~\cite{KENYON1996272, WALTERS20092913}.

\paragraph*{Our results and the organization of the paper.}

In the algorithms due to Aupperle et. al~\cite{aupperle1988covering}, the input size is considered to be the number of lattice points $N$ lying inside $P$. However, this is an inefficient way to represent orthogonal polygons. Rather, it would be more efficient to represent an orthogonal polygon $P$ as a sequence of its vertices (in either clockwise or counter-clockwise direction). This is because a polygon with $n$ vertices can have arbitrarily large number of lattice points $N$ inside it~\cite{bar1994}.

Even if we consider the total number of bits required to represent the vertices to be $n'$, then $N$ can be exponential in $n'$. As an example, consider $P$ to be a large square with vertices $(0, 0)$, $(0, 2^t)$, $(2^t, 0)$, $(2^t, 2^t)$. Even though it can be represented by just $\OO(t)$ bits, there are $\Theta(4^t)$ lattice points inside it. Moreover, if {\sc OPT} is the minimum number of squares required to completely cover $P$, then recall that {\sc OPT} can also potentially be exponential in the number of bits needed to represent the vertices (for example $P$ being a $1 \times 2^t$ rectangle).

In this paper we consider the input to be the $n$ vertices of the polygon $P$, and we design efficient algorithms for \OPCS, with respect to $n$. The algorithm by Aupperle et. al~\cite{aupperle1988covering}. becomes exponential now, as there can be exponentially many lattice points inside $P$. In their paper, designing algorithms which are polynomial in the number $n$ of polygon vertices was stated as an open question. Note that output sensitive algorithms~\cite{bar1994} are also exponential in the worst case.

To the best of our knowledge, our paper is the first to answer this open question. We present an algorithm with a running time of $\mathcal O(n^{10})$, where $n$ is the number of vertices  of the input polygon; along with related structural results, in \cref{sec:structure} and \cref{sec:polytime}. It is interesting to note that our algorithm can also work for polygons with rational coordinates, as we can simply scale up the polygon by an appropriate factor without changing the number of vertices $n$.

In \cref{sec:knobbed}, we consider the $p$-{\sc OPCS} problem and further optimize the above algorithm. We consider the polygon to have $n$ vertices and at most $k$ knobs (edges with both endpoints subtending $90^\circ$ to the interior of the polygon, formally defined in \cref{def:conv-conc-vert}). We design a recursive algorithm running in time $\mathcal O(n^2 + n \cdot k^{10})$. This algorithm is more efficient than the former whenever $k = o(n^{9/10})$. 

Moreover, we think that the claimed proof of NP-hardness of \OPCSH by Aupperle et. al~\cite{aupperle1988covering} is incorrect, as they incorrectly reduce from a polynomial-time solvable variant of a problem. In \cref{sec:hardness}, we provide a correct proof of NP-hardness of \OPCSH. Our proof relies on some of the existing constructions~\cite{aupperle1988covering}, but we also provide novel constructions for some gadgets used in the reduction. Note that such a result implies that we do not expect an algorithm for \OPCSH that has a running time polynomial in $n$ as well as $N$, as the reduction creates an instance with $N = \Theta(n)$.

All other notations and definitions used in the paper can be found in \cref{sec:prelims}.

\section{Preliminaries}\label{sec:prelims}

\subparagraph{Orthogonal polygons.} We denote by $P$ our input polygon with $n$ vertices. For distinction, we associate the terms `vertices' and `edges' to the polygon, the terms `corners' and `sides' to other geometric objects, and the terms `nodes' and `arcs' to graphs. In the entire paper, we always consider the vertices of the polygon to have integral coordinates. Unless mentioned otherwise, the polygon $P$ is assumed to not have any holes. We denote the number of lattice points inside $P$ by $N$. We assume that all arithmetic operations take constant time.

Let the vertices of $P$ in order be $v_1, v_2, \ldots, v_n$, where $v_i = (x_i, y_i)$ with $x_i$ and $y_i$ being integer coordinates. For convenience, we define $v_0 := v_n$ and $v_{n+1} := v_1$. Therefore $(v_i, v_{i+1})$ is a polygon edge of $P$ for all $i \in \{0, 1, 2, \ldots, n\}$. For any natural number $d$, we denote the set $\{1, 2, \ldots, d\}$ by the notation $[d]$. Let the minimum number of squares required to cover $P$ be denoted by $\textsc{\OPCS}(P)$. For any two points $a$, $b$ on the plane, we denote the (Euclidean) length of the line segment joining $a$ and $b$ as $\overline{ab}$.

We use the terms \textit{left}, \textit{right}, \textit{top}, \textit{bottom}, \textit{vertical}, \textit{horizontal} to  indicate \textit{greater x-coordinate}, \textit{smaller x-coordinate}, \textit{greater y-coordinate}, \textit{smaller y-coordinate}, \textit{parallel to y-axis}, and \textit{parallel to x-axis}, respectively. By \textit{boundary} of a square (rectangle), we mean its four sides. A block lying on the boundary of a square (rectangle) is one which lies inside the square (rectangle) and where the boundary of the block overlaps with the boundary of the square. 

\begin{proposition}\label{prop:invariants}
   The minimum number of squares to cover an orthogonal polygon $P$, remains the same when $P$ is scaled up by an integral factor and/or rotated $90^\circ$ (clockwise/counter-clockwise).
\end{proposition}

We formally define convex and concave vertices of an orthogonal polygon.
\begin{definition}[Concave and convex vertices]\label{def:conv-conc-vert}
A vertex $v_i$ of an orthogonal polygon $P$ is a convex vertex if $v_i$ subtends an angle of $90^\circ$ to the interior of $P$. A vertex $v_j$ of $P$ is a concave vertex if it subtends an angle of $270^\circ$ to the interior of $P$(Figure~\ref{fig:con-vert}). 
\end{definition}

\begin{figure} [!htbp]
\centering    
\includegraphics[height=30mm]{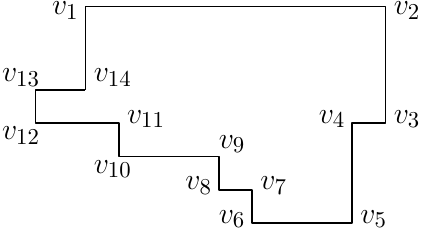}
\caption{Convex vertices: $\{v_1, v_2, v_3, v_5, v_6, v_8, v_{10}, v_{12}, v_{13}\}$, concave vertices: $\{v_4, v_7, v_9, v_{11}, v_{14}\}$}\label{fig:con-vert}
\end{figure}

It is interesting to look at maximal squares that are completely contained in the region of the input polygon $P$.
\begin{definition}[Valid square]
    An axis parallel square $S$ is said to be valid if $S$ is fully contained in the region of $P$.
\end{definition}

\begin{definition}[Maximal Square]\label{def:maximal}
    A valid square $S$ is a maximal square, if no other valid square with a larger area contains $S$ entirely (Figure~\ref{fig:max-sq}).
\end{definition}

\begin{figure} [ht]
    \hfill 
    \begin{subfigure}[t]{0.4\textwidth} 
        \centering 
        \includegraphics[width=30mm]{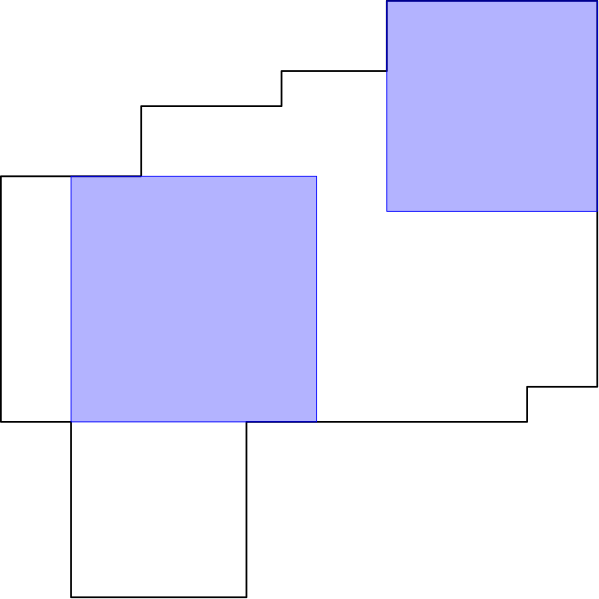}
        \caption{Some maximal squares}
    \end{subfigure} \hfill
    \begin{subfigure}[t]{0.4\textwidth}
        \centering
        \includegraphics[width=30mm]{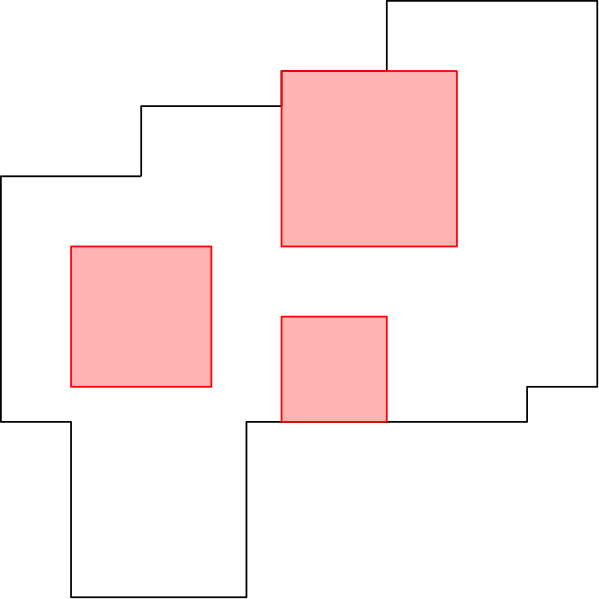}
        \caption{Valid squares which are not maximal}
    \end{subfigure} \hfill \hfill \hfill
    \caption{Orthogonal Polygons and Covering with Squares} \label{fig:max-sq}
\end{figure}

Note that for a minimum square covering of an orthogonal polygon $P$, some squares may be maximal squares while the other squares may be valid squares which are not maximal. However, for each of such valid square $S$ which is not maximal, we can replace it with any maximal square containing $S$ to obtain another minimum square covering of $P$ containing only maximal squares.

\begin{observation} \label{prop:max-sq-cov}
    There exists a minimum square covering of any orthogonal polygon $P$ where every square of the covering is a maximal square. 
\end{observation}

We restate the definition of knobs as in previous works \cite{bar1994, aupperle1988covering}.

\begin{definition}[Knob]\label{def:knobs}
    A knob is an edge $(v_i, v_{i+1})$ of the orthogonal polygon $P$ such that $v_i$ and $v_{i+1}$ are both convex vertices of $P$ (Figure~\ref{fig:knobs}). 
    
    Suppose $v_i$ and $v_{i+1}$ share the same x coordinate. Then the edge $(v_i, v_{i+1})$ is a left knob if it is a left boundary edge (\textit{i.e.} all points just to the left of the edge are outside $P$), otherwise they form a right knob. 
    
    Similarly, when $v_i$ and $v_{i+1}$ share the same y coordinate, then the edge $(v_i, v_{i+1})$ is a top knob if it is a top boundary edge (\textit{i.e.} all points just to the top of the edge are outside $P$), otherwise they form a bottom knob. 
\end{definition}

Next, we define a non-knob convex vertex.
\begin{definition}[Non-knob convex vertex] \label{def:non-knob-vert}
    A convex vertex $v_i$ is said to be a non-knob convex vertex if neither $(v_i, v_{i + 1})$ nor $(v_{i - 1}, v_i)$ is a knob. Equivalently, both $v_{i-1}$ and $v_{i + 1}$ are concave vertices.
\end{definition}

We introduce a special structure called \emph{strips} as follows.

\begin{definition}[Strip] \label{def:strip}
    A strip of an orthogonal polygon $P$ is a maximal axis-parallel non-square rectangular region $Y$ lying inside $P$, such that each of the longer parallel side of $Y$ is completely contained in a polygon edge of $P$ (Figure~\ref{fig:strip}).  
\end{definition}

\begin{figure}[ht]  
    \hfill    
    \begin{subfigure}[t]{0.45\textwidth}
        \includegraphics[width=\textwidth]{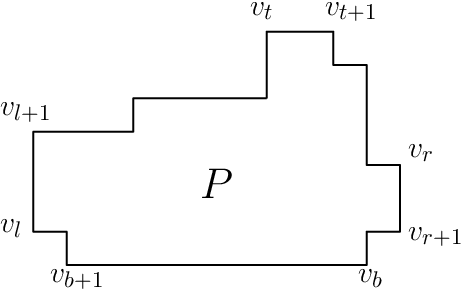}
        \caption{$(v_t, v_{t + 1})$ is a top knob, $(v_l, v_{l + 1})$ is a left knob, $(v_r, v_{r + 1})$ is a right knob, $(v_b, v_{b + 1})$ is a bottom knob}\label{fig:knobs} 
    \end{subfigure} \hfill
    \begin{subfigure}[t]{0.45\textwidth}
        \includegraphics[width=\textwidth]{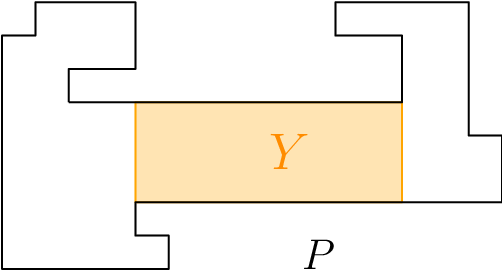}
        \caption{Strip $Y$ of the polygon $P$}\label{fig:strip} 
    \end{subfigure} \hfill \hfill
    \caption{Knobs and Strips}
\end{figure}

We make the following observation

\begin{observation}\label{obs:one-strip-per-pair}
    For any pair of parallel polygon edges $e_1$, $e_2$ of $P$, there can be at most one strip with its sides overlapping with $e_1$ and $e_2$. 
\end{observation} 

\subparagraph*{Efficient representation of squares.} We introduce the following notations to represent a sequence of squares placed side by side. We start by defining a rec-pack, which is a $t \times \eta t$ or $\eta t \times t$ rectangle lying inside the given orthogonal polygon $P$, where $\eta \in \mathbb N$. 

\begin{definition}[Rec-pack] \label{def:rec-pack}
    A rec-pack of $P$  is an axis parallel rectangle $R$ of dimension $t \times \eta t$ or $\eta t \times t$ that lies completely inside $P$ (Figure~\ref{fig:rec-pack}). We define the length $t$ of the shorter side of the rec-pack $R$ to be the \emph{width} of $R$ and the aspect ratio $\eta$ to be the \emph{strength} of $R$.
\end{definition}

\begin{remark}
    A valid square is a rec-pack of strength $1$. Any square is a \emph{trivial rec-pack}.
\end{remark}

For a rec-pack $R$ with width $t$ and strength $\eta$, it can be covered with exactly $\eta$ many $t \times t$ valid squares. We formally define this operation as an \emph{extraction}.

\begin{definition}[Extraction of rec-pack]\label{def:extraction}
    Given an orthogonal polygon $P$ and a rec-pack $R$ of width $t$ and strength $\eta$, we define the extraction of $R$, $\Ex(R)$ to be the set of $\eta$ valid squares of size $t \times t$ placed side by side to cover the entire region of $R$ (Figure~\ref{fig:extraction}). 
    
    For a set $\mathscr R$ of rec-packs, we define its extraction $\Ex(\mathscr R)$ to be the union of the extraction of each rec-pack in $\mathscr R$. 

    $$\Ex(\mathscr R) = \bigcup \limits_{R \in \mathscr R} \Ex(R)$$
\end{definition}

\begin{figure}[ht]  
    \hfill    
    \begin{subfigure}[t]{0.40\textwidth}
        \includegraphics[width=\textwidth]{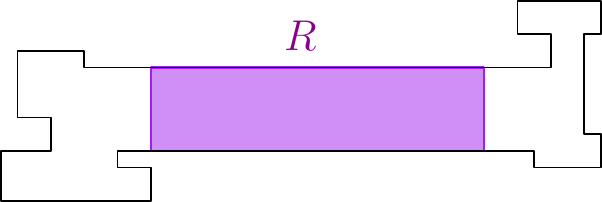}
        \caption{Rec-pack $R$ with strength $4$}\label{fig:rec-pack} 
    \end{subfigure} \hfill
    \begin{subfigure}[t]{0.40\textwidth}
        \includegraphics[width=\textwidth]{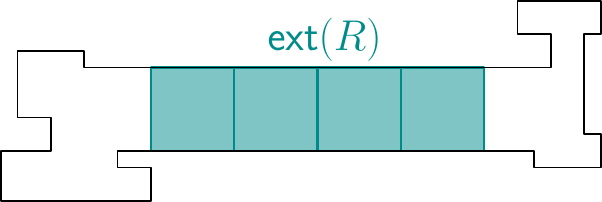}
        \caption{Extraction $\Ex(R)$ of rec-pack $R$}\label{fig:extraction} 
    \end{subfigure} \hfill \hfill
    \caption{Rec-packs and Extractions}
\end{figure}

\begin{remark}
    For a square $S$ (a trivial rec-pack), the extraction is a singleton set containing $S$ itself, \textit{i.e.} $\Ex(S) = \{S\}$. Moreover, for any set $\mathscr S$ of squares, we have $\Ex(\mathscr S) = \mathscr S$.
\end{remark}

This structure helps us to efficiently represent a solution of a square covering. If the set of squares constituting the solution has a subset of 
$\eta$ many side-by-side placed $t \times t$ valid squares, they can be simply replaced by the rec-pack defined by their union (of strength $\eta$ and width $t$). 

To exploit this, we need to formally define a covering using rec-packs.

\begin{definition}
    A set of rec-packs $\mathscr R$ is said to cover the input orthogonal polygon $P$ if the set of squares $\Ex(\mathscr R)$ cover the polygon $P$. Moreover, $\mathscr R$ is said to be a minimum covering, if $\Ex(\mathscr R)$ is a minimum covering of $P$, and no two rec-packs produce the same sqaure on extraction, \textit{i.e.} for all $R_1, R_2 \in \mathscr R$, $\Ex(R_1) \cap \Ex(R_2) = \emptyset$.
\end{definition}

\begin{remark}
    For any set $\mathscr S$ of squares covering $P$, $\mathscr S$ is also a set of rec-packs covering $P$. Similarly, for any minimum covering of squares $\mathscr S$ of $P$, $\mathscr S$ can also be viewed as a set of rec-packs which is a minimum covering of $P$. Similarly, if a set of rec-packs $\mathscr R$ is a covering (minimum covering) of $P$, then the set of squares $\Ex(\mathscr R)$ is a set of trivial rec-packs which is also a covering (minimum covering) of $P$.
\end{remark}




\subparagraph*{Graph theory.} Given a graph $G$, we denote the node set by $V(G)$ and the arc set by $E(G)$. The neighbourhood of a node $v\in V(G)$ is denoted by $N(v)$. The closed neighbourhood of $v$ is denoted by $N[v]$. Given a graph $G$, a subgraph $H$ is an induced subgraph of $G$ if $V(H) \subseteq V(G)$ and $E(H)$ contains all arcs of $E(G)$ such that both endpoints are in $V(H)$. In this case, we also say that the node set $V(H)$ induces the subgraph $H$. Recall the definition of chordal graphs and simplicial nodes \cite{10.1007/978-1-4613-8369-7_1/chordal}.

\begin{definition}[Chordal graphs and simplicial nodes] \label{def:chordal}
     A chordal graph $G$ is a simple graph in which every cycle of length at least four has a chord. A node $v \in V(G)$ is simplicial if $N[v]$ induces a complete graph.
\end{definition}

We get the following result directly from the definition.

\begin{proposition}\label{prop:sub-chordal}
    Any induced subgraph of a chordal graph is also chordal. 
\end{proposition}

We now state \emph{Dirac's Lemma} on chordal graphs \cite{10.1007/978-1-4613-8369-7_1/chordal}.

\begin{proposition}[Dirac's lemma] \label{prop:dirac}
    Every chordal graph has a simplicial node. Morever, if a chordal graph is not a complete graph, it has at least two non-adjacent simplicial nodes.
\end{proposition}

We finally reiterate one of the major results of the works of Aupperle et. al \cite{aupperle1988covering}.

\begin{proposition}\label{prop:G-chordal}
    For a simple orthogonal polygon $P$ (without holes), the associated graph $G(P)$ is chordal. 
\end{proposition}
\section{Overview of the Algorithms and the Reduction}\label{sec:overview}

\subsection{Polynomial-time algorithm for \OPCS}

The algorithm by Bar-Yehuda and Ben-Chanoch~\cite{bar1994} starts with an empty set and keeps adding \emph{essential} or \emph{unambiguous} squares one by one, while maintaining the invariant that there is always a minimum covering containing the current set of squares. This algorithm runs in time $\OO(n + \text{OPT})$, OPT being the minimum number of squares needed to cover the polygon. Recall that OPT may not be bounded by a polynomial in $n$. Our approach will be the same, except along with unambiguous squares, we will also use rec-packs while building our solution. This is because, our aim is to design an algorithm that runs in time polynomial in the number of vertices $n$ of the polygon.

The basic framework will be to keep adding squares/rec-packs to the existing set of rec-packs, that preserve the invariant: the currently constructed set of rec-packs can always be extended to a minimum covering of $P$. We stop as soon as the set of rec-packs cover the entire polygon $P$; by the invariance property, this must be a minimum covering.

The primary challenge for designing an algorithm that runs in time polynomial in the number of vertices, lies in finding unambiguous squares if they exist, and even figuring out what happens if there are no unambiguous squares. Firstly, we prove that there will always be at least one unambiguous square. This comes from the fact that there will always be a simplicial node in any induced subgraph of the chordal graph $G(P)$. However, finding an unambiguous square is still not trivial to achieve in time polynomial in $n$.

We break down the task of finding unambiguous squares into two subtasks: (i) obtaining a polynomial-sized set of squares which contain at least one unambiguous square, and (ii) detecting which of them is unambiguous.

Using this framework directly would still give us an output-sensitive algorithm, since we are enumerating each square one by one. However, this can be sped up by including rec-packs to the partial solution. Suppose at an instant, the algorithm finds an unambiguous square $S$ and adds it to the current set of rec-packs. At such a stage, let $\mathscr R$ be the set of rec-packs currently constructed by our algorithm, which can be extended to a minimum covering $\mathscr R^\star$ of $P$. We ensure that the algorithm further detects if $S$ is in some non-trivial rec-pack $R$ in $\mathscr R^\star$. In such a case, we replace the square $S$ by the rec-pack $R$ in our current set of rec-packs. Since $R$ may contain multiple new squares, it is possible that this step includes multiple squares to our solution in a single iterations, and this might result in a reduction in the number of iterations. In fact, we show that $\OO(n^2)$ iterations are enough, making the running time guarantee of the algorithm to be polynomial in $n$, and independent of the output. The exact time complexity turns out to be $\OO(n^{10})$; details of this are discussed in \cref{sec:polytime}.

\subsection{Recursive algorithm using separating squares}

Next, we consider the problem $p$-\OPCS, where the number of knobs is at most $k$. We design a recursive algorithm for $p$-\OPCS which is faster than our previous algorithm if $k$ is small enough.

For our recursive algorithm, we crucially use the structure of separating squares: maximal squares that have a corner at a non-knob convex vertex of $P$, and deletion of which separates the polygon. Let $P$ have $l$ non-knob convex vertices. We find a separating square $S$, the deletion of which separates the polygon into smaller polygons $Q'_1, Q'_2, \ldots, Q'_{t'}$, $t' \le n$. We now appropriately group $Q'_1, \ldots Q'_{t'}$ to get $Q_1, \ldots, Q_t$, $t \le 12$ such that each of the polygons $Q_1 \cup S, Q_2 \cup S, \ldots, Q_t \cup S$ have at most $k$ knobs, and at most $l - 1$ non-knob convex vertices. This allows us to recurse on the polygons $Q_1 \cup S, Q_2 \cup S, \ldots, Q_t \cup S$. The base case of the recursion is when there are no non-knob convex vertices, and this is solved by the algorithm of \cref{sec:polytime}.

The running time analysis of this recursive framework highly relies on the following results:
\begin{itemize}
    \item There are $\OO(n)$ recursive steps and $\OO(n)$ base cases
    \item Each recursive step requires $\OO(n)$ time
    \item The number of vertices for a polygon that appears as a base case is $\OO(k)$
\end{itemize}
This framework therefore gives us a running time guarantee of $\OO(n^2 + n \cdot k^{10})$. This is discussed in detail, in \cref{sec:knobbed}.

\subsection{Reducing {\sc OPCSH} from {\sc Planar 3-CNF}}

The proof of NP-hardness of {\sc OPCSH} by Aupperle et. al~\cite{aupperle1988covering} incorrectly reduces from the problem of the existence for a tautology of a {\sc Planar 3-CNF} instance (linear-time solvable), instead of the problem of satisfiability of a {\sc Planar 3-CNF} instance (NP-complete~\cite{doi:10.1137/0211025}). This is because the gadget used for clauses can be covered by 12 squares if all literals appear as {\sf true}; otherwise 13 squares are required. Intuitively, when the literals of some clause are such that some evaluate to {\sf true} and some evaluate to {\sf false}, the number of squares required to cover the gadget must be less than when all literals are {\sf false}. This is because the former setting makes the clause evaluate to {\sf true}, but the latter makes the clause evaluate to {\sf false}. 

We reduce {\sc OPCSH} from satisfiability of {\sc Planar 3-CNF}. In our construction, just like in~\cite{aupperle1988covering} we have variable gadgets, clause gadgets and connector gadgets. Our variable and connector gadgets are exactly the same as those in~\cite{aupperle1988covering}. However, we introduce a novel construction for a clause gadget. Our construction requires 29 squares to cover a clause gadget, if all literals appear as {\sf false}; otherwise 28 squares are required. This completes the proof of NP-hardness of \OPCSH.  This is discussed in detail, in \cref{sec:hardness}.
\section{Structural and Geometric Results}\label{sec:structure}

We prove some structural and geometric results of a minimum square covering of an orthogonal polygon $P$, and those of a set of rec-packs forming a minimum covering of $P$.

\subsection{Structure of minimum coverings}

Due to \cref{prop:max-sq-cov} we direct our focus to minimum coverings where every square is a maximal square. First, we define how a maximal square can be obtained from a convex vertex of the given orthogonal polygon $P$.

\begin{lemma} \label{lem:unique-mcs}
    The maximal square of an orthogonal polygon $P$, covering a convex vertex $v$ is unique and can be found in $\mathcal O(n)$.
\end{lemma}

\begin{proof}
    
    Let $\mathcal R$ denote the region defined by the interior of the $90^\circ$ angle between the rays formed by extending the two edges adjacent to $v$ (\textit{i.e.} a quadrant). The largest valid square $S$ in this quadrant with one corner at $v$ is the unique maximal square covering $v$. 
    
    To algorithmically find $S$ in $\mathcal O(n)$ time, we iterate over all polygon edges $(v_i, v_{i + 1})$, and check the following:

    \begin{itemize}
        \item if $(v_i, v_{i + 1})$ lies entirely outside $\mathcal R$, it can never constrain the maximal square covering $v$, so we continue to the next iteration
        \item if some portion of $(v_i, v_{i + 1})$ lies inside $\mathcal R$, the largest square in $\mathcal R$, with one vertex at $v$ and the strict interior of the square does not overlap with $(v_i, v_{i + 1})$. This can be done in $\mathcal O(1)$, just by comparing the coordinates of $v_i$, $v_{i+1}$ and $v$.
    \end{itemize}

    Finally, the square with the minimum area is the unique maximal square covering $v$, as all other edges allow larger (or equally large) squares.
\end{proof}

We define a maximal corner square of a vertex as follows.

\begin{definition}[Maximal Corner Square of a vertex]
    The Maximal Corner Square of a convex vertex $v$, or $MCS(v)$ is the unique maximal square covering $v$.    
\end{definition}

\begin{remark}\label{rem:mcs-always}
    Due to Proposition~\ref{prop:max-sq-cov} and Lemma~\ref{lem:unique-mcs}, there is a minimum square covering of an orthogonal polygon $P$ such that for each of the convex vertices $v$, $MCS(v)$ is one of the squares in the minimum covering. 
\end{remark}

Note that any maximal square should be bounded by either two vertical polygon edges or two horizontal polygon edges. 

\begin{lemma}\label{lem:2-side-poly}
    If $S$ is a maximal square of an orthogonal polygon $P$, then either both the vertical sides of $S$ overlap with some polygon edges in $P$ or both the horizontal sides of $S$ overlap with some polygon edges of $P$.
\end{lemma}
\begin{proof}
    Assume the contrary. Suppose for some maximal square $S$ there is one horizontal side and one vertical side which does not overlap with any polygon edges of $P$. Further, we can assume that these are the top and the right edges of $S$ (\cref{prop:invariants}). Then we can further grow $S$ fixing its bottom-left corner, which contradicts that $S$ is a maximal square.
\end{proof}

\subsection{Simplicial nodes in the associated Graph and partial solutions}


We define a partial solution for {\sc OPCS} to be a subset of a set of rec-packs which is a minimum covering of $P$.


\begin{definition}[Partial solution]\label{def:part-sol}
    Given an orthogonal polygon $P$, a set of rec-packs $\mathscr R$ is a \emph{partial solution} for {\sc OPCS} if there is a minimum covering set of rec-packs $\mathscr R'$ with $\mathscr R \subseteq \mathscr R'$. 
\end{definition}

Now consider a partial solution $\mathscr R$. We define $G^{\mathscr R}(P)$ as follows.

\begin{definition}
    Let $G^{\mathscr R}(P)$ denote the induced subgraph of the associated graph $G(P)$ consisting of nodes corresponding to blocks in $P$ which are not covered by rec-packs in $\mathscr R$. 
\end{definition}

\cref{prop:G-chordal} and \cref{prop:sub-chordal} imply that $\GS$ is chordal, and \cref{prop:dirac} implies that $\GS$ has a simplicial node $p$ if $\GS$ is non-empty. Note that nodes in $\GS$ or $G(P)$ are blocks in $P$. Let $A$ denote the union of the block $p$ and its neighbouring blocks in $\GS$. $A$ induces a clique in $\GS$ and hence in $G(P)$ (\cref{def:chordal}). Therefore, there exists a maximal square $S_A$ that covers all blocks in $A$ (\cref{prop:sq-clique}). With this, we define \emph{unambiguous squares given a partial solution $\mathscr R$}.

\begin{definition}[Unambiguous squares given a partial solution]\label{def:unambiguous}
    Let $\mathscr R$ be a partial solution for an orthogonal polygon $P$. We call a maximal square $S$ to be an unambiguous square given $\mathscr R$, if there is simplicial node $p$ in $\GS$ such that $S$ covers $p$ and all its neighbours in $\GS$. 
\end{definition}

\begin{remark}
    For a partial solution $\mathscr R$ which does not completely cover $P$, a simplicial vertex always exists in $\GS$; and hence an unambiguous square given $\mathscr R$ always exists.
\end{remark} 

It is interesting to note that for a given simplicial node $p$ in $\GS$, there can be multiple maximal squares that cover $p$ and all its neighbours in $\GS$; they cover up the exact same set of uncovered blocks, but they may overlap differently with rec-packs in $\mathscr R$. \emph{Umabiguous} squares are essentially equivalent to \emph{essential} squares defined in the works of Bar-Yehuda and Ben-Chanoch~\cite{bar1994}. 

\begin{lemma}\label{lem:unambiguous}
    If $\mathscr R$ is a partial solution for an orthogonal polygon $P$ and $S$ is an unambiguous square given $\mathscr R$, then $\mathscr R \cup \{S\}$ is also a partial solution.
\end{lemma}

\begin{proof}
    Let $p$ be a simplicial node in $\GS$ such that $S$ covers $p$ and its neighbours in $\GS$. Since $\mathscr R$ is a partial solution, there exists a set of rec-packs $\mathscr R'$ which is a minimum cover for $P$, satisfying $\mathscr R \subseteq \mathscr R'$. Let $S'$ be some square in $\Ex(R')$ that covers $p$. 

    Note that $S$ covers all such blocks in $\GS$ that $S'$ covers. $(\Ex(\mathscr R') \setminus \{S'\}) \cup \{S\}$ therefore also forms a minimum cover of $P$. This implies that $\Ex(\mathscr R) \cup \{S\}$ as well as $\mathscr R \cup \{S\}$ are partial solutions.
\end{proof}


Although we know that an unambiguous square given a partial solution $\mathscr R$ always exist, it is not trivial to algorithmically find one. The first step with which we achieve this is to get a polynomial-sized set of squares, which has at least one unambiguous square. The next result formally discusses this. 

\begin{lemma}\label{lem:grid-unambiguous}
    Let $\mathscr R$ be a partial solution of an orthogonal polygon $P$ not completely covering $P$. We define $C_x$ and $C_y$ as follows.
    \begin{align*}
        C_x :=  \{&(x,y) \in \mathbb Z^2 \mid x\text{ is an x-coordinate of a vertex of } P, \text{ and } y\text{ is a }\\
        & \text{y-coordinate of a vertex of }P\text{ or a corner of a rec-pack in }\mathscr R\}\\
        C_y :=  \{&(x,y) \in \mathbb Z^2 \mid x\text{ is an x-coordinate of a vertex of }P\text{ or a corner }\\
        & \text{of a rec-pack in }\mathscr R, \text{ and } y\text{ is a y-coordinate of a vertex of }P\}
    \end{align*}
    There exists an unambiguous square $S$ given $\mathscr R$ which has one corner in $C_x \cup C_y$.
\end{lemma} 

\begin{proof}
    Notice that \cref{prop:dirac} guarantees the existance of at least one simplicial node in $\GS$. In fact, there can be multiple simplicial nodes. Let $V'$ be the simplicial nodes with the largest number of neighbours in $\GS$ (highest degree). Among them, consider the topmost simplicial nodes in $V'$, and let $p$ be the leftmost among them. Therefore, $p$ is the simplicial node in $\GS$ with the largest number of neighbours, and among the topmost of such simplicial nodes it is the leftmost.
    
    Let $A$ denote the closed neighbourhood of $p$ in $\GS$. By \cref{def:chordal}, $A$ must induce a clique in $\GS$, and hence in $G(P)$. Therefore, there must be a square that covers $A$ completely (\cref{prop:sq-clique}). Let $S_0$ be a maximal square that covers $A$ completely, such that the position of the top-left corner of $S_0$ is leftmost among the topmost of all such squares covering $A$. 
    
    $S_0$ is maximal, so either the horizontal sides of $S_0$ individually overlap with two horizontal polygon edges, or the vertical sides of $S_0$  individually overlap with two vertical polygon edges (\cref{lem:2-side-poly}). We will assume that the horizontal sides of $S_0$ overlap with two horizontal polygon edges; the other case permits a symmetric argument. 
    
    We consider the following two cases. 

    \subparagraph*{Case-I: $S_0$ is a $1 \times 1$ square.} We first show that, in such a case, $\GS$ is a collection of isolated vertices.
    
    \begin{claim}
    If $S_0$ is a $1 \times 1$ square, then $\GS$ is a collection of isolated vertices.
    \end{claim} 
    \begin{proof}[Proof of claim]
        For the sake of contradiction, assume that $\GS$ has some arc. Consider the largest connected $G^\star$ of $\GS$. $G^\star$ induces a chordal graph, and hence must have a simplicial node, that has degree at least $1$. So the highest degree simplicial node in $\GS$ cannot be of degree 1 --- contradiction.
    \end{proof}
    
    Now, $p$ is the leftmost among the topmost of such simplicial nodes. Therefore, the immediate left block $p'$ of $p$, is either (i) covered by a rec-pack in $\mathscr R$ or (ii) is outside the polygon. 
    
    If $p'$ is covered by a rec-pack in $\mathscr R$, but $p$ is not, then the left side of $S_0$ is lies on the right side of some rec-pack $R$ in $\mathscr R$. Otherwise, if $p'$ is outside the polygon, then the left side of $S$ lies on some polygon edge. Therefore, the top-left corner of $S_0$ must share the x-coordinate of either a corner of a rec-pack in $\mathscr R$ or a polygon vertex. Further, the top-left corner of $S_0$ shares a y-coordinate of a polygon vertex as the horizontal sides of $S_0$ overlap with polygon edges. This means that the top-left vertex of $S_0$ is in $C_y$.

    \subparagraph*{Case-II: $S_0$ is not a $1 \times 1$ square.} Recall that we assumed the horizontal sides of $S_0$ to overlap with polygon edges. If the left (resp. right) side of $S_0$ shares x-coordinate of some polygon edge or a corner of some rec-pack in $\mathscr R$, then the top-left (resp. top-right) corner of $S_0$ is in $C_y$; then we are done. 

    So, It suffices to assume that neither the left side nor the right side of $S_0$ share the x-coordinate with a polygon edge or with a corner of some rec-pack in $\mathscr R$. Let $S_l$ and $S_r$ be squares congruent to $S_0$, that are shifted one unit to the left and right respectively (\cref{fig:Sl-Sr}). 
    
    \begin{claim}
        $S_l$ and $S_r$ are maximal valid squares.
    \end{claim}
    \begin{proof}
        $S_l$ and $S_r$ are valid as the left and right sides of $S_0$ do not overlap with any polygon edge. Moreover, the horizontal sides of $S_l$ and $S_r$ overlap with polygon edges, the same edges with which the horizontal sides of $S_0$ overlap. So, $S_l$ and $S_r$ are maximal as well. 
    \end{proof}
    
    We further consider three more subcases. 
    
    \begin{itemize}
        \item \textbf{Case II(a): All blocks inside $S_0$ along its right boundary, are also covered by some rec-pack in $\mathscr R$.} Therefore, $S_l$ covers the same set of uncovered blocks as $S_0$. Moreover, the top-left corner of $S_l$ lies to the left of of the top-left corner of $S_0$. This contradicts the definition of $S_0$, and hence this subcase never arises.
        \item \textbf{Case II(b): All blocks inside $S_0$ along its left boundary, are also covered by some rec-pack in $\mathscr R$.} Consider the following process: we keep moving $S_0$ horizontally to the right until one of the following happens:
        \begin{enumerate}
            \item it is obstructed by a polygon edge to its right.
            \item moving it any further uncovers some previously covered block.
            \item moving it any further will make one of its horizontal sides lose overlap with a polygon edge.
        \end{enumerate}
        Let the final square obtained be $S'$. Due to the second condition, $S'$ and $S_0$ cover the same set of uncovered blocks. Notice that since the process terminated with $S'$, one of the three above conditions must hold for $S'$. If further right movement is restricted due to a polygon edge, the right side of $S'$ overlaps with a polygon edge. If the right movement is restricted as a covered block becomes uncovered, then the left side of $S'$ overlaps with the right-side of some rec-pack of $\mathscr R$. Finally, if moving it any further makes $S'$ lose overlap with a horizontal polygon edge, then there is a vertex of $P$, that coincides with the top-left corner or the bottom-left corner of $S'$. 
        
        Therefore, in all cases, one of the vertical sides share an x-coordinate of either a polygon vertex, or a rec-pack in $\mathscr R$. Since the horizontal sides of $S_0$ overlap with polygon edges, there is a corner of the $S'$ in $\mathcal C_y$.

        \begin{figure}[ht]  
            \hfill    
            \begin{subfigure}[t]{0.40\textwidth}
                \includegraphics[width=\textwidth]{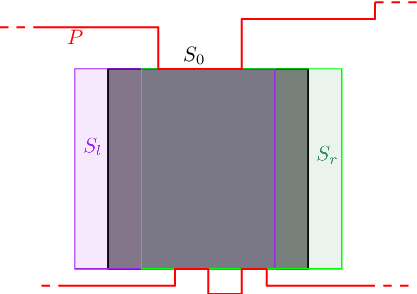}
                \caption{Construction of $S_l, S_r$ given $S, P$}\label{fig:Sl-Sr} 
            \end{subfigure} \hfill
            \begin{subfigure}[t]{0.40\textwidth}
                \includegraphics[width=\textwidth]{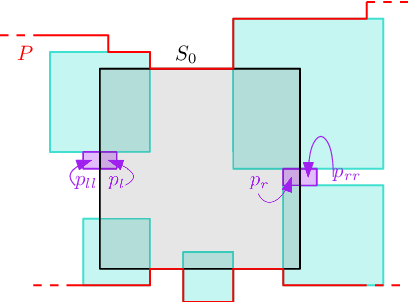}
                \caption{Blocks $p_l, p_{ll}, p_r, p_{rr}$ given $S_0$ and rec-packs in $\mathscr R$}\label{fig:pl-pr} 
            \end{subfigure} \hfill \hfill
            \caption{Cases in proof of \cref{lem:grid-unambiguous}}
        \end{figure}
        
        \item \textbf{Case II(c): There is a block $p_l$ inside $S_0$ along its left boundary, and a block $p_r$ inside $S_0$ along its right that are both not covered by any rec-pack in $\mathscr R$ (Figure~\ref{fig:pl-pr}).} Let $p_{ll}$ be the block immediately to the left of $p_{l}$ and $p_{rr}$ be the block immediately to the right of $p_{r}$. As the vertical sides of $S_0$ are assumed to not overlap with a vertical side of a rec-pack or polygon edges, the blocks $p_{ll}$ and $p_{rr}$ must be uncovered blocks lying inside $P$. $S_l$ covers $p_{ll}$ and $S_r$ covers $p_{rr}$. However if $p$ is not on the left boundary of $S_0$, then $p$ is covered by $S_r$ (which also covers $p_{rr}$). This implies that $p$ and $p_{rr}$ are adjacent in $\GS$. Similarly, if $p$ is on the left boundary of $S_0$, then $p$ is covered by $S_l$ (which also covers $p_{ll}$). Thus, $p$ and $p_{ll}$ are adjacent in $\GS$. Therefore in either case there is some node $p'$ which is not covered by $S_0$ but adjacent to $p$ in $\GS$. This contradicts the definition of $S_0$ and therefore such a case does not arise.
    \end{itemize} 
\end{proof}

\cref{lem:grid-unambiguous} gives us a set of squares which contains at least one unambiguous square. We now require a method to detect if a given square is unambiguous or not. The following result takes us in this direction. 

\begin{lemma}\label{lem:1-pad}
    Let $\mathscr R$ be a partial solution of an orthogonal polygon $P$, which does not completely cover $P$. Let $S$ be any valid  square. We define $D_x$ and $D_y$ as a set of lattice points as follows.
    \begin{align*}
        D_x &:=  \{(x, y) \in \mathbb Z^2 \mid \exists x'\text{ which is an x-coordinate of a vertex of } P \text{ or a corner of a}\\
        & \text{ rec-pack in }\mathscr R \cup \{S\}\text{ such that }|x - x'| \le 1, y\text{ is the y-coordinate of a vertex in }P\}\\
        D_y &:=  \{(x, y) \in \mathbb Z^2 \mid \exists y'\text{ which is a y-coordinate of a vertex of } P \text{ or a corner of a}\\
        & \text{ rec-pack in }\mathscr R \cup \{S\}\text{ such that }|y - y'| \le 1, x\text{ is the x-coordinate of a vertex in }P\}
    \end{align*}
    Let $p$, $p_1$ be neighbouring nodes in $\GS$ such that $p \in S$, $p_1 \notin S$. Then, there is a neighbour $p'$ of $p$ in $\GS$, such that $p' \notin S$ but $p, p' \in S^\star$ for some maximal square $S^\star$ having a corner in $D_x \cup D_y$.
\end{lemma}

Before we start formally proving \cref{lem:1-pad}, we provide some intuition on the statement, and briefly explain why this is needed. Note that, $D_x$ is a collection of x-coordinates of vertices/corners in $P$, $\mathscr R$ or $S$ as well as the x-coordinates differing by at most $1$, intersecting with y-coordinates of the vertices of $P$. $D_y$ is defined symmetrically. Let $\mathcal D$ be the set of all maximal squares with a corner in $D_x \cup D_y$. Now, what \cref{lem:1-pad} mentions is: if there is a neighbour of $p \in S$ outside $S$, then one of the squares in $\mathcal D'$ will cover $p$ as well as something outside $\mathscr R \cup \{S\}$.

This will be particularly useful in detecting unambiguous squares as follows. Let $S$ be the square that is to be tested if it is an unambiguous square given $\mathscr R$. If $S$ contains a simplicial node $p$ such that $S$ covers all neighbours of $p$ in $\GS$, then by definition, $S$ is an unambiguous square. No square in $\mathcal D'$ can cover $p$ as well as any uncovered block $q$ outside $S$, because $q$ must be a neighbour of $p$ in $\GS$, $S$ already covers all neighbours of $p$. On the other hand, if $S$ is not an unambiguous square, all blocks inside $S$ would either be covered by rec-packs in $\mathscr R$ or by a square in $\mathcal D'$ that covers something else that is not covered by $\mathscr R \cup \{S\}$. We now prove \cref{lem:1-pad}.

\begin{proof}[Proof of \cref{lem:1-pad}]
    Without loss of generality, assume that the x-coordinate and the y-coordinate of $p_1$ are not less than the x-coordinate and the y-coordinate of $p$ respectively. Let $L$ be the axis-parallel rectangular region with $p$ and $p_1$ being diagonally opposite corner blocks. Any valid square containing $p$ and $p_1$ contains the entire region of $L$. So, all blocks lying in $L$ are neighbours of $p$ in $\GS$.
    
    Consider placing a token on $p_1$. We move the token as follows:
    
    \begin{itemize}
        \item If the block immediately below the token is not covered by $\mathscr R \cup \{S\}$, and does not have a y-coordinate less than that of $p$, then move the token one block down.
        \item If the previous step is not possible, and if the block immediately to the left of the token is not covered by $\mathscr R \cup \{S\}$, and does not have an x-coordinate less than that of $p$, then move the token one block to the left.
        \item Stop when none of these are applicable. Let the final block of the token be denoted as $p'$.
    \end{itemize}
    
    Note that the token stays inside $L$, and therefore, $p'$ must be a node in $\GS$. Further, $p'$ must be adjacent to $p$ in $\GS$. Since the process terminated, the block immediately to the left of $p'$ and the block immediately below $p'$ are either covered with some rec-pack in $\mathscr R \cup \{S\}$, or lie outside the polygon $P$. We consider the following cases. 

    \subparagraph*{Case I: $p$ and $p'$ differ in both $x$- and $y$-coordinates (Figure~\ref{fig:outside-oblique}).} Since the process stopped with the token on $p'$, the block immediately to the left and the block immediately below $p'$ are covered by rec-packs in $\mathscr R \cup \{S\}$. Let $S'$ be a maximal square covering both $p$ and $p'$. From \cref{lem:2-side-poly}, $S'$ has either its horizontal sides or its vertical sides overlapping with polygon edges. Assume that its horizontal sides overlapping with horizontal polygon edges $e_1$, $e_2$ (a similar argument can be made for vertical sides). We keep translating $S'$ to the left until:
    \begin{itemize}
        \item (i) it gets obstructed by a polygon edge to its left and translating it any further would take a part of the square outside the polygon $P$
        \item (ii) it touches $e_1$ at just a point, and translating it any further would make the square lose contact with $e_1$
        \item (iii) it touches $e_2$ at just a point, and translating it any further would make the square lose contact with $e_2$
        \item (iv) $p'$ is on the right boundary of the square, and translating it any further would take $p'$ outside the square
    \end{itemize}
     Let this translated square be $S''$. Due the the fourth condition, $S''$ contains $p$ and $p'$. Note that since $S''$ was obtained by translating $S'$ to the left, the y-coordinates of the corners of $S''$ will share the y-coordinates of $e_1$ or $e_2$. Therefore, the corners of $S''$ share y-coordinates of some vertices of $P$.
     
     If the translation of $S'$ to $S''$ terminated due to the first three conditions, then a vertical side of $S''$ must overlap with some polygon edge (possibly at just a point). Therefore, one of the corners of $S''$ shares its x-coordinate with a vertex in $P$, and also its y-coordinate with a vertex in $P$. So, $S^\star = S''$ satisfies the conditions of the lemma, as these coordinates appear in $D_x$, as well as $D_y$.
     
     We now consider the fourth condition: the translation was stopped as $p'$ is on the right boundary. Recall that both blocks immediately to the left and below $p'$ are covered by some rec-packs $\mathscr R \cup \{S\}$. Therefore, $S''$ has its right side one unit to the right of a vertical side of some edge of rec-pack in $\mathscr R \cup \{S\}$ (the rec-pack that covers the immediate left block of $p'$ but not $p$). Then, the x-coordinate of at least one corner of $S''$ differs by 1 from the x-coordinate of a corner of some rec-pack. Therefore, $S^\star = S''$ has a corner in $D_y$. 
        
        \begin{figure}[ht]  
            \hfill    
            \begin{subfigure}[t]{0.40\textwidth}
                \includegraphics[width=\textwidth]{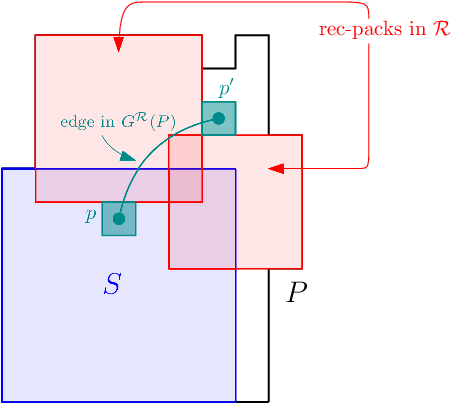}
                \caption{Case I in proof of \cref{lem:1-pad}}\label{fig:outside-oblique} 
            \end{subfigure} \hfill
            \begin{subfigure}[t]{0.40\textwidth}
                \includegraphics[width=\textwidth]{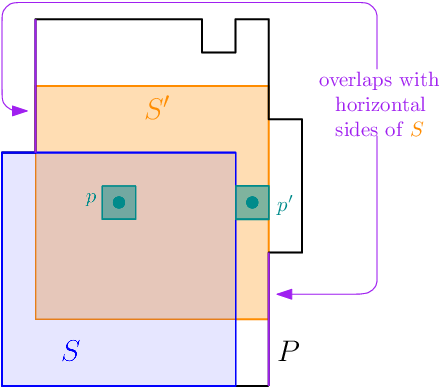}
                \caption{Case II in proof of \cref{lem:1-pad}}\label{fig:outside-orthogonal} 
            \end{subfigure} \hfill \hfill
            \caption{Cases in proof of \cref{lem:1-pad}}
        \end{figure}
        
    \subparagraph*{Case II: $p$ and $p'$ have the same $y$-coordinate.} Let $S'$ be a maximal square covering both $p$ and $p'$. If $S'$ has its horizontal sides overlapping with polygon edges, we have an identical argument as before: translate $S'$ as long as the four conditions hold true, and then the exact same steps prove that the translated square has a corner in $D_y$. Therefore, we only consider the case when $S'$ has vertical sides that overlap with polygon edges $e_1$, $e_2$ (Figure~\ref{fig:outside-orthogonal}). Now, we translate $S'$ to the top until:
    \begin{itemize}
        \item (i) it gets obstructed by a polygon edge to its top and translating it any further would take a part of the square outside the polygon $P$
        \item (ii) it touches $e_1$ at just a point, and translating it any further would make the square lose contact with $e_1$
        \item (iii) it touches $e_2$ at just a point, and translating it any further would make the square lose contact with $e_2$
        \item (iv) $p'$ (and hence $p$) is on the bottom boundary of the square, and translating it any further would take $p'$ (and hence $p$) outside the square
    \end{itemize}
    Let the square obtained be $S_t$. Similar to the previous case, if the translation terminated because of any of the first three conditions, then $S_t$ must have a horizontal side overlapping with some horizontal polygon edge (hence would have a corner in $D_x$, proving what we want). 
    
    Now, consider the fourth case: $p, p'$ are on the bottom boundary of $S_t$. We similarly construct $S_b$ by moving $S'$ to the bottom until either it gets obstructed by (i) a polygon edge to its bottom, (ii) overlaps with $e_1$ or $e_2$ at just a point, or (iii) $p'$ (and hence $p$) is on the top boundary of the $S_b$. Again as the first two of these cases already achieve $S_b$ having a corner in $D_x$, we consider the third case: $p, p'$ are on the top boundary of $S_b$. 
    
    We observe that the bottom side of $S_t$ does not lie below the bottom side of $S$, as $p, p'$ are on the bottom boundary of $S_t$. Similarly the top side of $S_b$ does not lie above or on the top side of $S$. Again if $S_t$ has its bottom side lying on the bottom side of $S$, $S_t$ would have a corner in $D_x$ and we would be done. So, we consider the case when $S_t$ has its bottom side lying above the bottom side of $S$ and $S_b$ has its top side lying below the top side of $S$. 
    
    \begin{claim}
        The top side of $S_t$ must lie above or on the top side of $S$
    \end{claim} 
    \begin{proof}[Proof of claim]
        Assume the contrary. Then, both the top side and the bottom edge of $S_t$ lie between the top side and bottom side of $S$. But $S_t$ covers $p'$ which lies to the right of the right side of $S$. This means $S_t$ has its left side lying completely in the strict interior of $S$, contradicting that $S_t$ has vertical sides overlapping with polygon edges $e_1$, $e_2$.
    \end{proof}
    
    This means, if we were to vertically move a square initially positioned as $S_t$ (top side of $S_t$ above the top side of $S$) to $S_b$ (top side of $S_b$ below the top side of $S$), there would be a square $S_1$ with its top side lying on the top side of $S$ which contains $p$ and $p'$. Moreover, as $S_1$ has its vertical sides overlapping with polygon edges $e_1$, $e_2$, $S^\star = S_1$ has a corner in $D_x$.
    
    \subparagraph*{Case II
    I: $p$ and $p'$ have the same $x$-coordinate.} An argument symmetric to the previous argument follows for this case. 

    This completes the proof.
\end{proof}

\subsection{Placing rec-packs given a partial solution}

We start with a result regarding placement of maximal squares to extend partial solutions.

\begin{lemma} \label{lem:max-sq-2}
    Let $S$ be a maximal square of an orthogonal polygon $P$ with side length $d$, such that
    \begin{itemize}
        \item the top and bottom sides of $S$ overlap with horizontal polygon edges $e_1$ and $e_2$ respectively.
        \item $e_1$ contains the top right corner $a$ of $S$
        \item $e_2$ contains the bottom right corner $b$ of $S$
        \item there is a strip $Y$ between $e_1, e_2$, and the right side of the strip is more than $d$ distance away from the right side of $S$.
    \end{itemize}
    
    Let $S'$ be the square generated by reflecting $S$ with respect to its right side. Then, for any partial solution $\mathscr R$ containing $S$ and with no overlap with $S'$, $\mathscr R \cup \{S'\}$ is also a partial solution.
\end{lemma}


\begin{proof}

    \begin{figure}
        \centering
        \includegraphics[width=0.8\linewidth]{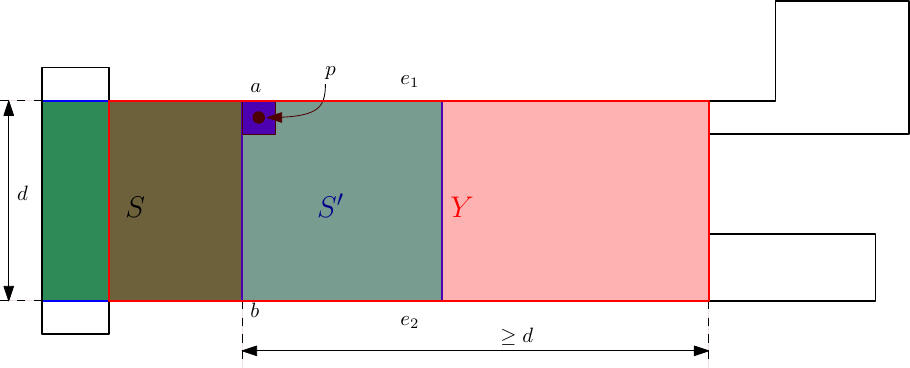}
        \caption{Proof of \cref{lem:max-sq-2}}
        \label{fig:max-sq-2}
    \end{figure}

     $S'$ is a valid square as it lies completely inside $Y$. Let $p$ be the block inside $S'$ which is located at the top-left corner of $S'$. Consider a partial solution $\mathscr R$ containing $S$ but not $S'$. 

     Any block $p'$ other than $p$ in $\GS$ is a neighbour of $p$, if and only if $p'$ is covered by $S'$. Hence all neighbours of $p$ in $\GS$ are covered by $S'$ and hence they must induce a clique in $\GS$. Therefore $p$ is a simplicial node in $\GS$, making $S'$ an unambiguous square given $\mathscr R$. Hence $\mathscr R \cup \{S'\}$ is a partial solution (\cref{lem:unambiguous}).
\end{proof}
    
Note that in \cref{lem:max-sq-2}, $S'$ forms a rec-pack of width $d$ and strength $1$. However, if the strip $Y$ is long enough, we can repeat the same procedure and reflect $S'$ with respect to its ride side to get $S''$. Now $S' \cup S''$ gives us a rec-pack of strength $2$. We can keep repeating this and get rec-packs of much larger strengths, as long as the strip $Y$ permits. We formalize this in the next result.

\begin{lemma} \label{lem:max-rec-pack-2}
    Let $S$ be a maximal square of an orthogonal polygon $P$ with side length $d$ such that, for some natural number $\eta \in \mathbb N$,
    \begin{itemize}
        \item the top and bottom sides of $S$ overlaps with horizontal polygon edges $e_1$ and $e_2$ respectively.
        \item $e_1$ contains the top right corner of $S$
        \item $e_2$ contains the bottom right corner of $S$
        \item there is a strip $Y$ between $e_1, e_2$, and the right side of the strip is more than $\eta d$ distance away from the right side of $S$.
    \end{itemize}
    Let $R$ be the rec-pack of width $d$ and strength $\eta$, such that
    \begin{itemize}
        \item the the vertical sides of $R$ are of length $d$, and the horizontal sides are of length $\eta d$
        \item the right side of the square $S$ coincides with the left side of $R$
    \end{itemize} 
    If there is a partial solution $\mathscr R$ containing $S$ such that $R$ does not overlap with any rec-pack in $\mathscr R$, then $\mathscr R \cup \{R\}$ is also a partial solution.
\end{lemma}

\begin{proof}
    The proof is just a repeated application of \cref{lem:max-sq-2}. Let $S_1$ be the square obtained by reflecting $S$ along its right side, let $S_2$ be the square obtained by reflecting $S_1$ along its right side, and so on, till we define $S_\eta$. Note that $R$ is same as the region defined by $S_1 \cup S_2 \cup \cdots \cup S_\eta$.

    Applying \cref{lem:max-sq-2} to $\mathscr R$ and $S$, we get that $\mathscr R \cup \{S_1\}$ is a partial solution. Applying it again to $\mathscr R \cup \{S_1\}$ and $S_1$, we get $\mathscr R \cup (\{S_1\} \cup \{S_2\})$ is a partial solution. Continuing in this fashion gives us that $\mathscr R \cup (\{S_1\} \cup \{S_2\} \cup \cdots \{S_\eta\})$. This directly implies that $\mathscr R \cup \{R\}$ is a partial solution.
\end{proof}

\begin{remark}
    \cref{lem:max-sq-2} and \cref{lem:max-rec-pack-2} holds true symmetrically for all four directions.
\end{remark}

Notice that in \cref{lem:max-rec-pack-2}, $\eta$ can be arbitrarily large. However, detecting if \cref{lem:max-rec-pack-2} is applicable does not become any more difficult for large values of $\eta$. This will be crucially used to add multiple squares as a single rec-pack to the partial solution, in a single iteration. In later sections we will show how this makes our algorithm run in polynomial time with respect to $n$, the guarantee being independent of the output.

\section{Polynomial-Time Algorithm with respect to the Number of Polygonal Vertices}\label{sec:polytime}
In this Section, we design polynomial-time algorithm for {\sc OPCS} with respect to the number of vertices $n$, of the orthogonal polygon. The algorithm runs in $\mathcal{O}(n^{10})$ time. We first discuss some subroutines the algorithm algorithm uses, followed by the algorithm itself with its analysis.

\subsection{Checking if a set of rec-packs covers an orthogonal polygon}

The first subtask we discuss is to check whether a given set of rec-packs $\mathscr R$ completely covers an orthogonal polygon $P$. However, this is equivalent to checking if the area of the rec-packs in $\mathscr R$ is equal to the area of $P$.

However, this is exactly Klee's measure problem in two dimensions~\cite{Klee1977CanTM}, asking the area of the union of rectangles. Using Bentley's Algorithm~\cite{ben-or1983}, this can be solved in time $\OO(|\mathscr R| \log |\mathscr R|)$.

\begin{lemma} \label{lem:check-algo}
    There exists an algorithm to check in $\OO(|\mathscr R| \log |\mathscr R|)$, whether the set of rectangles $\mathscr R$ completely cover an orthogonal polygon $P$.
\end{lemma}

We assumed that the area of the polygon $P$ is already known. This is a natural assumption as this can be computed in $\OO(n)$ time using the shoelace formula or the surveyor's area formula~\cite{bart1986}, which is asymptotically the same as just reading the polygon.

\subsection{Detecting unambiguous squares}

Next, we discuss an algorithm that takes in a square $S$, a partial solution $\mathscr R$ for {\sc OPCS} and the input orthogonal polygon $P$; and decides if $S$ is an unambiguous square given $\mathscr R$. We crucially use \cref{lem:1-pad} for the correctness of the algorithm.

\begin{lemma} \label{lem:detect-algo}
    Given an orthogonal polygon $P$, a partial solution $\mathscr R$ for {\sc OPCS} and valid square $S$, there is an algorithm that checks if $S$ is an unambiguous square given $\mathscr R$ in $\OO(n(n + |\mathscr R|)^2)$ time.
\end{lemma}

\begin{proof}
    We assume $S$ is maximal. We start by constructing the sets $D_x$ and $D_y$ as in \cref{lem:1-pad}. By definition, $|D_x|, |D_y| = \OO((n + |\mathscr R| + 1) \cdot n) = \OO(n^2 + |\mathscr R|n)$. Let $\mathcal D$ be the set of maximal squares with a corner in $D_x \cup D_y$. Among these, let $\mathcal D' \subseteq \mathcal D$ be the set of maximal squares which cover at least one block that is not covered by $\mathscr R$ or $S$.
    
    Suppose there is a block $q$ covered by $S$, which has a neighbour $q'$ in $\GS$ lying outside $S$. By \cref{lem:1-pad}, there will be a square $S_q \in \mathcal D'$ which covers $q$.
    
    $S$ is an unambiguous square given $\mathscr R$ if and only if there is a simplicial node $p$ in $\GS$ such that $S$ covers $p$ and all its neighbours in $\GS$ (\cref{def:unambiguous}). 
    Therefore, there is no square $S_p \in \mathcal D'$ that covers $p$. This means that $S$ is not unambiguous if and only if the rectangles in $\mathscr R \cup \mathcal D'$ cover up $S$ entirely (along with possibly some portion outside $S$ as well). 

    We can use \cref{lem:check-algo} to actually check this in time $\OO((|\mathscr R| + |\mathcal D'|) \log (|\mathscr R| + |\mathcal D'|)) = \OO((n^2 + |\mathscr R|n)\log(n^2 + |\mathscr R|n))$ by checking if the rectangles $\{R \cap S \mid R \in \mathscr R \cup \mathcal D'\}$ cover $S$ completely or not.

    We have $|\mathcal D'| \le |\mathcal D| \le 4|D_x| + 4|D_y| = \OO(n^2 + |\mathscr R|n)$. Computing $D_x, D_y$ are trivial from their definitions. $\mathcal D$ can be computed in $\OO(|\mathcal D|n) = \OO((n^2 + |\mathscr R|n)n)$ time using \cref{lem:unique-mcs}. Moreover, for every square in $\mathcal D$, it can be checked in $\OO(|\mathscr R|)$ time if it lies completely inside the region defined by $\mathscr R \cup \{S\}$. We throw out all such squares, and just keep the rest in $\mathcal D'$. Hence $\mathcal D'$ can be computed in $\OO(|\mathcal D'||\mathscr R|) = \OO((n^2 + |\mathscr R|n)|\mathscr R|)$ time. 

    This gives us a total time complexity of $\OO(n(n + |\mathscr R|)^2)$. Putting everything together, the algorithm we get is as follows.

\begin{algorithm}[H]
    \caption{Check-if-Unambiguous ~~~ \textbf{Input:} $P$, $\mathscr R$, $S$ }\label{alg:detect-algo}
    \begin{algorithmic}
    \If{$S$ is not maximal} \Comment{$\OO(n)$ time}
        \State{\Return{NOT-UNAMBIGUOUS}}
    \EndIf
    \State{Construct $D_x$, $D_y$ as in \cref{lem:1-pad}} 
    \State{$\mathcal D \gets $ \{maximal squares with a corner in $D_x \cup D_y$\}} \Comment{$\OO((n^2 + |\mathscr R|n) n)$}
    \State{$\mathcal D' \gets $ \{squares in $\mathcal D$ covering a block not covered by $\mathscr R \cup \{S\}$\}} \Comment{$\OO((n^2 + |\mathscr R|n)|\mathscr R|)$}
    \If{$\{R \cap S \mid R \in \mathscr R \cup \mathcal D'\}$ cover $S$ completely} \Comment{$\OO((|\mathscr R| + |\mathcal D'|) \log (|\mathscr R| + |\mathcal D'|))$}
        \State{\Return NOT-UNAMBIGUOUS}
    \Else
        \State{\Return UNAMBIGUOUS}
    \EndIf
\end{algorithmic}
\end{algorithm}
\end{proof}

\subsection{Generating rec-packs within a strip}

Our next step would be to algorithmically use \cref{lem:max-rec-pack-2} to extend an existing partial solution by adding a rec-pack. In order to do this, we look into a maximal square $S$, which we call the `seed', and check if there is a corresponding strip $Y$ (as defined in \cref{lem:max-rec-pack-2}). If they exist, we construct the corresponding rec-pack which does not intersect with the current partial solution, and add it to our partial solution. The resultant set of rec-packs still form a partial solution due to \cref{lem:max-rec-pack-2}. We formally discuss this algorithmic result.

\begin{lemma}\label{lem:gen-rec-pack}
    Given a maximal square $S$ of side length $d$, referred to as the `seed', and a partial solution $\mathscr R$ containing $S$, there is an algorithm to test if there are horizontal polygon edges $e_1, e_2$ such that
    \begin{itemize}
        \item the top and bottom sides of $S$ overlaps with horizontal polygon edges $e_1$ and $e_2$ respectively.
        \item $e_1$ contains the top right corner of $S$
        \item $e_2$ contains the bottom right corner of $S$
        \item there is a strip $Y$ between $e_1, e_2$, and the right side of the strip is more than $d$ distance away from the right side of $S$.
    \end{itemize}
    Moreover, if this exists, the algorithm finds and returns the rec-pack $R$ of width $d$, and the largest possible strength $\eta \ge 1$, such that:
    \begin{itemize}
        \item the vertical sides of $R$ are of length $d$ 
        \item right side of the square $S$ coincides with the left side of $R$
        \item $R$ does not overlap with partial solution $\mathscr R$
    \end{itemize}
    The algorithm runs in total $\mathcal O(n + |\mathscr R|)$ time.

\end{lemma}

\begin{proof}
Consider the following algorithm.
\begin{algorithm}[H]
        \caption{Generate-Rec-pack-left ~~~ \textbf{Input:} $P$, seed $S$, set of rec-packs $\mathscr R$}\label{alg:gen-rec-pack}
        \begin{algorithmic}
            \State{Find $e_1$, $e_2$ containing top-right and bottom-right corner of $S$} \Comment{$\OO(n)$ check}
            \If{$e_1$ or $e_2$ is not found} \State{\Return NONE}
            \EndIf
            \State{$x_1 \gets$ x coordinate of right side of $S$}
            \State{$x_2 \gets \min\{\text{x coordinate $x'$ of vertex in $P$ lying between }e_1, e_2 \text{ such that } x' > x_1$\}}
            \State{$x_3 \gets \min\{\text{x coordinate $x'$ of corner in $\mathscr R$ lying between }e_1, e_2 \text{ such that } x' > x_1$\}}
            \State{$d \gets$ side length of $S$}
            \If{$\min(x_2, x_3) - x_1 \le d$}\Comment {rec-pack $R$ does not exist}
            \State{\Return NONE} 
            \EndIf
            \State{$\eta \gets \left \lfloor \dfrac{\min(x_2, x_3) - x_1}{d} \right \rfloor$} \Comment{gets maximum possible strength $\eta$}
            \State{\Return rec-pack $R$ of strength $\eta$, width $d$ sharing left-side with right-side of $S$}
        \end{algorithmic}
    \end{algorithm}

    \begin{figure}
        \centering
        \includegraphics[width=0.8\linewidth]{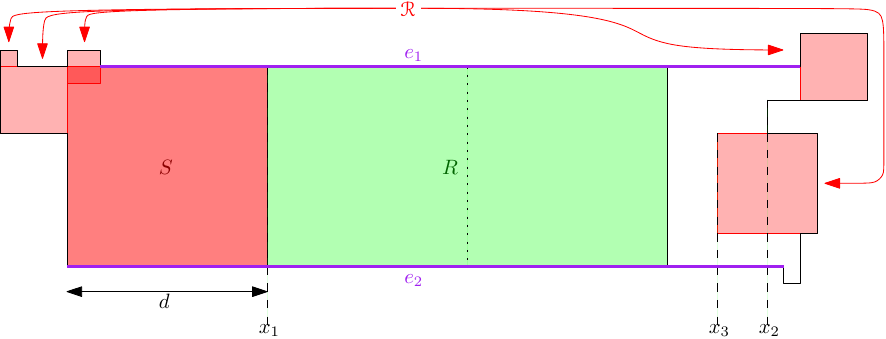}
        \caption{Setting and proof of \cref{lem:gen-rec-pack}}
        \label{fig:gen-rec-pack}
    \end{figure}

    This algorithm first checks for the existence of polygonal edges $e_1$ and $e_2$ by looping over all polygon-edges. Next it checks if there is a valid rec-pack $R$, by checking the x-coordinate $x_2$ of the polygon vertex that bounds the right side of $Y$, as well as the x-coordinate $x_3$, of some rec-pack in $\mathscr R$ that restricts the right side of $R$, as they cannot overlap (Figure~\ref{fig:gen-rec-pack}). If either of these fail, there is no rec-pack $R$, and the algorithm returns `NONE'.

    Otherwise, it finds the largest value of $\eta$ and returns the rec-pack $R$ as required. The entire algorithm runs in time $\OO(n + |\mathscr R|)$ time, as the algorithm just consists of loops over the vertices of $P$ and corners of rec-packs in $\mathscr R$.
\end{proof}

\begin{remark}
    \cref{lem:gen-rec-pack} is applicable for all four directions.
\end{remark}




\subsection{Polynomial-time algorithm: building up partial solutions}
Finally, we are ready to state the polynomial time algorithm for {\sc OPCS} when the input orthogonal polygon $P$ is given in terms of its $n$ vertices. We use the following framework:
\begin{enumerate}
    \item Start with an empty partial solution.
    \item Keep extending the partial solution either by adding unambiguous squares (\cref{alg:detect-algo}) or by adding rec-packs generated by a seed (\cref{alg:gen-rec-pack}).
    \item Terminate when $P$ is completely covered (\cref{lem:check-algo}).
\end{enumerate}

This must be a minimum cover as throughout the algorithm, the set of rec-packs is guaranteed to be a partial solution (\cref{lem:unambiguous} and \cref{lem:max-rec-pack-2}). 

To find unambiguous squares given a partial solution $\mathscr R$ of {\sc OPCS}, we generate all possible maximal squares having end points in $C_x \cup C_y$ as defined in \cref{lem:grid-unambiguous}. We are guaranteed to get an unambiguous square given $\mathscr R$ (\cref{lem:grid-unambiguous}). Next, we test for all such generated squares if it is unambiguous using \cref{alg:detect-algo}. 

Formally, the algorithm is as follows.
\begin{algorithm}[H]
    \caption{Extend-Partial-Solution ~~~ \textbf{Input:} Orthogonal Polygon $P$ with $n$ vertices}\label{alg:polytime}
    \begin{algorithmic}
        \State{$\mathscr R \gets \emptyset$} \Comment{Empty partial solution}
        \While{$\mathscr R$ does not cover $P$} \Comment{$\OO(|\mathscr R| \log |\mathscr R|)$ check, \cref{lem:check-algo}}
            \State{Construct $C_x, C_y$ as \cref{lem:grid-unambiguous}} \Comment{$|C_x|,|C_y| = \OO(n (n + |\mathscr R|))$}
            \State{$\mathscr S' \gets \{\text{squares with corner in }C_x \cup C_y\}$} \Comment{$|\mathscr S'| = \OO(n  (n + |\mathscr R|))$, time $= \OO(|\mathscr S'| \cdot n)$}
            \For{$S \in \mathscr S'$} \Comment{$|\mathscr S'| = \OO(n \cdot (n + |\mathscr R|))$ loop}
                \If{$S$ is unambiguous in $P$ given $\mathscr R$}\Comment{$\OO(n(n + |\mathscr R|)^2)$, \cref{alg:detect-algo}}
                    \State{$\mathscr R \gets \mathscr R \cup \{S\}$} \Comment{$\mathscr R$ remains a partial solution, \cref{lem:unambiguous}}
                    \State{$R_l \gets $ generated rec-pack with seed $S$ to the left}. \Comment {$\OO(n)$, \cref{alg:gen-rec-pack}}
                    \State{$R_r \gets $ generated rec-pack with seed $S$ to the right}. \Comment {$\OO(n)$, \cref{alg:gen-rec-pack}}
                    \State{$R_t \gets $ generated rec-pack with seed $S$ to the top}. \Comment {$\OO(n)$, \cref{alg:gen-rec-pack}}
                    \State{$R_b \gets $ generated rec-pack with seed $S$ to the bottom}. \Comment {$\OO(n)$, \cref{alg:gen-rec-pack}}
                    \For{$R \in \{R_l, R_r, R_t, R_b\}$}
                        \If{$R \ne$ NONE} \Comment{\cref{alg:gen-rec-pack} output}
                            \State{$\mathscr R \gets \mathscr R \cup \{R\}$} \Comment{$\mathscr R$ remains a partial solution, \cref{lem:max-rec-pack-2}}
                        \EndIf
                    \EndFor
                    \State{\textbf{break} out of the \textbf{for} loop, continue to the next \textbf{while} loop iteration}
                \EndIf
            \EndFor
        \EndWhile
        \State{\Return $\sum \limits_{R \in \mathscr R}(\text{strength of }R)$} \Comment{$\OO(|\mathscr R|)$ loop}
    \end{algorithmic}
\end{algorithm}

From the discussion above, we obtain the following statement.

\begin{lemma}\label{lem:stops}
    Algorithm~\ref{alg:polytime} (Extend-Partial-Solution) terminates in finite time and outputs the minimum number of squares to cover $P$.
\end{lemma}
\begin{proof}
    \textbf{Termination.} Each iteration of the while loop takes finite time and all steps outside the while loop take finite time as well. Moreover, the number of while loop iterations is also finite, as in each step the algorithm covers some previously uncovered block. Hence \cref{alg:polytime} terminates in finite time.

    \textbf{Minimum Covering.} By \cref{lem:unambiguous} and \cref{lem:max-rec-pack-2}, $\mathscr R$ is always a partial solution, and therefore if $\mathscr R$ covers $P$ completely, it has to be a minimum cover.
\end{proof}

Our next step is to bound the running time of \cref{alg:detect-algo} with respect to the final set of rec-packs $\mathscr R$.

\begin{lemma}\label{lem:time-square-count}
    If the final set of rec-packs from \cref{alg:polytime} is $\mathscr R$, then its total running time is $\OO(n^2|\mathscr R|(n + |\mathscr R|)^3)$. 
\end{lemma}
\begin{proof}
    First we analyse the running time of each iteration of the while loop.
    \begin{itemize}
        \item Checking the condition of the while loop takes $\OO(|\mathscr R| \log |\mathscr R|)$ time.
        \item Generating $\mathscr S'$ takes time $\OO(n|\mathscr S'|) = \OO(n^2(n + |\mathscr R|))$ time.
        \item Checking for all possible squares in $\mathscr S'$ if they are ambiguous takes time $\OO(n (n + |\mathscr R|) \cdot n(n + |\mathscr R|)^2) = \OO(n^2(n + |\mathscr R|)^3)$. This will asymptotically be the slowest step. 
        \item Generating rec-packs with the seed $S$ and the respective checks are only done four times per while loop iteration. Therefore they take up a total overhead of $\OO(n + |\mathscr R|)$, per while loop iteration (\cref{lem:gen-rec-pack}).
    \end{itemize}

    Hence, each iteration of the while loop takes $\OO(n^2(n + |\mathscr R|)^3)$ time. Moreover, each iteration increases the size of $\mathscr R$ by at least $1$. So the total running time of \cref{alg:polytime} should be $\OO(|\mathscr R| \cdot n^2(n + |\mathscr R|)^3)$.
\end{proof}


The dependence of the running time on the final set of rec-packs $\mathscr R$, makes it seem like an output-sensitive algorithm. However, we bound the size of $\mathscr R$ to be quadratic in $n$.

\begin{lemma}\label{lem:R-bound}
    If $\mathscr R$ be the final set of rec-packs in \cref{alg:polytime}, then $|\mathscr R| = \OO(n^2)$
\end{lemma}

\begin{proof}
    Let $\mathscr R_g \subseteq \mathscr R$ be the set of rec-packs generated from a seed, and let $\mathscr R_u \subseteq \mathscr R$ be the set of unambiguous squares appearing as trivial rec-packs in $\mathscr R$. Therefore, $\mathscr R = \mathscr R_g \cup \mathscr R_u$.

    \begin{claim}\label{clm:Ru-S}
        For any set of squares $\mathscr S$ such that $\Ex(\mathscr R_g) \cup \mathscr S$ covers $P$, we must have $|\mathscr S| \ge |\mathscr R_u|$.
    \end{claim}
    \begin{proof}[Proof of \cref{clm:Ru-S}]
    Since $\mathscr R$ is a minimum cover (\cref{lem:stops}), $\Ex(\mathscr R) = \Ex(\mathscr R_g) \cup \mathscr R_u$ is a minimum cardinality set of squares covering $P$. Therefore, for any set of squares $\mathscr S'$ that covers $P$, we must have $|\mathscr S'| \ge |\Ex(\mathscr R_g) \cup \mathscr R_u|$. In particular, for any set of squares $\mathscr S$ such that $\Ex(\mathscr R_g) \cup \mathscr S$ covers $P$, we must have $|\Ex(\mathscr R_g) \cup \mathscr S| \ge |\Ex(\mathscr R_g) \cup \mathscr R_u|$. This means that $|\mathscr S| \ge |\mathscr R_u|$, as $\mathscr R_u$ and $\Ex(\mathscr R_g)$ are disjoint. 
    \end{proof}
    
    We now construct such a set of squares $\mathscr S$ of size $\mathcal O(n^2)$, with $\Ex(R_g) \cup \mathscr S$ entirely covering $P$. Let us draw a vertical line and a horizontal line through each vertex $v$ of $P$. We will refer to these as \emph{vertex-induced lines of $P$}. As $P$ has $n$ vertices, there are at most $n$ vertex-induced vertical lines and $n$ vertex-induced horizontal lines. 

    These vertex-induced lines form a grid-like structure consisting of rectangular grid cells. Each rectangular grid cell is formed between two consecutive vertex-induced vertical lines and two consecutive vertex-induced horizontal lines. Let $abcd$ be a such a rectangular grid cell, that does not lie outside $P$ (Figure~\ref{fig:abcd-def}). Let $a$ be the top-left corner, $b$ be the top-right corner, $c$ be the bottom-right corner and $d$ be the bottom-left corner. Without loss of generality, assume that the rectangle $abcd$ has its vertical sides are shorter than or the same as its horizontal sides.
    
    Consider the set of rec-packs in $R_g$ to already be placed. We now present a method to cover the entire rectangular region $abcd$ (possibly covering some more portion of the interior of the polygon) with at most $5$ more valid squares. 

    \begin{figure}[ht]  
    \centering    
    \includegraphics[width=60mm]{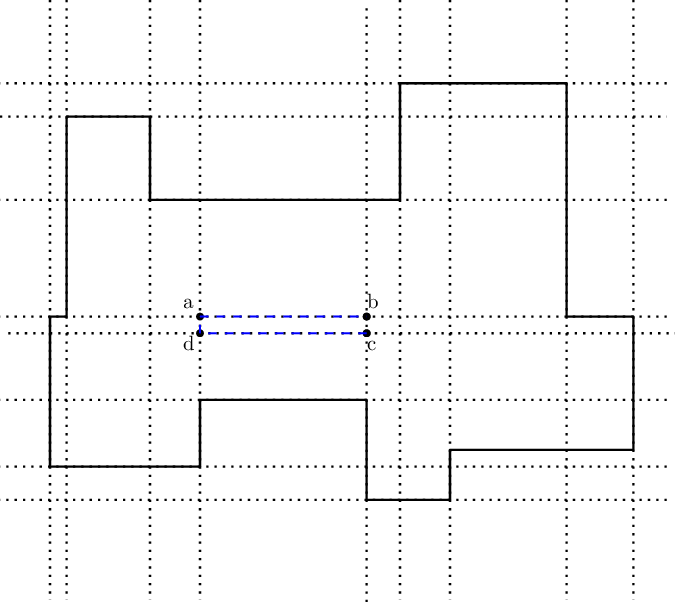}
    \caption{Defining the rectangular region $abcd$}\label{fig:abcd-def} 
    \end{figure}
    
    Consider the two consecutive vertex-induced vertical lines $l_1, l_2$, passing through $a,d$ and passing through $b,c$ respectively. Notice that there cannot be any vertex $v$ of $P$ lying strictly between these two lines; otherwise the vertical vertex-induced line due to this vertex will make $l_1$ and $l_2$ not consecutive. As $abcd$ lies inside $P$, the following must hold. (Figure~\ref{fig:abcd-e1-e2}):

    \begin{figure} [ht]
        \hfill 
        \begin{subfigure}[t]{0.45\textwidth} 
            \centering 
            \includegraphics[width=50mm]{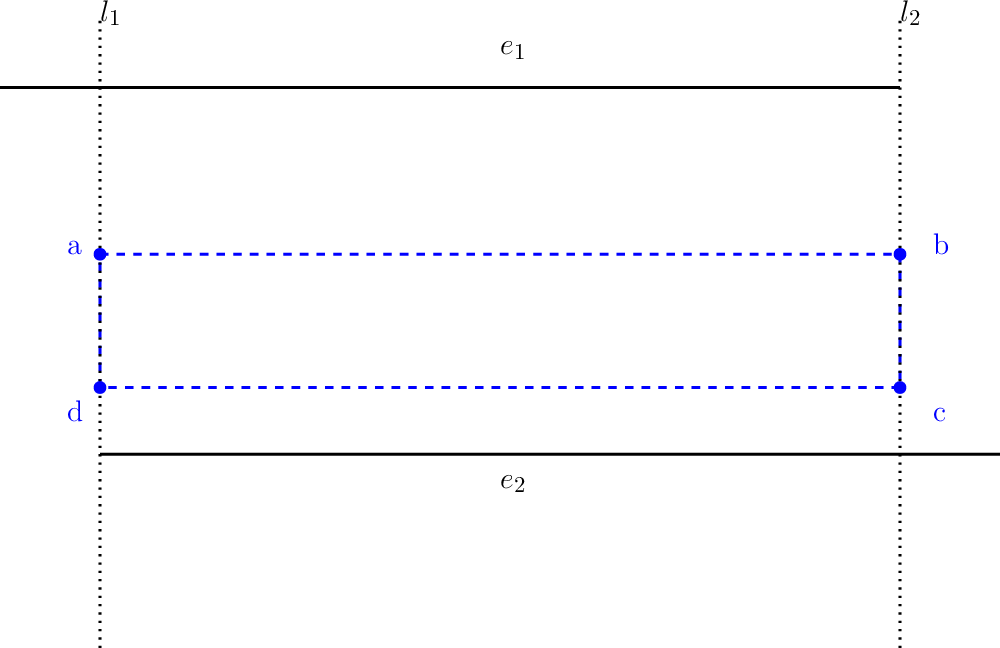}
            \caption{The edges $e_1, e_2$}\label{fig:abcd-e1-e2}
        \end{subfigure} \hfill
        \begin{subfigure}[t]{0.45\textwidth}
            \centering
            \includegraphics[width=50mm]{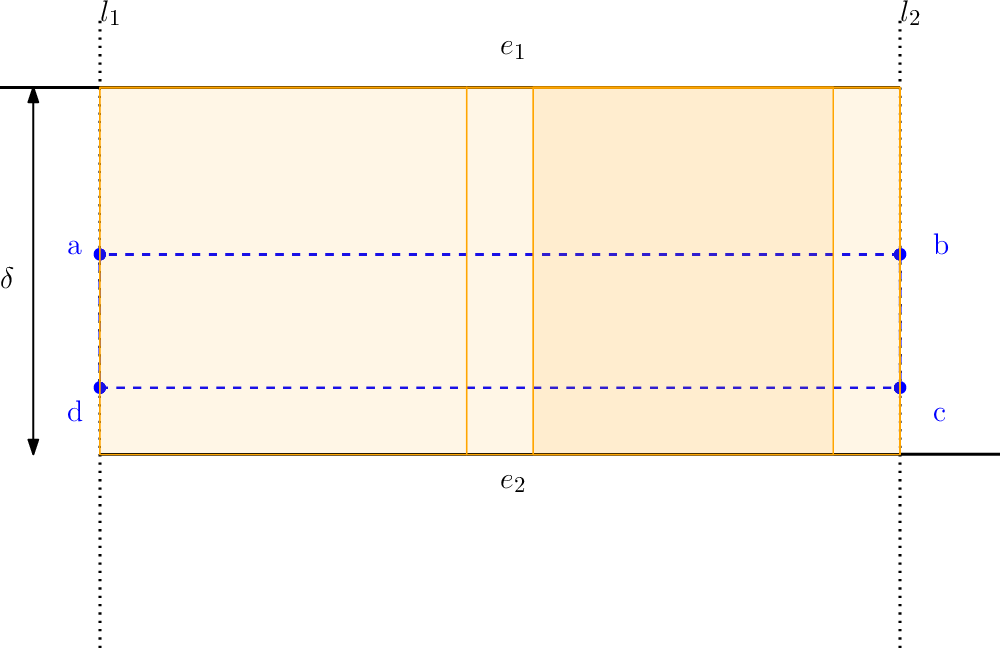}
            \caption{Covering of $abcd$}\label{fig:abcd-cover}
        \end{subfigure} \hfill \hfill \hfill
        \caption{$3$ more squares cover $abcd$}
    \end{figure}

    \begin{itemize}
        \item There must be a horizontal polygon edge $e_1$ of $P$ that intersects both $l_1$ and $l_2$ (possibly at its endpoint), and that either overlaps with the line segment $ab$ or lies above $ab$.
        \item There must be a horizontal polygon edge $e_2$ of $P$ that intersects both $l_1$ and $l_2$ (possibly at its endpoint), and that either overlaps with the line segment $cd$ or lies below $cd$.
        \item The rectangular region enclosed by $e_1$, $e_2$, $l_1$, $l_2$ entirely lies inside the polygon $P$.
    \end{itemize}

    If $\overline{e_i}$ denotes the length of the edge $e_i$, then $\overline{e_1}, \overline{e_2} \ge \overline{ab}$, as $e_1$, $e_2$ intersect both $l_1$, $l_2$. Let $\delta$ denote the (vertical) distance between $e_1$ and $e_2$.  We now discuss the following cases.

    \begin{figure}[ht]  
    \centering    
    \includegraphics[width=80mm]{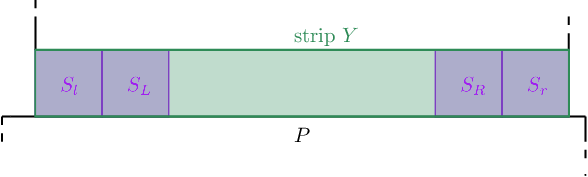}
    \caption{Defining $S_l, S_r, S_L, S_R$ in the strip $Y$}\label{fig:strip-arg} 
    \end{figure}

    \begin{itemize}
        \item \textbf{Case I: $5 \delta \ge \overline{ab}$.} In this case, we can draw at most five $\delta \times \delta$ squares which cover the entire rectangular region enclosed by $e_1$, $e_2$, $l_1$, $l_2$ and hence covers $abcd$ (Figure~\ref{fig:abcd-cover}). Therefore $abcd$ can be covered by at most $5$ squares.
        \item \textbf{Case II: $5 \delta < \overline{ab}$.} In this case, there must be a unique strip $Y$ that is enclosed within $e_1$, $e_2$ (\cref{obs:one-strip-per-pair}). Moreover, $Y$ must have an aspect ratio of more than $5$. Let $S_l$ be the leftmost $\delta \times \delta$ square lying inside $Y$ and $S_L$ be the square obtained by reflecting $S_l$ about its right side. Similarly, let $S_r$ be the rightmost $\delta \times \delta$ square lying inside $Y$ and $S_R$ be the square obtained by reflecting $S_r$ about its left side (Figure~\ref{fig:strip-arg}). Since the aspect ratio of $Y$ is more than $5$, the squares $S_l$,$S_L$,$S_R$,$S_r$ are pairwise non-overlapping $\delta \times \delta$ squares lying inside $Y$. Let $Y'$ denote the region in $Y$ not covered by $S_l$, $S_L$, $S_R$, $S_r$. Since $Y$ is a strip with $\delta$ as the length of the shorter side, any valid square that covers some region outside $Y$ and some region inside $Y$, must overlap with $S_l$ or $S_r$, and can never overlap with $S_L$ and $S_R$. 
        
         Let $S_0 \in \mathscr R$ be the first $\delta \times \delta$ square placed  that covers some portion in $Y$ when \cref{alg:polytime} is run on $P$. We look into the rec-packs generated when $S_0$ is considered as a seed (\cref{alg:gen-rec-pack}); rec-packs generated are only to the right or to the left. Recall that we look at the longest rec-packs, and any square covering points outside strip $Y$ does not overlap with $S_L$, $S_R$. Therefore, $S_0$ along with the rec-packs generated in $\mathscr R_g$ with $S_0$ as a seed together cover $Y'$, and possibly some portions of $S_l$, $S_L$, $S_R$, $S_r$. Hence, with all rec-packs in $R_g$ already placed, we only need $5$ more squares, $S_0$, $S_l$, $S_L$, $S_R$, $S_r$ to cover $Y$ (and therefore cover $abcd$).
    \end{itemize}

    If we repeat the same process for each rectangular grid-cell $abcd$ formed by consecutive vertex-induced lines, then we can cover up the entire polygon $P$ using just $5$ more squares for each such rectangular region lying inside $P$ (along with the rec-packs in $R_g$). However, the total number of such rectangular grid-cells is $\OO(n^2)$. Therefore, if $\mathscr S$ is the set of all such squares (at most $5$ of them per rectangular region $abcd$), then $|\mathscr S| = \OO(n^2)$ and $\Ex(\mathscr R_g) \cup S$ is a set of squares covering the entire polygon $P$. As discussed in \cref{clm:Ru-S}, we have $|\mathscr R_u| \le |\mathscr S| = \OO(n^2)$.

    Moreover, as for each unambiguous square $S$, \cref{alg:polytime} adds at most four rec-packs with seed $S$, the number of rec-packs generated from a seed can be at most $4$ times the number of unambiguous squares. Therefore $|\mathscr R_g| \le 4|\mathscr R_u| = \OO(n^2)$.

    Hence, $|\mathscr R| = |\mathscr R_g| + |\mathscr R_u| = \OO(n^2)$.
\end{proof}

Putting together \cref{lem:time-square-count} and \cref{lem:R-bound}, we get that the total time complexity of \cref{alg:polytime} is $\OO(n^2 \cdot n^2 \cdot (n + n^2)^3) = \OO(n^{10})$

\begin{theorem}\label{thm:polytime}
    For an orthogonal polygon $P$ with $n$ vertices, \cref{alg:polytime} solves \textsc{OPCS} in $\OO(n^{10})$ time.
\end{theorem}

\begin{remark}
    To report the solution instead of just the count, we could return the set $\mathscr R$ of rec-packs in \cref{alg:polytime}. Here, rec-packs are used to efficiently encode multiple squares (potentially exponentially many) into constant sized information.
\end{remark}

\section{Improved Algorithm for Orthogonal Polygons with $k$ Knobs}\label{sec:knobbed}

In this section, we consider the $p$-{\sc OPCS} problem, where our input polygon $P$, has $n$ vertices, and at most $k$ knobs (\cref{def:knobs}), where $k$ is a parameter. We design an algorithm for $p$-{\sc OPCS} that is more efficient than \cref{alg:polytime}, whenever the input instances are such that $k = o(n^{9/10})$.

First, we define a special structure called separating squares, and prove some of its structural results. We will crucially use this to construct our recursive algorithm.

\subsection{More on structure of minimum covering}

We define a separating square, and use the definition to further explore properties of non-knob convex vertices (\cref{def:non-knob-vert}). 

\begin{definition}[Separating Square]\label{sep-square}
    For a convex vertex $v$ of an orthogonal polygon $P$, $MCS(v)$ is said to be a separating square if the region inside $P$ but outside $MCS(v)$ is not a connected region (Figure~\ref{fig:sep-sq}).  
\end{definition}

\begin{figure} [ht]
    \hfill 
    \begin{subfigure}[t]{0.45\textwidth} 
        \centering 
        \includegraphics[width=50mm]{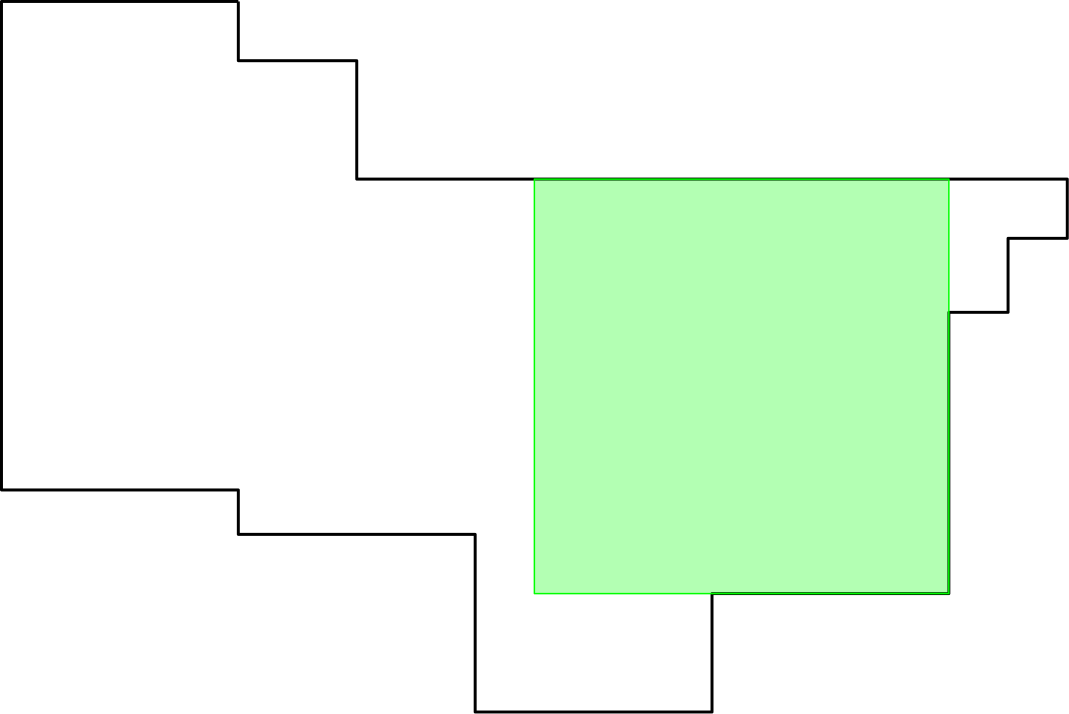}
        \caption{A separating square}
    \end{subfigure} \hfill
    \begin{subfigure}[t]{0.45\textwidth}
        \centering
        \includegraphics[width=50mm]{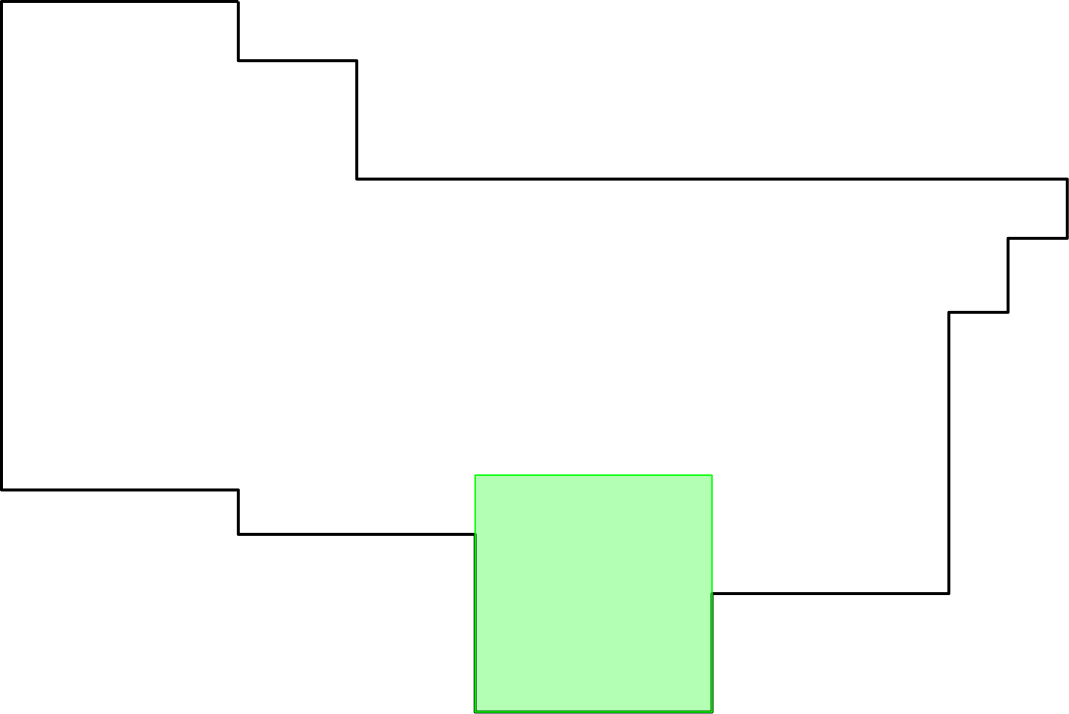}
        \caption{A maximal square which is not a separating square}
    \end{subfigure} \hfill \hfill \hfill
    \caption{Separating Maximal Corner Square}\label{fig:sep-sq}
\end{figure}

This gives us the following result.

\begin{lemma}\label{lem:sep-vert}
    If $v_i$ is a non-knob convex vertex (\cref{def:non-knob-vert}) of an orthogonal polygon $P$, then $MCS(v_i)$ is a separating square which separates $v_{i-2}$ and $v_{i+2}$. 
\end{lemma}
\begin{proof}
    If $v_i$ is a non-knob convex vertex, then any curve lying inside $P$, and having end points at $v_{i+2}$ and $v_{i-2}$ must intersect $MCS(v_i)$ and hence $MCS(v_i)$ must be a separating square separating $v_{i-2}$ and $v_{i+2}$.
\end{proof}

\begin{remark} \label{rem:find-sep}
    Given an orthogonal polygon $P$, we can find a non-knob convex vertex in $\OO(n)$ time (or report that it does not exist) by a simple check on all vertices.
\end{remark}

We now use this to recursively obtain simpler.

\subsection{Recursion with separating squares}

A separating square separates an input orthogonal polygon $P$ into unconnected uncovered regions. We will construct two or more polygons from these uncovered polygons which still preserves the information about $\textsc{OPCS}(P)$. 

First, we define the following.

\begin{definition}~\label{def:parts}
    Given an orthogonal polygon $P$, let $S$ be a maximal separating square which is a maximal square due to a non-knob convex vertex of $P$. We classify the set of connected components of $P$ that are separated by $S$ as follows,

    \begin{itemize}
        \item $Q_t$ be the connected components separated by $S$ that only intersect at more than one point with the top side of $S$ (and no other side). 
        \item $Q_b$ be the connected components separated by $S$ that only intersect at more than one point with the bottom side of $S$ (and no other side). 
        \item $Q_l$ be the connected components separated by $S$ that only intersect at more than one point with the left side of $S$ (and no other side). 
        \item $Q_r$ be the connected components separated by $S$ that only intersect at more than one point with the right side of $S$ (and no other side). 
        \item $Q_{tr}$ be the connected components separated by $S$ that only intersect at more than one point with the top side and right side of $S$ and also at the top-right corner of $S$. 
        \item $Q_{br}$ be the connected components separated by $S$ that only intersect at more than one point with the bottom side and right side of $S$ and also at the bottom-right corner of $S$.
        \item $Q_{tl}$ be the connected components separated by $S$ that only intersect at more than one point with the top side and left side of $S$ and also at the top-left corner of $S$. 
        \item $Q_{bl}$ be the connected components separated by $S$ that only intersect at more than one point with the left side and right side of $S$ and also at the bottom-left corner of $S$.
        \item $Q_{trb}$ be the connected components separated by $S$ that only intersect at more than one point with the top side, bottom side and right side of $S$ and also at the top-right corner and the bottom-right corner of $S$. 
        \item $Q_{tlb}$ be the connected components separated by $S$ that only intersect at more than one point with the top side, bottom side and left side of $S$ and also at the top-left corner and the bottom-left corner of $S$. 
        \item $Q_{rbl}$ be the connected components separated by $S$ that only intersect at more than one point with the right side, bottom side and left side of $S$ and also at the bottom-right corner and the bottom-left corner of $S$.
        \item $Q_{rtl}$ be the connected components separated by $S$ that only intersect at more than one point with the right side, top side and left side of $S$ and also at the top-right corner and the top-left corner of $S$.
    \end{itemize}

    We define $\mathscr Q' = \{Q_t, Q_b, Q_r, Q_l, Q_{tr}, Q_{br}, Q_{tl}, Q_{bl}, Q_{trb}, Q_{tlb}, Q_{rbl}, Q_{rtl}\}$. Further, we define $\mathscr Q = \{Q \in \mathscr Q' | Q \text{ is non-empty}\}$.
\end{definition}

Please refer to Figure~\ref{fig:sep-red} for illustrations of some of these components.

\begin{lemma}\label{lem:sep-red}
Given an orthogonal polygon $P$, let $S$ be a maximal separating square which is a maximal square due to a non-knob convex vertex of $P$. Let $\mathscr Q$ be as defined in~\cref{def:parts}. Then the following must hold true.

    \begin{itemize}
        \item $\forall Q \in \mathscr Q$, $S \cup Q$ is a connected orthogonal polygon without holes.
        \item $\text{\OPCS}(P)= \left(\sum \limits_{Q \in \mathscr Q} \text{\OPCS}(S \cup Q)\right) - (|\mathscr Q| - 1)$
    \end{itemize}
    
\end{lemma}

\begin{figure}[!htbp]  
\centering    
\includegraphics[width=110mm]{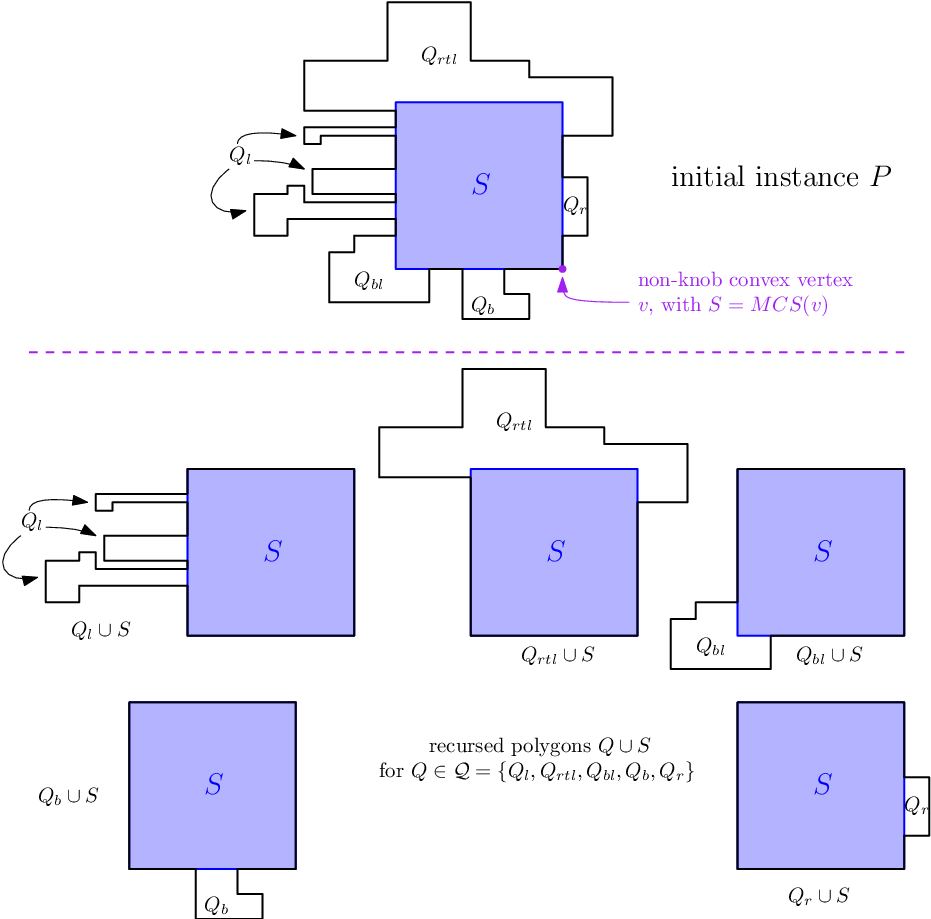}
\caption{$\text{\OPCS}(P)= \left(\sum \limits_{Q \in \mathscr Q} \text{\OPCS}(S \cup Q)\right) - (|\mathscr Q| - 1)$ with $\mathscr Q = \{Q_l, Q_{rtl}, Q_r, Q_{bl}, Q_b\}$ in this case}\label{fig:sep-red} 
\end{figure}

\begin{proof}

    First, we observe that since $S$ is a separating square which is a maximal square due to a non-knob convex vertex $v$, there cannot be a component that intersects with three vertices of $S$ (otherwise $S$ would not be maximal and could be grown by fixing a corner at $v$).

    To see that for all $Q \in \mathscr Q$, $S \cup Q$ is connected, it is sufficient to observe that we can find a curve in the strict interior of $S \cup Q$ from any point $t_1$ in its interior to any point $t_2$ in its interior; either lying inside a single connected component of $Q$ (if $t_1, t_2$ are in that component), or through $S$ (if $t_1, t_2$ are in different components of $Q$). Further, $Q \cup S$ cannot have holes as $P$ does not have holes. 

    Since $S$ is maximal, observe that there cannot be a valid square of $P$, that covers two different points from distinct $Q, Q' \in \mathscr Q$. Also, in a minimum covering with $S$ as one of the squares, all other valid squares must cover at least one point from the interior of exactly one $Q \in \mathscr Q$ (otherwise this square would be completely inside $S$, hence redundant). Therefore if $P$ is covered using a set $\mathscr S$ of $C$ squares, where $S \in \mathscr S$, we can cover $Q \cup S$ using the squares in $\mathscr S$ that cover some part of $Q$ and the square $S$ itself. If we do this for all $Q \in \mathscr Q$ this uses $C + (|\mathscr Q| - 1)$ squares, as $S$ is in this cover of all $|\mathscr Q|$ instances $Q \cup S$. But $S$ appears only once in $\mathscr S$. Since we start with a minimum cover of $P$ and find a valid cover of all $Q \cup S$ polygons with exactly $(|\mathscr Q| - 1)$ more squares in total, we get, 
    
    $$\text{\OPCS}(P) \ge \left(\sum \limits_{Q \in \mathscr Q} \text{\OPCS}(S \cup Q)\right) - (|\mathscr Q| - 1)$$
    
    Now, given any minimum covering of all $Q \cup S$, all such instances must contain $S$ (as $S$ is a maximum corner square). So we superimpose these coverings and delete all but one copy of $S$ to get a covering of $P$. This time we started from a minimal covering of the individual $Q \cup S$ polygons to get a valid covering of $P$. Including this with the rest of the result, we obtain
    
    $$\text{\OPCS}(P) = \left(\sum \limits_{Q \in \mathscr Q} \text{\OPCS}(S \cup Q)\right) - (|\mathscr Q| - 1)$$
\end{proof}

We now prove a crucial result: such a recursive step does not increase the number of knobs in each individual instance $Q \cup S$.

\begin{lemma}\label{lem:sep-red-k}
    Let $S$ be a maximal separating square which is a maximal square due to a non-knob convex vertex of an orthogonal polygon $P$ with $n$ vertices and $k$ knobs. Let $\mathscr Q$ be defined as in \cref{def:parts}. Then, for every $Q \in \mathscr Q$, $(Q \cup S)$ is an orthogonal polygon without holes with at most $n$ vertices and at most $k$ knobs. Moreover, any vertex in $Q \cup S$ which is a corner of $S$ is part of a knob in $Q \cup S$ that coincides with a side of $S$.
\end{lemma}

\begin{proof}
    \cref{lem:sep-red} already proves that $\forall Q \in \mathscr Q, (Q \cup S)$ is an orthogonal polygon without holes. Further, the number of vertices can only decrease because the only time a new vertex (vertex not in $P$) would be introduced is when $S$ already distributes the existing vertices of $P$ in each of $(Q \cup S), Q\in \mathscr Q$ (causing no total increase in the number of vertices in $(Q \cup S)$ than in $P$). We now prove that $\forall Q \in \mathscr Q, Q \cup P$ has at most $k$ knobs. 

    Consider $Q \cup S$ for some $Q \in \mathscr Q$. Without loss of generality assume $Q$ is either $Q_l$ or $Q_{tl}$ or $Q{tlb}$ (all other cases have symmetric arguments). We now show that for all knobs in $Q \cup S$, there is a distinct knob in $P$ (which would show that $Q \cup S$ has at most $k$ knobs if $P$ has at most $k$ knobs. Clearly any (distinct) knob $(u1, u2)$ of $Q \cup S$ such that $u1$, $u_2$ are not corners in $S$, $(u_1, u_2)$ must be a (distinct) knob in $P$ (knob in $Q$, in particular). Hence we only need to consider knobs which $(u_1, u_2)$ in $Q \cup S$ such that either $u_1$ or $u_2$ is a vertex of $S$. Without loss of generality, we assume that $u_1$ is a corner in $S$.

    \begin{itemize}
        \item \textbf{Case I: $Q = Q_{tlb}$}. In this case the top-left and the bottom-left corners of $S$ completely lie inside $Q$ and hence are not vertices of $Q \cup S$. Therefore $u_1$ is either the bottom-right corner or the top-right corner of $S$ (both are vertices in $Q \cup S$). Without loss of generality, assume $u_1$ to be the top-right corner of $S$. Consider $e_1$ to be the horizontal edge of $Q \cup S$ (\textit{i.e.} the edge overlapping with top edge of $S$) with $u_1$ as a corner. Since $Q_{tlb}$ intersects with the right corner of $S$, $e_1$ has to shorter than the side length of $S$. As the entire region inside $S$ is inside $Q \cup S$, the other end point of $e_1$ must be a concave vertex. Therefore $(u_1, u_2) \ne e_1$. Moreover, as the other end point of the vertical edge of $Q \cup S$ from $u_1$ is the bottom-right corner of $S$, $u_2$ must be the bottom right corner of $S$; $(u_1, u_2)$ must be a left knob (and the only knob) in $Q \cup S$ (which proves that all vertices in $Q \cup S$ which are corners in $S$ are part of knobs in $Q \cup S$, along a side of $S$). Moreover, for this knob $(u_1, u_2)$, can find a left knob of the region $P \setminus Q$ (which must be a knob in $P$ as $P \setminus Q$ intersects $Q$ only in top, bottom and right edges). Therefore, in this case, for every knob in $S \cup Q$, we can find a distinct knob in $P$. 
        \item \textbf{Case II: $Q = Q_{tl}$}. We consider four subcases. 
        \begin{itemize}
            \item \textbf{Case II(a): $Q$ does not touch the top-right corner or the bottom-left corner of $S$.} By arguments similar to above, we can show that right side of $S$ and the bottom side of $S$ are right and bottom knobs of $Q \cup S$ (which proves that all vertices in $Q \cup S$ which are corners in $S$ are part of knobs in $Q \cup S$, along a side of $S$). By a similar argument $P \setminus Q$ must have a right knob and a bottom knob which are also a right knob and a bottom knob of $P$. Therefore, in this case, for every knob in $S \cup Q$, we can find a distinct knob in $P$. 
            \item \textbf{Case II(b): $Q$ touches the bottom-left corner of $S$, but not the top-right corner of $S$.} Similar to before, the right side of $S$ is a right knob in $Q \cup S$ (which proves that all vertices in $Q \cup S$ which are corners in $S$ are part of knobs in $Q \cup S$, along a side of $S$) and we can find a right knob in $P \setminus Q$ which is also a right knob in $P$. However, there can a bottom knob $(u_1, u_2)$ in $Q \cup S$ which has its right endpoint $u_1$ as the bottom-right corner of $S$, but $u_2$ is a convex vertex in $Q$ (and hence in $P$). In this case, either $(u_1, u_2)$ is a bottom knob in $P$ (in which case we are done), or there is another vertex $u_3$ of $P$ such that $u_1$ lies between $u_2$ and $u_3$; therefore, $u_3$ must be vertex in $Q_r$ or $Q_{tr}$. Again, if $u_3$ is a a convex vertex in $P$, then $(u_2, u_3)$ is a knob in $P$ and we are done. Otherwise there will be a bottom knob in the component $Q' \in \{Q_r, Q_{tr}\}$ containing $u_3$. Further as $Q'$ can only intersect with $P \setminus Q'$ in a right or a top edge, the bottom knob of $Q'$ must be a bottom knob of $P$. This completes the argument for this case that for every knob in $S \cup Q$, we can find a distinct knob in $P$.
            \item \textbf{Case II(c): $Q$ touches the top-right corner of $S$, but not the bottom-left corner of $S$.} Symmetric argument similar to the previous case.
            \item \textbf{Case II(d): $Q$ touches the bottom-left corner and the top-right corner of $S$.} This means the non-knob convex vertex $v$ for much $S$ was the maximal square covering $v$ must be the bottom-right corner of $S$. However, since $Q$ touches the bottom-left corner and the top-right corner of $S$, we can extend $S$ by fixing its bottom right corner at $v$, contradicting the maximality of $S$. Hence this case can never happen.
        \end{itemize}
        \item \textbf{Case III: $Q = Q_{l}$}. We consider four subcases. 
        \begin{itemize}
            \item \textbf{Case III(a): $Q$ does not touch the top-left corner or the bottom-left corner of $S$.} By arguments similar to Case I, we can show that right side, the bottom side and the top side of $S$ are right, bottom and top knobs of $Q \cup S$ respectively (which proves that all vertices in $Q \cup S$ which are corners in $S$ are part of knobs in $Q \cup S$, along a side of $S$). And by similar argument $P \setminus Q$ must have a right knob, a bottom knob and a top knob which are also a right knob, a bottom knob and a top knob of $P$. Therefore, in this case, for every knob in $S \cup Q$, we can find a distinct knob in $P$. 
            \item \textbf{Case III(b): $Q$ touches the top-left corner of $S$, but not the bottom-left corner of $S$.} Again similar to Case II(a), the right side and the bottom side of $S$ is a right knob and a bottom knob in $Q \cup S$ (which proves that all vertices in $Q \cup S$ which are corners in $S$ are part of knobs in $Q \cup S$, along a side of $S$); and we can find a right knob and a bottom knob in $P \setminus Q$ (and also in $P$). Moreover, there can be a top knob of $Q \cup S$ which has one vertex in $S$ and one vertex in $Q$. Again, this case is similar to the second case of Case II(b) and we can find a top knob in $P$ which is not entirely contained in $Q$. Therefore, in this case, for every knob in $S \cup Q$, we can find a distinct knob in $P$. 
            \item \textbf{Case III(c): $Q$ touches the bottom-left corner of $S$, but not the top-left corner of $S$.} Symmetric argument similar to the previous case.
            \item \textbf{Case III(d): $Q$ touches both bottom-left corner and the top-left corner of $S$.} By arguments similar to Case I, we can show that right side of $S$ is right knob of $Q \cup S$ (which proves that all vertices in $Q \cup S$ which are corners in $S$ are part of knobs in $Q \cup S$, along a side of $S$) and there is a right knob in $P \setminus Q$ which is also a right knob in $P$. Moreover, there can be a top knob of $Q \cup S$ which has one vertex in $S$ and one vertex in $Q$; and a bottom knob of $Q \cup S$ which has one vertex in $S$ and one vertex in $Q$. Again, this case is similar to the second case of Case II(b) and we can find a top knob (bottom knob) in $P$ which is not entirely contained in $Q$. Therefore, in this case, for every knob in $S \cup Q$, we can find a distinct knob in $P$
        \end{itemize}
    \end{itemize}

    Therefore in all cases, we can map knobs of $Q \cup S$ for any $Q \in \mathscr Q$ to distinct knobs in $P$. Therefore the total number of knobs of $Q \cup S$ can be at most $k$, if $P$ had at most $k$ knobs.
\end{proof}

We use the following framework:
\begin{enumerate}
    \item Find a non-knob convex vertex $v$ (if any) in $\OO(n)$ time (\cref{rem:find-sep}).
    \item If no such non-knob convex vertex exists, solve \OPCS using \cref{alg:polytime} (base cases).
    \item If it exists, construct $S := MCS(v)$, construct $\mathscr Q$ and use this recursion to recurse into $|\mathscr Q|$ similar instances where the number of knobs do not increase.
\end{enumerate}

Each recursive step which is not a base case, can be done in linear time $\OO(n)$, by a simple traversal of the polygon vertices. As $|\mathscr Q| \le |\mathscr Q'| \le 12$, we get the following result.

\begin{lemma}\label{lem:recurse}
    If $P$ is an orthogonal polygon with $n$ vertices and at most $k$ knobs, in $\OO(n)$ time, we can either report that no non-knob convex vertex exists, or $z \le 12$ smaller instances of orthogonal polygons $P_1, \ldots, P_z$ which individually have at most $k$ knobs and $n$ vertices.
\end{lemma}

\paragraph*{Polygons without non-knob convex vertices.}

\cref{lem:recurse} implies that whenever we have a non-knob convex vertex, we can recurse in linear time. However, we need to analyse what happens if there are no non-knob convex vertices. Our first result is to bound the number of vertices of such polygons. 

\begin{lemma}\label{lem:sep-free-k-cnt}
    There are at most $4k - 4$ vertices in an orthogonal polygon $P$ with at most $k$ knobs and no non-knob convex vertex. 
\end{lemma}

\begin{proof}
    Let $n_x, n_v$ be the number of convex vertices and concave vertices in $P$ respectively. Since $P$ is an orthogonal polygon with no holes, we have $n_v = n_x - 4$. Moreover, as only convex vertices are part of knobs and there are at most $k$ knobs, we must have $n_x \le 2k$. Therefore the total number of vertices is $n_x + n_v = 2n_x - 4 \le 4k - 4$. 
\end{proof}

Therefore, for such polygons, we have $n = \OO(k)$. We can detect this in $\OO(n) = \OO(k)$ time, and solve \OPCS in $\OO(n^{10}) = \OO(k^{10})$ time using \cref{alg:polytime}.

\begin{lemma}\label{lem:base-case}
    Given an orthogonal polygon $P$ with $n$ vertices and at most $k$ knobs and with no non-knob convex vertices, we can solve \OPCS in $\OO(k^{10})$ time. 
\end{lemma}

\subsection{A recursive algorithm}

We now have the results to design our exact algorithm solving $p$-\OPCS on input orthogonal polygons $P$ having $n$ vertices and at most $k$ knobs. We formally state our recursive framework as an algorithm.
\begin{algorithm}[H]
    \caption{Separating-Square-Recursion ~~~ \textbf{Input:} $P$ with $n$ vertices and at most $k$ knobs}\label{alg:sep-sq-rec}
    \begin{algorithmic}
    \If{$P$ has no non-knob convex vertex} \Comment{$\OO(n)$, \cref{lem:recurse}}
        \State{$C \gets$\OPCS$(P)$} \Comment{\cref{alg:polytime}, $\OO(k^{10})$, refer to \cref{lem:base-case}}
        \State{\Return $C$}
    \EndIf
    \State{$v \gets $ some non-knob convex vertex}
    \State{$S \gets MCS(v)$} \Comment{$\OO(n)$, \cref{lem:unique-mcs}}
    \State{Construct $\mathscr Q$ as in \cref{lem:sep-red}} \Comment{$\OO(n)$}
    \State{$s \gets 0$}
    \For{$Q \in \mathscr Q$} \Comment{recurse, \cref{alg:sep-sq-rec}, $|\mathscr Q| \le 12$}
        \State{$t \gets $ return value when \textbf{Separating-Square-Recursion} is run on $Q \cup S$} 
        \State{$s \gets s + t$}
    \EndFor
    \State{\Return $(s - (|\mathscr Q| - 1))$} \Comment{\cref{lem:sep-red}}
\end{algorithmic}
\end{algorithm}


The correctness of Separating-Square-Recursion (Algorithm~\ref{alg:sep-sq-rec}) is a direct consequence of Lemma~\ref{lem:sep-red} and the correctness of Algorithm~\ref{alg:polytime}, \textit{i.e.} Theorem~\ref{thm:polytime}. We now analyse the time complexity of \cref{alg:sep-sq-rec}.

Firstly, all steps of Algorithm~\ref{alg:sep-sq-rec} take $\OO(n)$ time other than the calls to Algorithm~\ref{alg:polytime} and the recursion step. We now prove some results of this algorithm which helps us bound the number of recursive calls to Algorithm~\ref{alg:sep-sq-rec}.

We observe that if the input polygon is $P_0$ at some recursive step, the chosen separating square for a non-knobbed convex vertex $v$, is $S = MCS(v)$ and the recursed polygons are $P_1, \ldots, P_z$, then the vertices of each of $P_i$ are either vertices in $P_0$ or corners of $S$. With this in mind we prove the following result.

\begin{lemma}
    Let an orthogonal polygon $P$ with $n$ vertices and at most $k$ knobs be the original input to Algorithm~\ref{alg:sep-sq-rec}. At some recursive step, let the input be $P_0$, and let the algorithm choose $v$ to be a non-knob convex vertex of $P_0$. Then $v$ must also be a non-knob convex vertex of the original polygon $P$. Moreover, any corner of $S$ that is also a vertex of $Q \cup S$ can not be a non-knob convex vertex.
\end{lemma}

\begin{proof}
    If $v$ is a non-knob convex vertex in $P_0$, then $v$ must be a vertex in $P$ (and not a vertex introduced by some separating square at some recursive step). This is because all vertices introduced by a separating square at an intermediate recursive step have a knob along the side of the separating square itself (\cref{lem:sep-red-k}).

    Next, if $v$ is a vertex participating in a knob $(u,v)$ in $P$, then $u$ is a convex vertex. Due to this, at any intermediate recursion step, there cannot a polygon $P'$ ($P_0$ in particular) with $v$ as its vertex, that has a concave endpoint to the edge originating from $v$ and along $uv$. Thus, $v$ must be a non-knob convex vertex in $P$.
\end{proof}

We now prove that if $v$ is chosen as a non-knob convex vertex at some recursive step, then no subsequent recursive steps can again choose $v$.

\begin{lemma}\label{lem:bounded-vertex}
    Let $v$ be a non-knob convex vertex of the original input polygon $P$. Then $v$ is chosen as the non-knob convex vertex in at most one recursive step of Algorithm~\ref{alg:sep-sq-rec}.
\end{lemma}
\begin{proof}
    If $v$ is a non-knob convex vertex of the original input polygon $P$, and let it be chosen at some recursive step $\mathfrak r$ for the first time. $v$ cannot be chosen as a non-knob convex vertex by any recursive step $\mathfrak r'$ which does not lie in the subtree of $\mathfrak r$ in the recursion tree (as $v$ does not even appear as a vertex in those instances). However, once $v$ is chosen as the non-knob convex vertex, $v$ becomes a part of a knob (\cref{lem:sep-red-k}) in the subsequent steps (and hence not a non-knob convex vertex). Therefore $v$ is never chosen a the non-knob convex vertex in the subsequent recursive steps either.
\end{proof}

Now, we can bound the number of recursive calls to Algorithm~\ref{alg:sep-sq-rec}.

\begin{lemma}\label{lem:recursive-call}
    The recursion tree of Separating-Square-Recursion (Algorithm~\ref{alg:sep-sq-rec}) has at most $n$ internal nodes running Separating-Square-Recursion, and at most $12n$ leaf nodes which solve the base case by invoking Algorithm~\ref{alg:polytime}.
\end{lemma}
\begin{proof}
 Due to Lemma~\ref{lem:bounded-vertex}, the number of recursive steps where Algorithm~\ref{alg:sep-sq-rec} recurses (internal nodes in recursion tree), is bounded by the number $n$ of vertices in the original input polygon. As any recursive step calls at most $12$ more recursive steps, the number of steps where \cref{alg:sep-sq-rec} achieves the base-case condition and calls Algorithm~\ref{alg:polytime} is bounded by $12n$ (leaf nodes in recursion tree).    
\end{proof}

Finally, we complete the analysis of our algorithm by bounding the total running time.

\begin{theorem}\label{thm:k-knobbed}
    Separating-Square-Recursion (Algorithm~\ref{alg:sep-sq-rec}) when run on an orthogonal polygon $P$ with $n$ vertices and at most $k$ knobs, solves \OPCS in $\OO(n^2 + k^{10} \cdot n)$ time.
\end{theorem}

\begin{proof}
    Each internal node of the recursion tree for the algorithm takes $\OO(n)$ time (\cref{lem:recurse}). Therefore, following from the bound in Lemma~\ref{lem:recursive-call}, the total time taken by the internal nodes of the recursion tree for the algorithm is $\OO(n^2)$. 
    
    Each leaf node of the recursion tree is an execution of Algorithm~\ref{alg:polytime}, each taking $\OO(k^{10})$ time (\cref{lem:base-case}). Now, as there are $\OO(n)$ leaf nodes in the recursion tree (Lemma~\ref{lem:recursive-call}), the base cases together take $\OO(k^{10} \cdot n)$ time. Total running time becomes $\OO(n^2 + k^{10} \cdot n)$.
\end{proof}

\begin{remark}
    We should only prefer \cref{alg:sep-sq-rec} when $k = o(n^{9/10})$, otherwise \cref{alg:polytime} provides better or the same asymptotic running time.
\end{remark}

Note that we may use any exact algorithm for {\sc OPCS} to solve the base case.

\begin{corollary}
    If there is an algorithm solving \OPCS in time $T(n)$ for polygons with $n$ vertices, there exists an algorithm solving $p$-\OPCS on orthogonal polygons with $n$ vertices and at most $k$ knobs in time $\OO(n^2) + n \cdot T(4k-4)$.
\end{corollary}

\subsection{Discussion on orthogonally convex polygons}

We discuss a well-studied special case of orthogonal polygons: orthogonally convex polygons. 

\begin{definition}[Orthogonally convex polygon] \label{def:ortho-convex}
    An orthogonal polygon $P$ is said to be \emph{orthogonally convex} (Figure~\ref{fig:ortho-convex}), if the following hold true.
    \begin{itemize}
        \item $P$ is a simply connected polygon with polygon edges parallel to the x-axis or the y-axis.
        \item The intercept of any line parallel to x-axis or y-axis with $P$ produces \emph{one} continuous (possibly empty) line segment.
    \end{itemize}
\end{definition}

\begin{figure} [ht]
    \hfill 
    \begin{subfigure}[t]{0.45\textwidth} 
        \centering 
        \includegraphics[height=50mm]{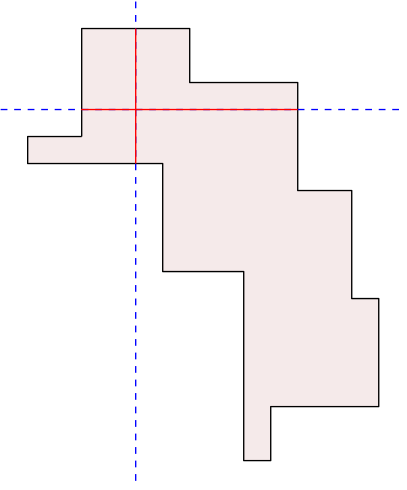}
        \caption{Orthogonally Convex Polygon}
    \end{subfigure} \hfill
    \begin{subfigure}[t]{0.45\textwidth}
        \centering
        \includegraphics[height=50mm]{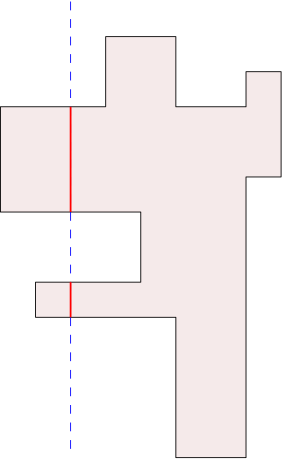}
        \caption{Not Orthogonally Convex}
    \end{subfigure} \hfill 
    \caption{Orthogonally Convex Polygons}\label{fig:ortho-convex}
\end{figure}

It is easy to construct a simple orthogonal polygon having an arbitrarily large number of knobs. However, we show that for orthogonally convex polygons, the number of knobs must be exactly $4$. 

\begin{lemma}\label{lem:4-knobs}
    Any orthogonally convex polygon $P$ contains exactly $4$ knobs. Moreover, $P$ contains exactly one left knob, exactly one right knob, exactly one top knob and exactly one bottom knob.
\end{lemma}
\begin{proof}

    \textbf{Existence.} Consider the leftmost vertical line that intersects $P$ to form a non-empty vertical line segment of intersection. This must intersect with a vertical polygon edge of $P$. The endpoints of this vertical polygon edge form a left knob. Hence a left knob always exists. Symmetric arguments yield that a top knob, a bottom knob and a right knob exists as well.
    
    \begin{figure}[!htbp]  
    \centering    
    \includegraphics[height=70mm]{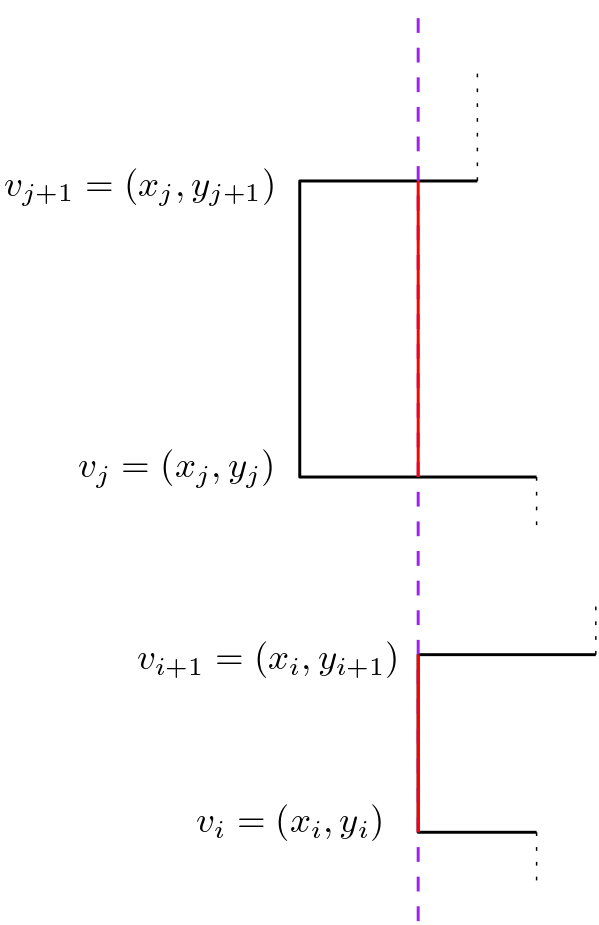}
    \caption{For proof of Lemma~\ref{lem:4-knobs}}\label{fig:4-knobs-proof} 
    \end{figure}
    
    \textbf{Uniqueness.} We show that there cannot be two distinct left knobs. For the sake of contradiction, consider that there are two left knobs $(v_i, v_{i + 1})$ and $(v_j, v_{j+1})$ with $x_i = x_{i + 1} \ge x_j = x_{j+1}$ (refer to Figure~\ref{fig:4-knobs-proof}). Consider the vertical line $x = x_i$. This line intersects $P$ at the entire edge $(v_i, v_{i+1})$ with the intercept the y coordinates being $[\min(y_i, y_{i+1}), \max(y_i, y_{i+1})]$. Now consider any curve lying inside $P$ with $v_i$ and $v_j$ as endpoints. This curve must also intersect the line $x = x_i$ (as $x_j \le x_i$) at some y coordinate outside $[\min(y_i, y_{i+1}), \max(y_i, y_{i+1})]$. This means the intersection of the line $x = x_i$ with $P$ is not a single line segment, which contradicts the assumption that $P$ is orthogonally convex. Therefore, there can be exactly one left knob. By symmetric arguments, there is exactly one right knob, exactly one top knob and exactly one bottom knob.
\end{proof}

Therefore, if we use \cref{alg:sep-sq-rec} to solve \OPCS on orthogonally convex polygons, we can substitute $k=4$ in the analysis, giving us a running time of $\OO(n^2 + n \cdot 4^{10}) = \OO(n^2)$.

\begin{corollary}\label{cor:ortho-convo-time}
    Separating-Square-Recursion (Algorithm~\ref{alg:sep-sq-rec}) takes $\OO(n^2)$ time to solve \OPCS on orthogonally convex polygons. 
\end{corollary}

\section{Hardness Results for Polygons with Holes} \label{sec:hardness}

In this Section, we consider the {\sc OPCSH} problem. We consider the setting as described by Aupperle et. al~\cite{aupperle1988covering}: we need to solve \OPCSH in orthogonal polygons, where the input consists of all $N$ lattice points lying in the polygon or on the boundary. We first state the issue with their claimed proof for \OPCSH being NP-complete. Then we state a correct proof of hardness. We will use some structures defined in their paper, but use novel constructions for some gadgets used in the reduction.

\subsection{Issue with the existing proof}

The proof by Aupperle et. al~\cite{aupperle1988covering} for NP-completeness of \OPCSH gives a reduction from \textsc{Planar 3-CNF}~\cite{doi:10.1137/0211025} to \OPCSH. The reduction gives a polynomial-time algorithm to reduce any \textsc{Planar 3-CNF} instance (say $\psi$) and transforms it to an instance of \OPCSH. First, the formula is negated to obtain $\phi = \neg \psi$, which is a formula in disjunctive normal form (DNF) due to De-Morgan's law. This means that $\psi$ is satisfied if and only if $\phi$ evaluates to {\sf false} on some truth assignment of the variables. Now, $\phi$ is reduced to an instance of \OPCSH using three kinds of gadgets: \emph{wires}, \emph{variable gadgets} and \emph{junction gadgets}. A variable gadget is introduced for each variable of $\phi$ and a junction gadget is introduced for each conjunction in $\phi$. Further wire gadgets (or simply, wires) are introduced to `connect' junction gadgets to variable gadgets. 

The variable gadgets, wires and the junction gadgets shall be placed according to the formula $\phi$. For each variable gadget or each wire, it is easy to compute, in polynomial time, the minimum number of squares needed to cover the gadget. It is possible that some of the squares used for covering a wire may cover a region inside a junction gadget. Each variable gadget permits two kinds of covering (representing the two kinds of truth assignment to a variable in $\phi$). Let us denote these coverings by {\sf true}-covering and {\sf false}-covering, respectively, Both coverings require the same number of squares. Similarly we can extend this observation to see that there are two ways in which a variable gadget and its adjacent wires can be covered optimally; both coverings use the same number of squares. A combination of how a variable gadget is covered affects which portions of all its corresponding junction gadgets (gadgets for the \textsc{and}-clauses containing the corresponding variable in $\phi$) will be covered. Let $V$ be the number of squares required to cover all variable gadgets and $W$ be the number of squares required to cover all wire gadgets. 

In the paper by Aupperle et. al~\cite{aupperle1988covering}, it is claimed that the junction gadgets (or simply, a junction) are such that, the minimum number of squares needed to cover a junction gadget is $12$ if all three variable gadgets corresponding to the \textsc{and}-clause are covered with a covering which makes the \textsc{and}-clause evaluate to {\sf true} ({\sf true}-covering if the variable appears non-negated in the clause, or {\sf false}-covering if the variable is negated in the clause); otherwise the minimum number of squares needed is $13$. Therefore, junctions act similar to \textsc{and}-gates. Let the total number of junctions be $j$. Then the paper concludes by stating that $\psi$ is satisfiable if and only if the reduced instance can be covered with less than $(V + W + 13j)$ squares; otherwise exactly $(V + W + 13j)$ squares are required. Hence solving \OPCSH would also solve \textsc{Planar 3-CNF}.

We take a deeper look into this reduction by performing the following case work on the behaviour of $\psi$ on various truth assignments of the variables.

\begin{itemize}

    \item \textbf{Case I, $\psi$ is {\sf true} for all assignments}. This implies $\phi = \neg \psi$ is {\sf false} for all assignments. Since $\phi$ is in DNF, all assignments must render \emph{all} \textsc{and}-clauses as {\sf false}. Therefore, for all minimum covering of the variable gadgets, coverings all junctions would require $13$ squares. This means the minimum number of squares needed to cover such an instance is $V + W + 13j$.

    \item \textbf{Case II, $\psi$ is {\sf false} for all assignments}. This implies $\phi = \neg \psi$ is {\sf true} for all assignments. Since $\phi$ is in DNF, all assignments must render \emph{at least one} \textsc{and}-clauses as {\sf true}. Therefore, for any arbitrary minimum covering of the variable gadgets, there would be at least one junction that can be covered with $12$ squares (while others take at most $13$ squares). This means the minimum number of squares needed to cover such an instance is strictly less than $V + W + 13j$.

    \item \textbf{Case III, $\psi$ is {\sf true} for some assignments and {\sf false} for some}. This implies $\phi = \neg \psi$ is also {\sf true} for some assignments and {\sf false} for some. Consider any assignment such that $\phi$ is {\sf true}. Since $\phi$ is in DNF, this assignments must render \emph{at least one} \textsc{and}-clauses as {\sf true}. Therefore, if we cover the variable gadgets corresponding to such an assignment, there would be at least one junction that can be covered with $12$ squares (while others take at most $13$ squares). This means the minimum number of squares needed to cover such an instance is strictly less than $V + W + 13j$.
\end{itemize}

Therefore, the minimum number of squares to cover the reduced instance is equal to $V + W + 13j$ if and only if $\psi$ is {\sf true} for all assignments (Case I) and less than $V + W + 13j$ otherwise (Case II \& Case III). Therefore, solving \OPCSH     solves tautology for $\psi$ instead of satisfiablity. Moreover, tautology in a CNF formula can be checked in linear time (Since $\psi$ is a tautology if and only if for all clauses $c$ in $\phi$ there is a variable $x$ such that both $x$ and $\neg x$ appear in $c$). Hence this reduction does not prove NP-hardness of \OPCSH.

\subsection{Fixing the construction}

The issue in this claimed NP-hardness reduction in arose as the clauses require more number of squares when the literals of a clause do not all evaluate to {\sf true} even if this assignment sets the clause to {\sf true}.

We construct such a junction that requires $29$ squares if \emph{all three} of its literals evaluate to {\sf false}, whereas the junctions require $28$ squares when at least one literal evaluates to {\sf true}. Therefore our modified junction behaves similar to an \textsc{or}-gate.

\paragraph*{Recapping the constructions of variable gadgets and wires.}

The construction due to Aupperle et. al~\cite{aupperle1988covering} starts by defining \emph{even-lines} (\emph{odd-lines}) as horizontal lines with an even (odd) y-coordinate or a vertical line with an even (odd) x-coordinate. The construction of variable gadgets and wires follow these properties:

\begin{itemize}

    \item All maximal squares of variable gadgets and wires are $2 \times 2$ in dimension. 
    
    \item The wires connecting a variable gadget for a variable $x$, to a junction for a clause $c$, where $x$ appears \emph{non-negated} in $c$ --- connect to the variable gadget (as well as the junction) horizontally along \emph{two consecutive odd lines} (\textit{i.e.} an even line passes through the middle of such a connection). Please refer to diagrams of the original construction~\cite{aupperle1988covering}.
    
    \item The wires connecting a variable gadget for a variable $x$, to a junction for a clause $c$, where $x$ appears \emph{negated} in $c$ --- connect to the variable gadget (as well as the junction) horizontally along \emph{two consecutive even lines} (\textit{i.e.} an odd line passes through the middle of such a connection). Please refer to diagrams of the original construction~\cite{aupperle1988covering}.

    \item If a variable gadget for a variable $x$ is covered with a {\sf true}-covering, then all wires connecting it to a junction for a clause $c$ where $x$ appears \emph{non-negated}, will be covered in a way such that half a square protrudes out of destination and hence covers a part of the junction. Otherwise, for a {\sf false}-covering, the entire junction would remain uncovered, even when the variable gadget and the wire are covered. Please refer to diagrams of the original construction~\cite{aupperle1988covering}.
    
    \item If a variable gadget for a variable $x$ is covered with a {\sf false}-covering, then all wires connecting it to a junction for a clause $c$ where $x$ appears \emph{negated}, will be covered in a way such that half a square protrudes out of destination and hence covers a part of the junction. Otherwise, for a {\sf true}-covering, the entire junction would remain uncovered, even when the variable gadget and the wire are covered. Please refer to diagrams of the original construction~\cite{aupperle1988covering}.
    
\end{itemize}

We will use these construction for variable gadgets as it is, but we provide novel constructions for junction gadgets as mentioned earlier.

\paragraph*{Modifying the junction gadgets.}

We construct new junction gadgets as shown in Figure~\ref{fig:jn-gadgets}. Wires connect to these junction gadgets in the same way as described in the original construction~\cite{aupperle1988covering}.

\begin{figure} [ht]
    \hfill 
    \begin{subfigure}[t]{0.47\textwidth} 
        \centering 
        \includegraphics[width=65mm]{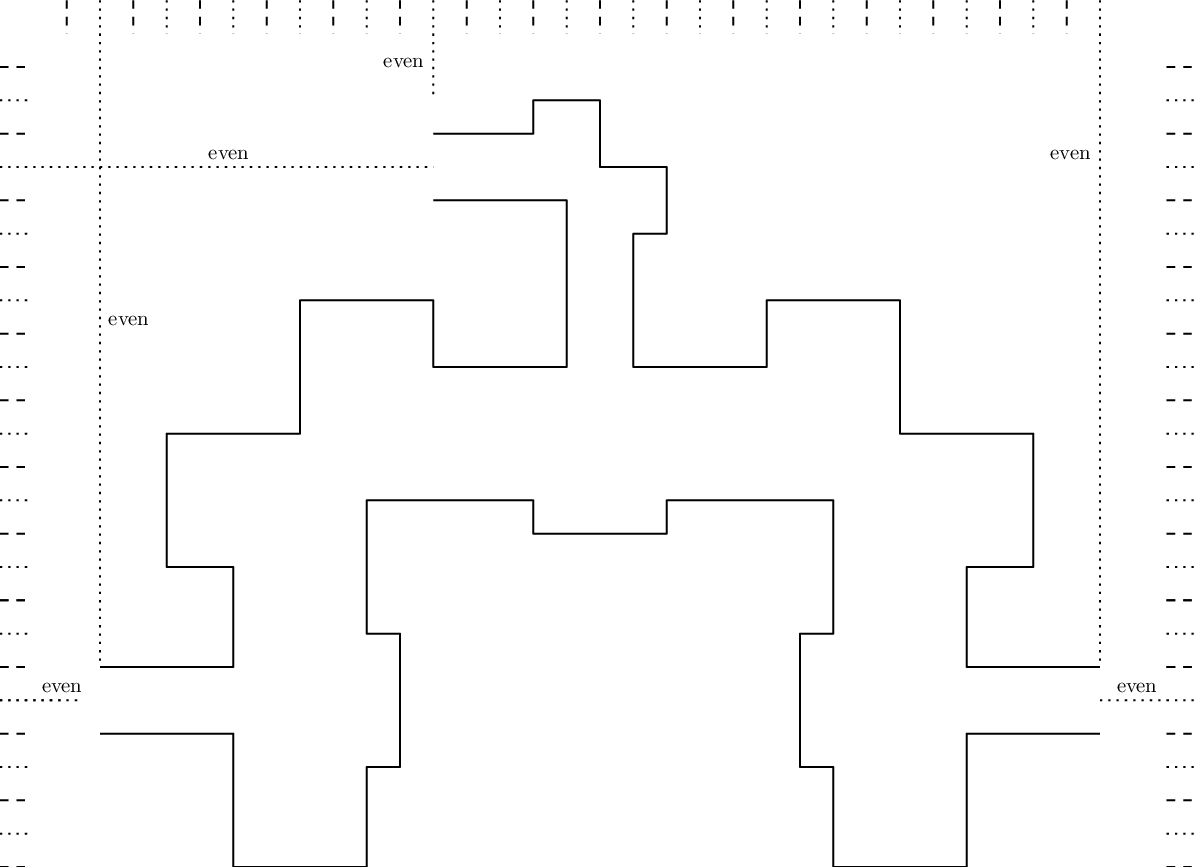}
        \caption{Junction for clauses of type $(x \vee y \vee z)$}\label{fig:jn-gadgets-a}
    \end{subfigure} \hfill
    \begin{subfigure}[t]{0.47\textwidth}
        \centering
        \includegraphics[width=65mm]{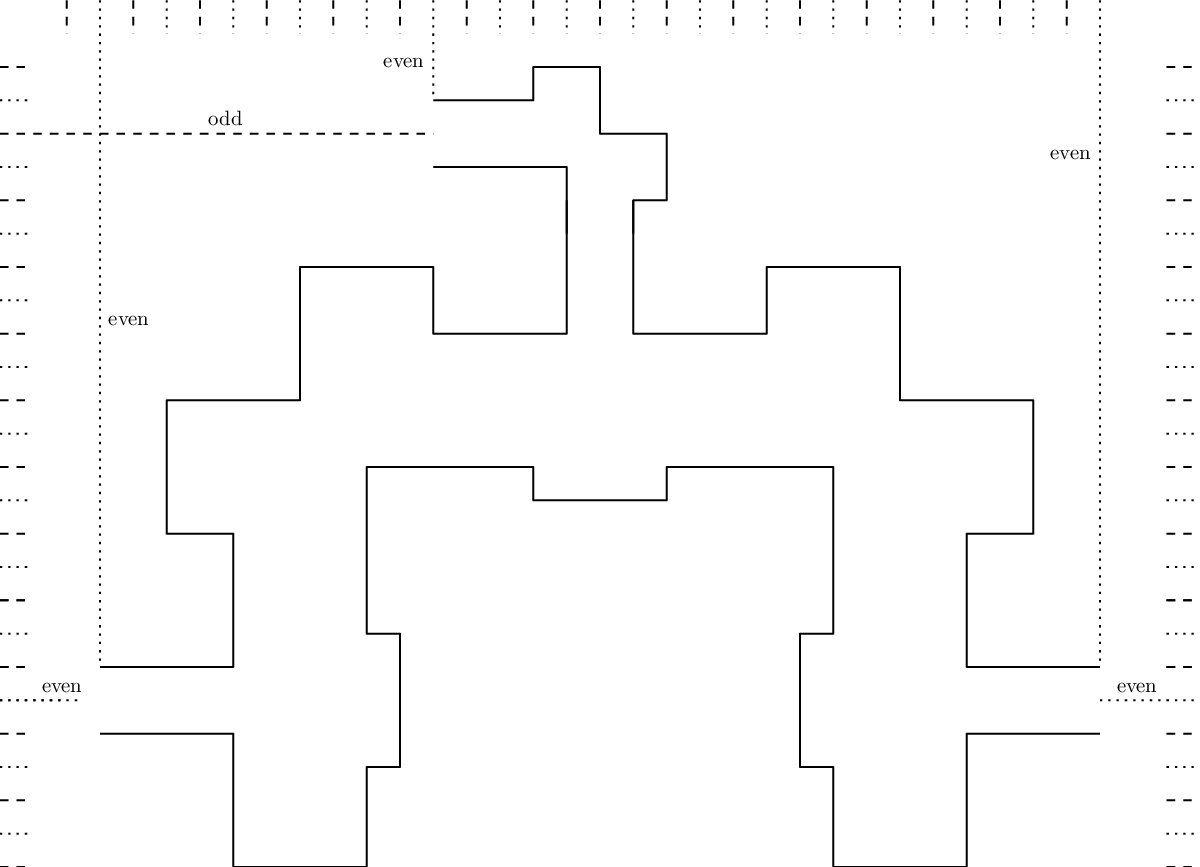}
        \caption{Junction for clauses of type $(x \vee \neg y \vee z)$}\label{fig:jn-gadgets-b}
    \end{subfigure} \hfill 
    \caption{Modified junction}\label{fig:jn-gadgets} 
\end{figure}

The junctions for clauses of type $(\neg x \vee \neg y \vee \neg z)$ can be constructed by shifting the gadget in Figure~\ref{fig:jn-gadgets-a} vertically up by one square. Similarly the junctions for clauses of type $(\neg x \vee y \vee \neg z)$ can be constructed by shifting the gadget in Figure~\ref{fig:jn-gadgets-b} vertically up by one square.

We now investigate minimal coverings of squares for each of these gadgets. For the remaining paper, we denote by the term `literal', a variable or its negation that appears in a clause.

\begin{lemma}\label{lem:new-jn}
    Each junction gadget requires $28$ squares to be covered when at least one literal in it evaluates to {\sf true}. Otherwise the junction gadget requires $29$ squares to be covered.
\end{lemma}

\begin{proof}
The following proof works for both the gadgets in Figure~\ref{fig:jn-gadgets-a} and Figure~\ref{fig:jn-gadgets-b}. Both configurations are such that there are nine $4 \times 4$ maximal squares due to convex vertices present due, and eleven $2 \times 2$ squares necessary for joining with wire gadgets (\cref{rem:mcs-always}). Other than these $20$ squares, we analyse other valid squares needed to cover it fully. 

In both gadgets, we can find eight blocks ($1 \times 1$ regions) $p_1, \ldots, p_8$, as shown in Figure~\ref{fig:jn-8}, such that these are not covered by any of the $20$ previously placed squares (irrespective of the {\sf true}/{\sf false} values carried through the wires); and for $i \ne j$, $p_i$ and $p_j$ cannot be covered by a single valid square. This means at least $8$ more valid squares are required to cover the entire gadget. This means at least $28$ squares are required to cover the entire gadget (for any truth values being carried through the wires).  

\begin{figure} [ht]
    \hfill 
    \begin{subfigure}[t]{0.47\textwidth} 
        \centering 
        \includegraphics[width=65mm]{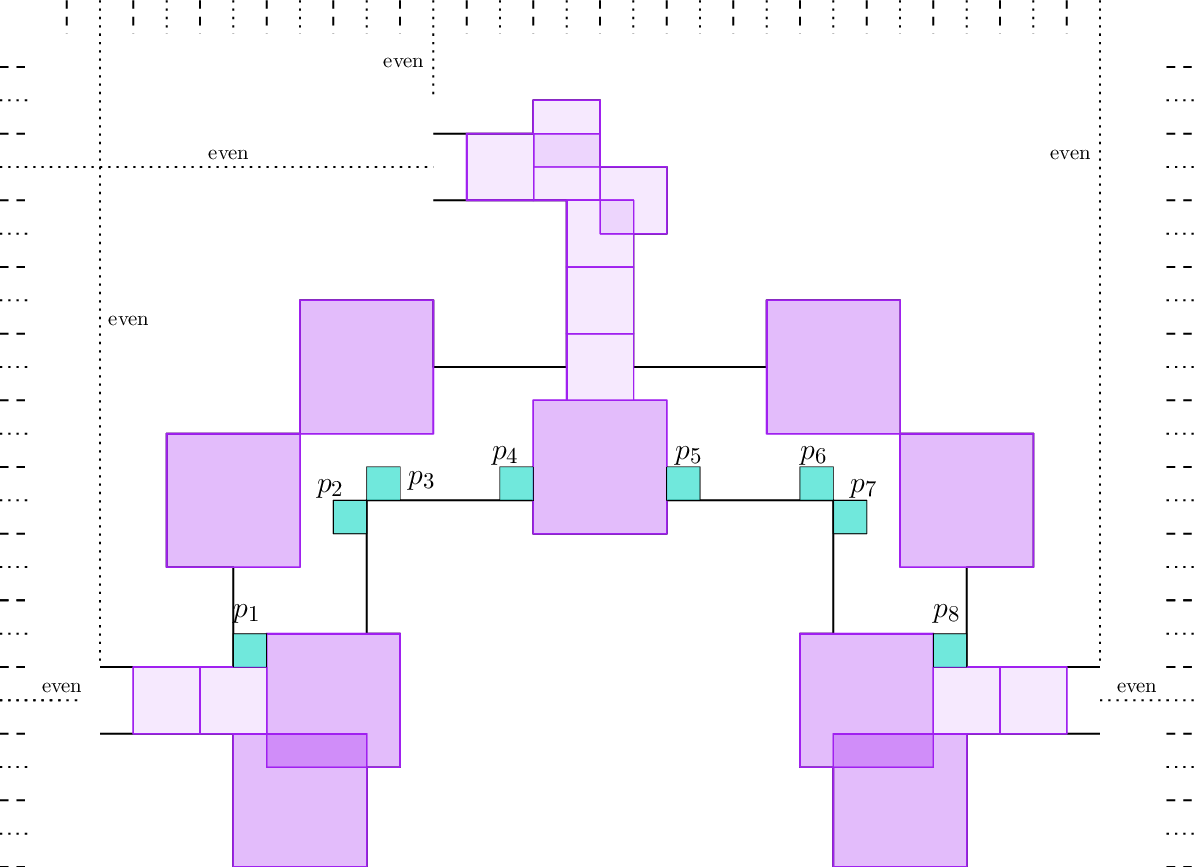}
        \caption{$p_1, \ldots, p_8$ for Figure~\ref{fig:jn-gadgets-a}}
    \end{subfigure} \hfill
    \begin{subfigure}[t]{0.47\textwidth}
        \centering
        \includegraphics[width=65mm]{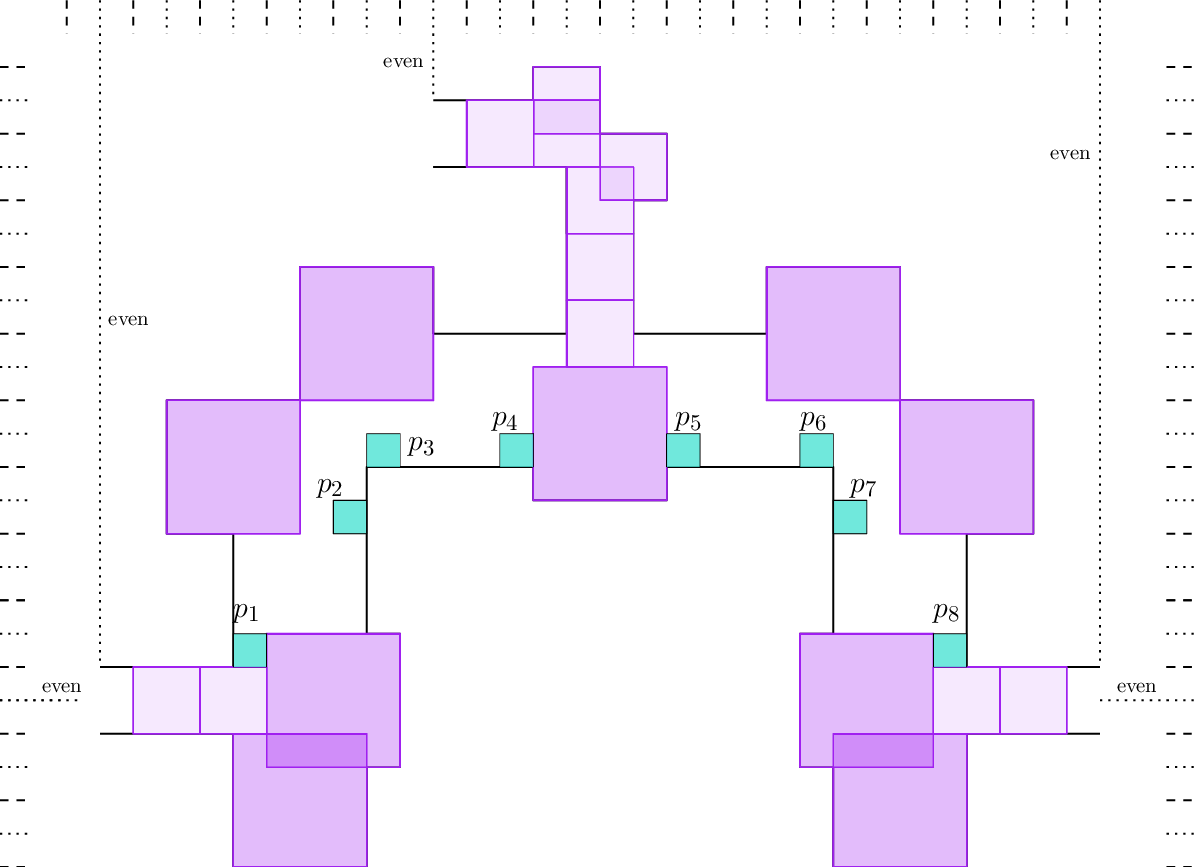}
        \caption{$p_1, \ldots, p_8$ for Figure~\ref{fig:jn-gadgets-b}}
    \end{subfigure} \hfill 
    \caption{Positions of $p_1, \ldots, p_8$}\label{fig:jn-8}
\end{figure}

\begin{figure} [ht]
    \hfill 
    \begin{subfigure}[t]{0.47\textwidth} 
        \centering 
        \includegraphics[width=65mm]{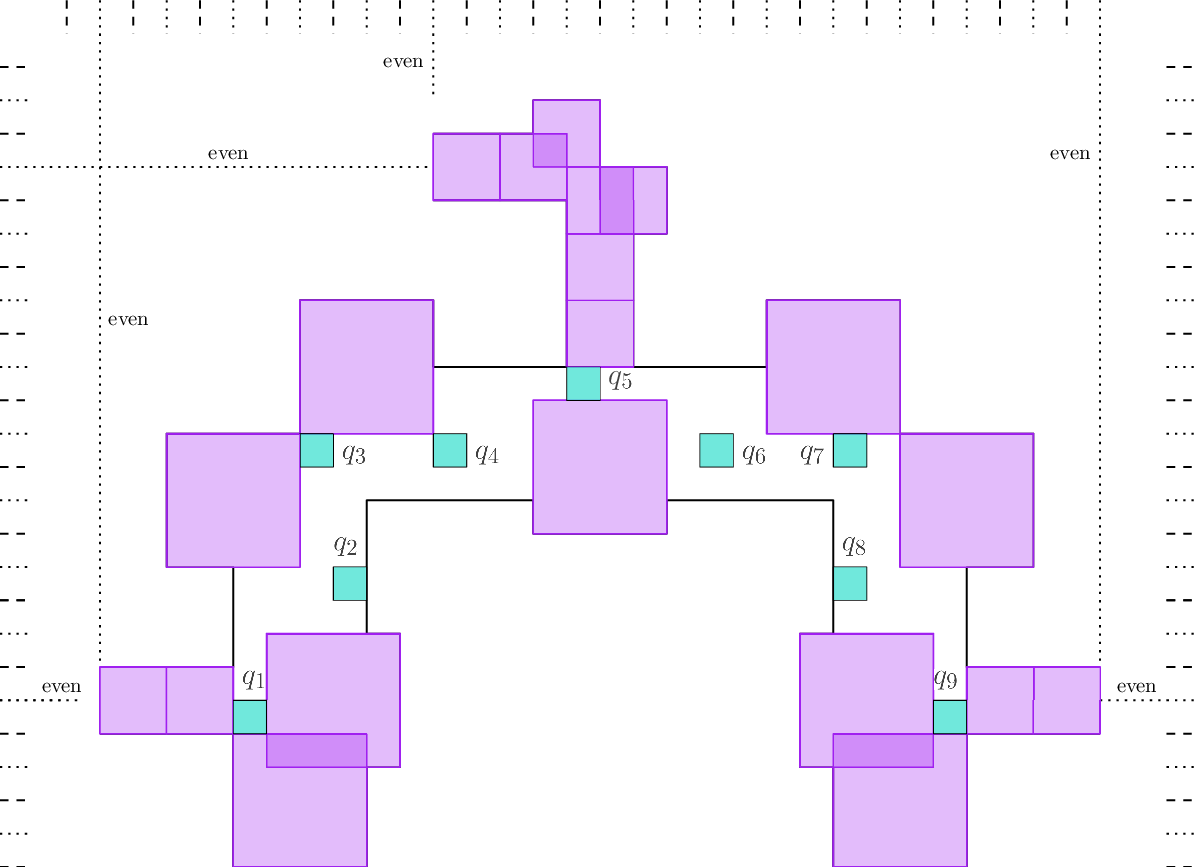}
        \caption{$q_1, \ldots, q_9$ for Figure~\ref{fig:jn-gadgets-a}}
    \end{subfigure} \hfill
    \begin{subfigure}[t]{0.47\textwidth}
        \centering
        \includegraphics[width=65mm]{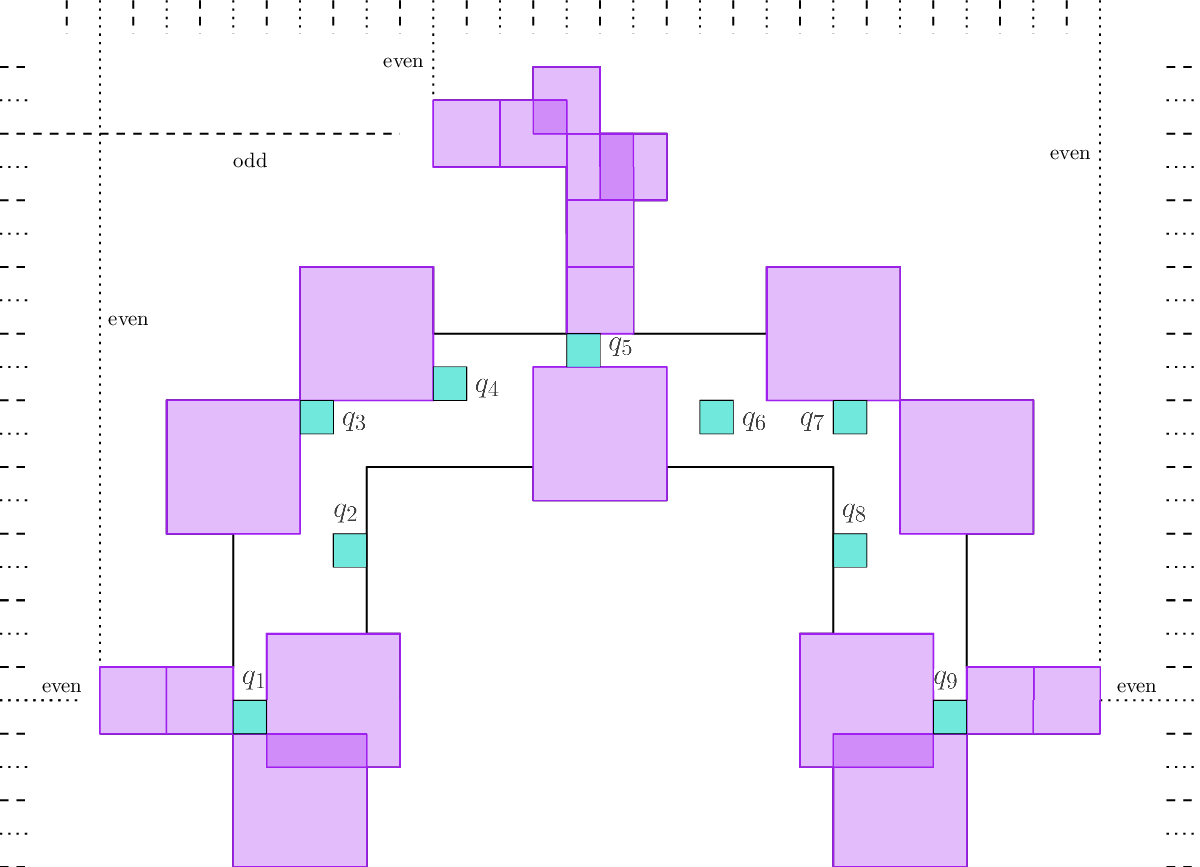}
        \caption{$q_1, \ldots, q_9$ for Figure~\ref{fig:jn-gadgets-b}}
    \end{subfigure} \hfill 
    \caption{Positions of $q_1, \ldots, q_9$}\label{fig:jn-9}
\end{figure}

Further, if all wires contain a {\sf false} value, we can find nine blocks ($1 \times 1$ regions) $q_1, \ldots, q_9$, as shown in Figure~\ref{fig:jn-9}, such that these are not covered by any of the $20$ previously placed squares; and for $i \ne j$, $q_i$ and $q_j$ cannot be covered by a single valid square. This means at least $9$ more valid squares are required to cover the entire gadget. Therefore at least $29$ squares are required to cover the entire gadget when all wires contain a {\sf false} value.  

In Figure~\ref{fig:proof-jn-a} and Figure~\ref{fig:proof-jn-b}, we indeed construct a set of $29$ valid squares covering the entire gadget when all three wires carry a {\sf false} value, and a set of $28$ valid squares covering the entire gadget when at least one of the wires contain a {\sf true} value. This construction proves that indeed $29$ squares and $28$ squares are precisely the minimum number of valid covering squares required to cover the gadgets if all literals are {\sf false}, and at least one literal is {\sf true}, respectively. This completes the proof.

\begin{figure} [htbp]
    \hfill 
    \begin{subfigure}[t]{0.47\textwidth} 
        \centering 
        \includegraphics[width=65mm]{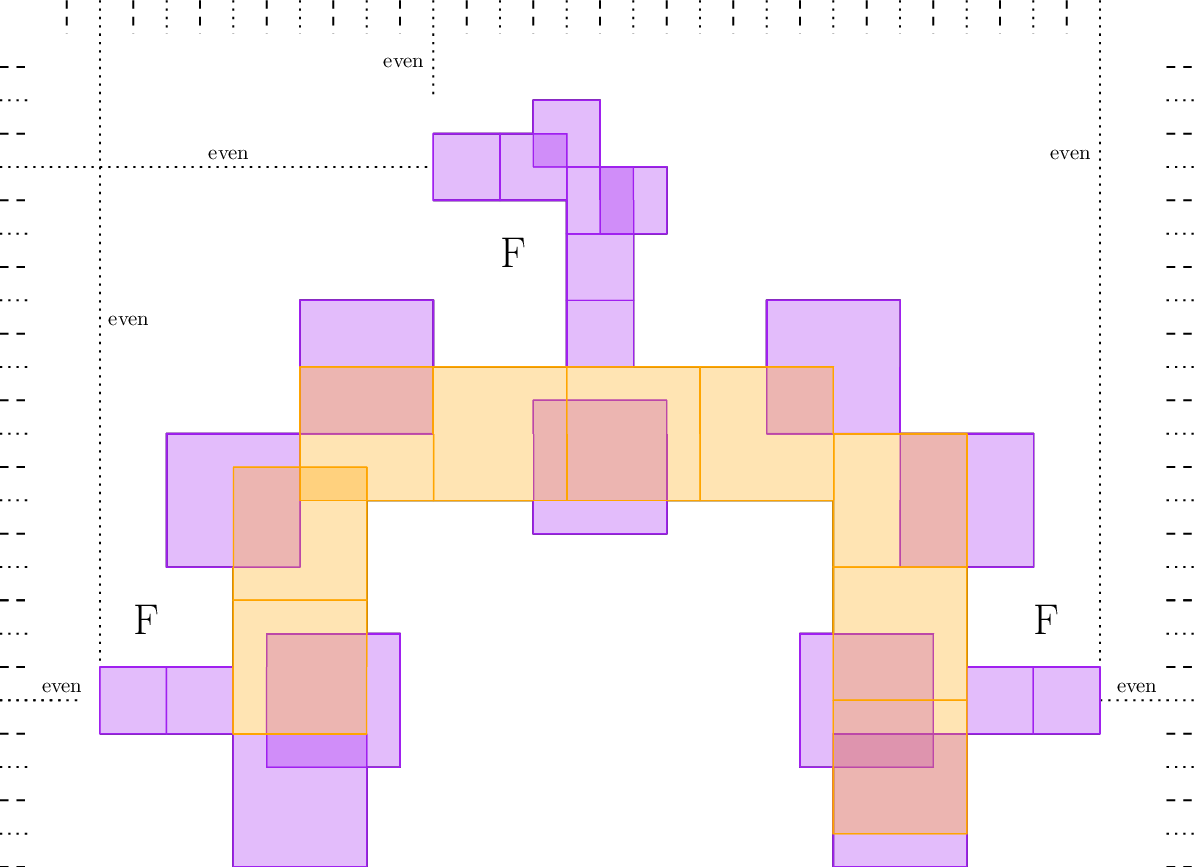}
        \caption{$29$ squares when literals are $(F,F,F)$}
    \end{subfigure} \hfill
    \begin{subfigure}[t]{0.47\textwidth}
        \centering
        \includegraphics[width=65mm]{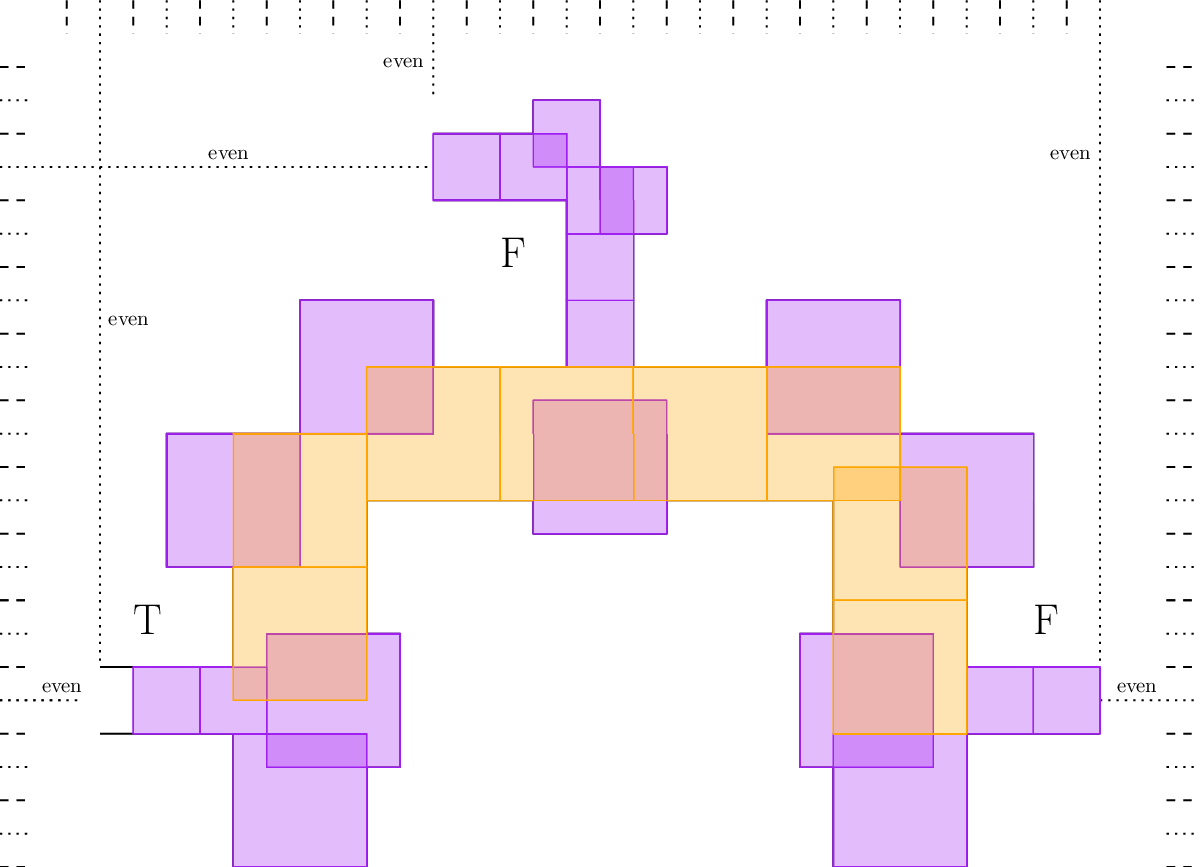}
        \caption{$28$ squares when literals are $(F, F, T)$ or $(T, F, F)$}
    \end{subfigure} \hfill \\
    \begin{subfigure}[t]{0.47\textwidth} 
        \centering 
        \includegraphics[width=65mm]{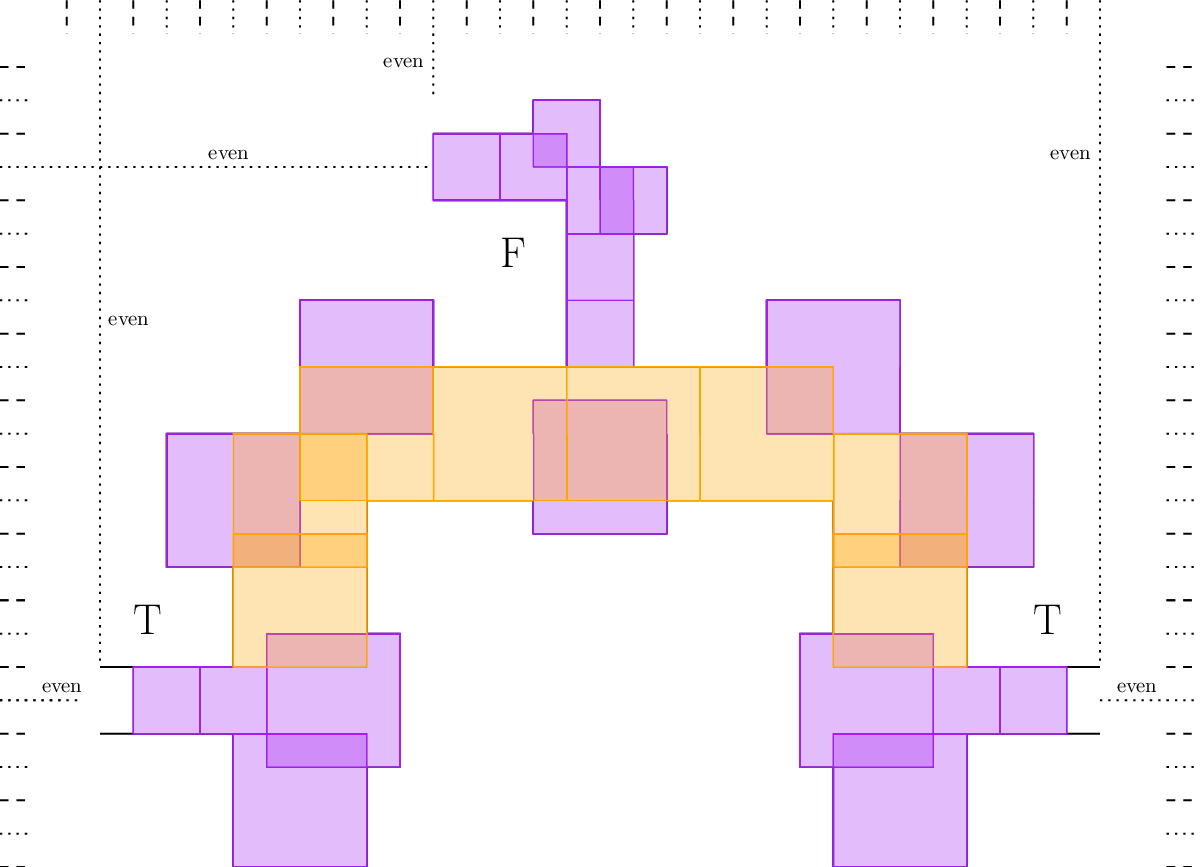}
        \caption{$28$ squares when literals are $(T,F,T)$}
    \end{subfigure} \hfill
    \begin{subfigure}[t]{0.47\textwidth}
        \centering
        \includegraphics[width=65mm]{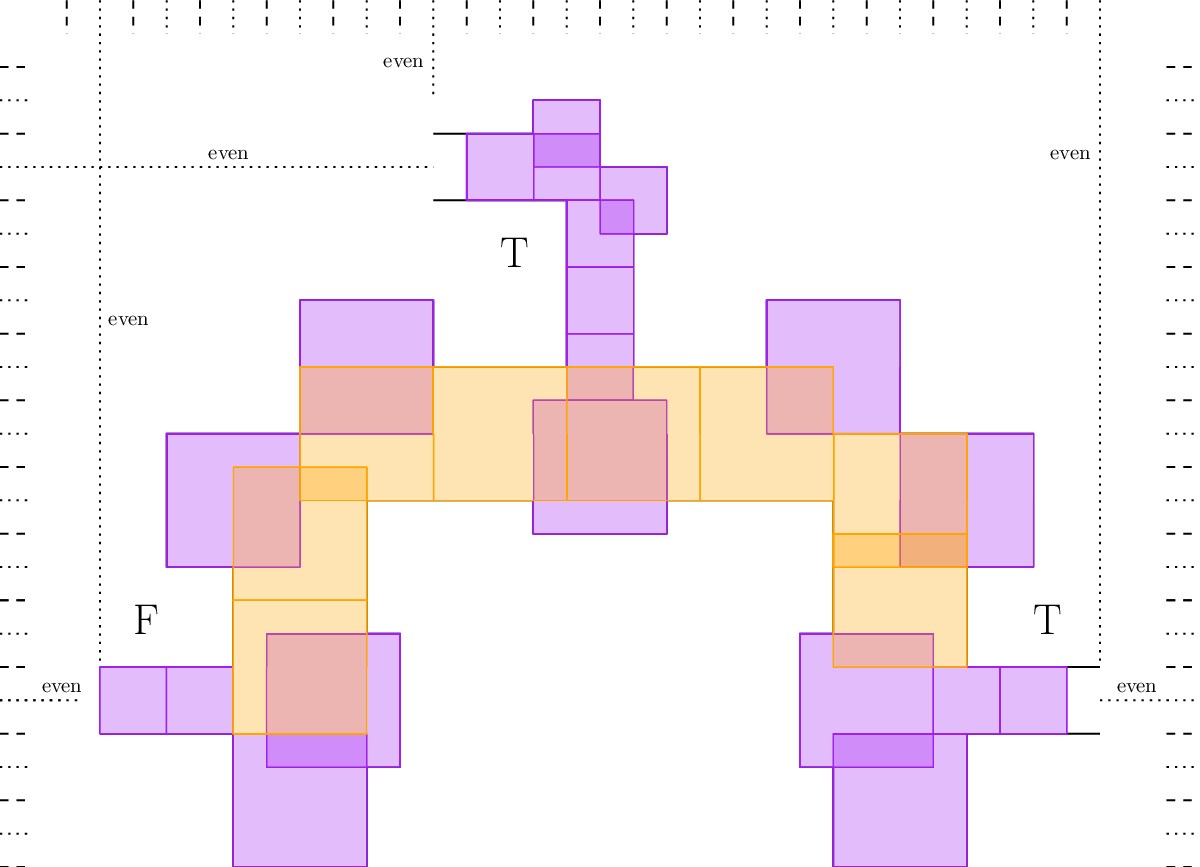}
        \caption{$28$ squares when literals are $(F, T, T)$ or $(T, F, F)$}
    \end{subfigure} \hfill \\
    \begin{subfigure}[t]{0.47\textwidth} 
        \centering 
        \includegraphics[width=65mm]{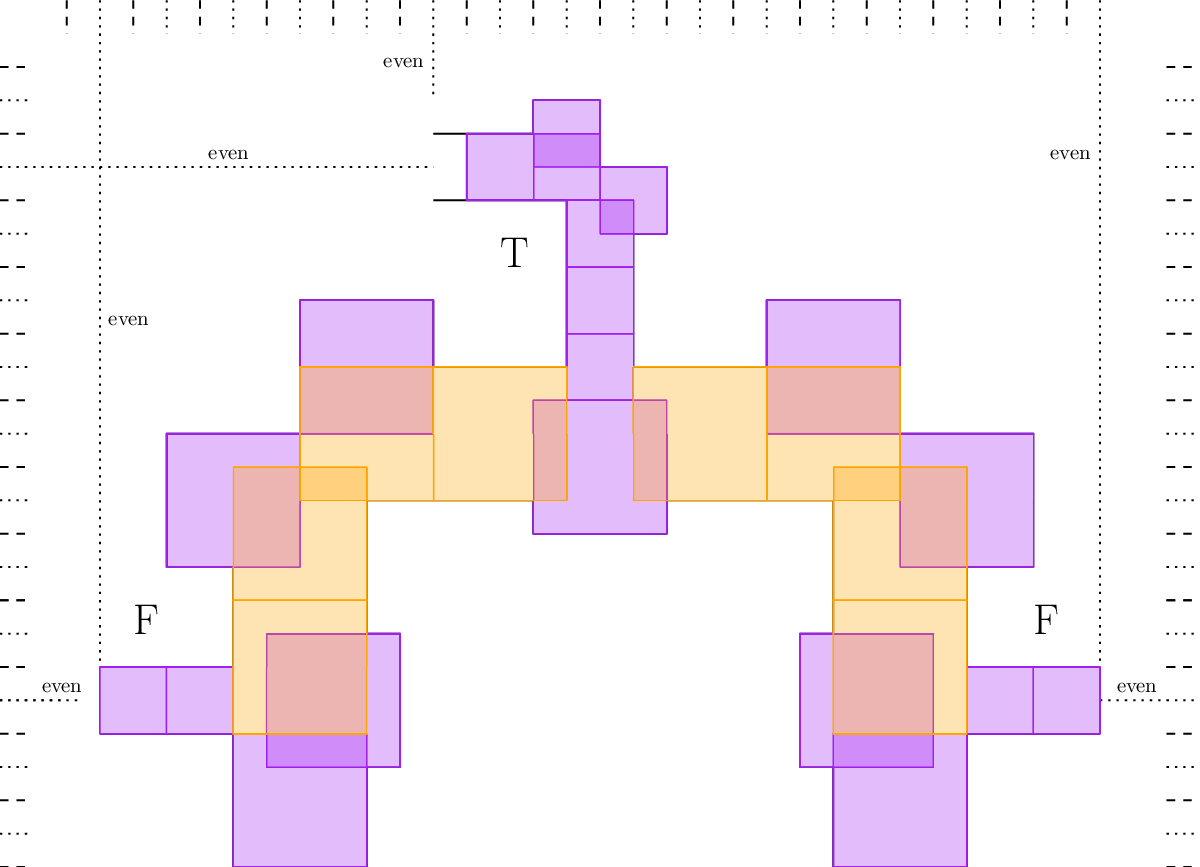}
        \caption{$28$ squares when literals are $(F,T,F)$}
    \end{subfigure} \hfill
    \begin{subfigure}[t]{0.47\textwidth}
        \centering
        \includegraphics[width=65mm]{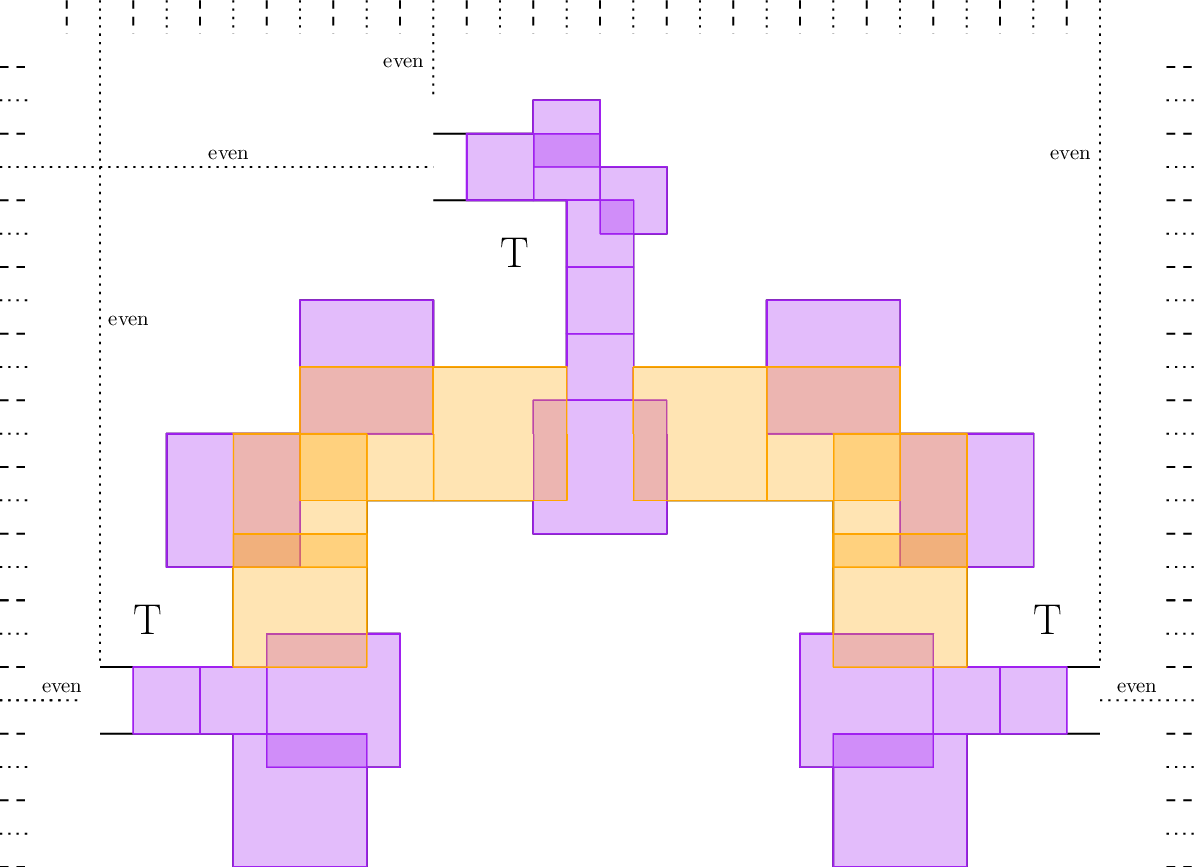}
        \caption{$28$ squares when literals are $(T, T, T)$}
    \end{subfigure} \hfill 
    \caption{Minimal covering of junctions in Figure~\ref{fig:jn-gadgets-a}}\label{fig:proof-jn-a} 
\end{figure}

\begin{figure} [htbp]
    \hfill 
    \begin{subfigure}[t]{0.47\textwidth} 
        \centering 
        \includegraphics[width=65mm]{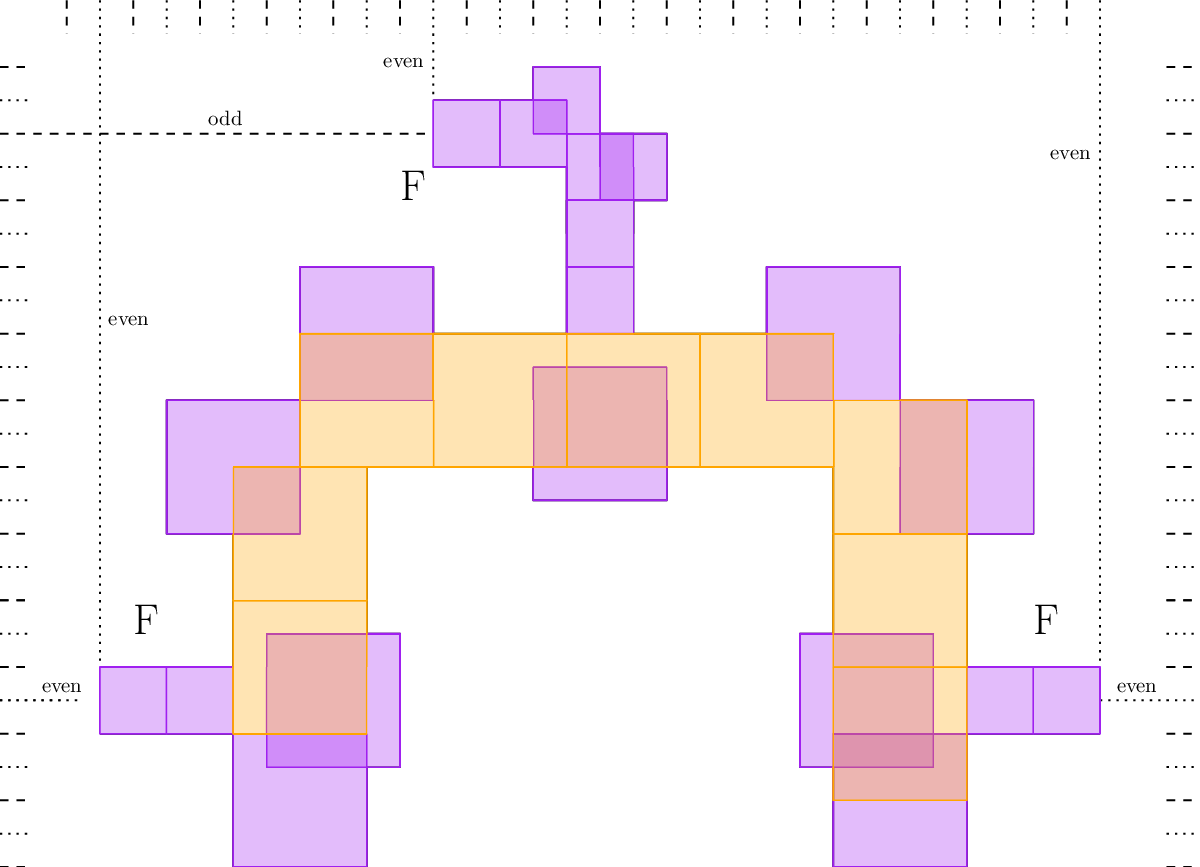}
        \caption{$29$ squares when literals are $(F,F,F)$}
    \end{subfigure} \hfill
    \begin{subfigure}[t]{0.47\textwidth}
        \centering
        \includegraphics[width=65mm]{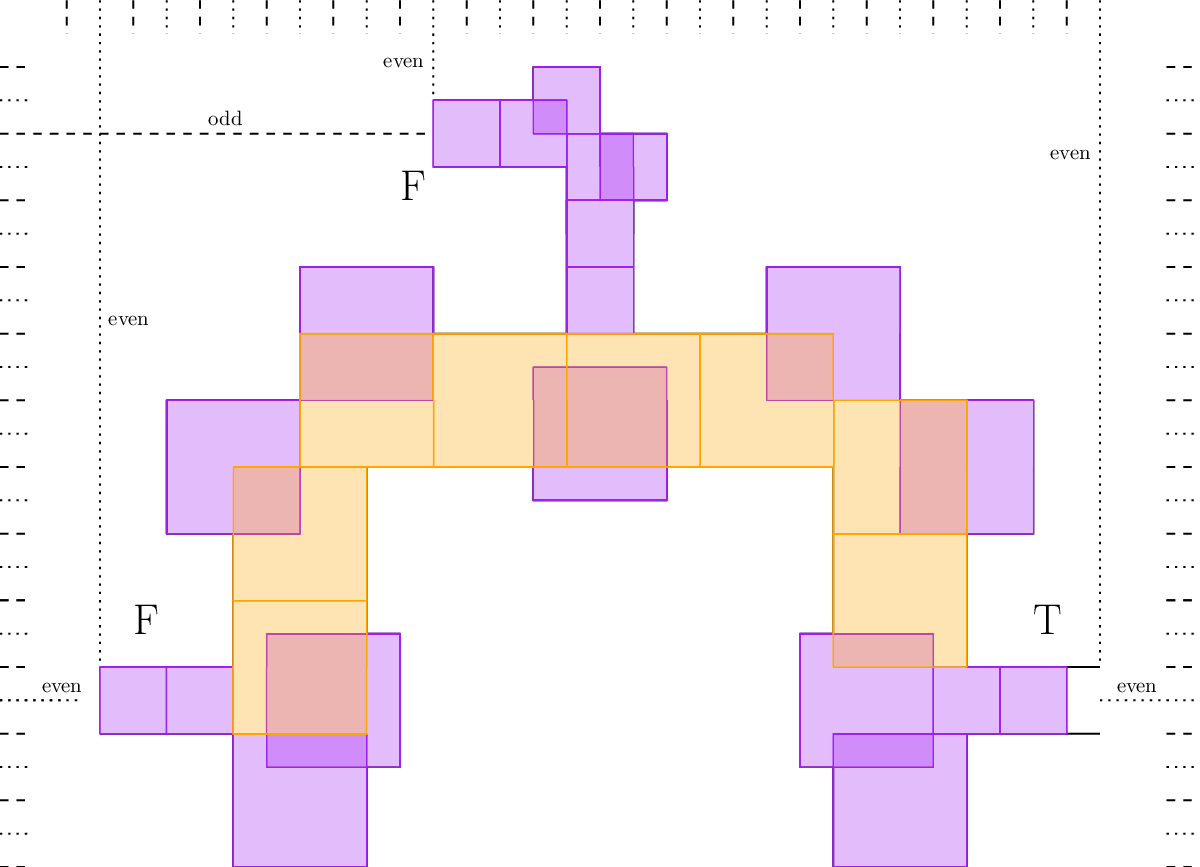}
        \caption{$28$ squares when literals are $(F, F, T)$ or $(T, F, F)$}
    \end{subfigure} \hfill \\
    \begin{subfigure}[t]{0.47\textwidth} 
        \centering 
        \includegraphics[width=65mm]{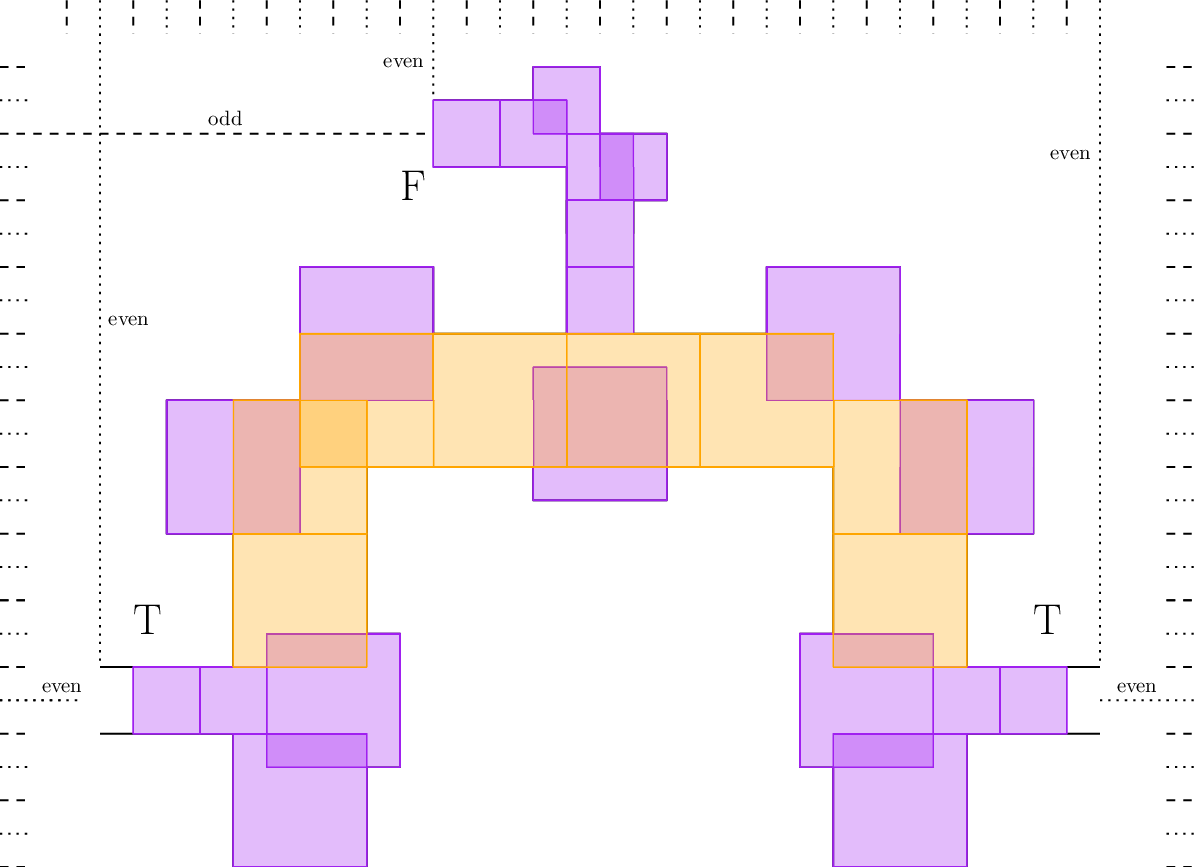}
        \caption{$28$ squares when literals are $(T,F,T)$}
    \end{subfigure} \hfill
    \begin{subfigure}[t]{0.47\textwidth}
        \centering
        \includegraphics[width=65mm]{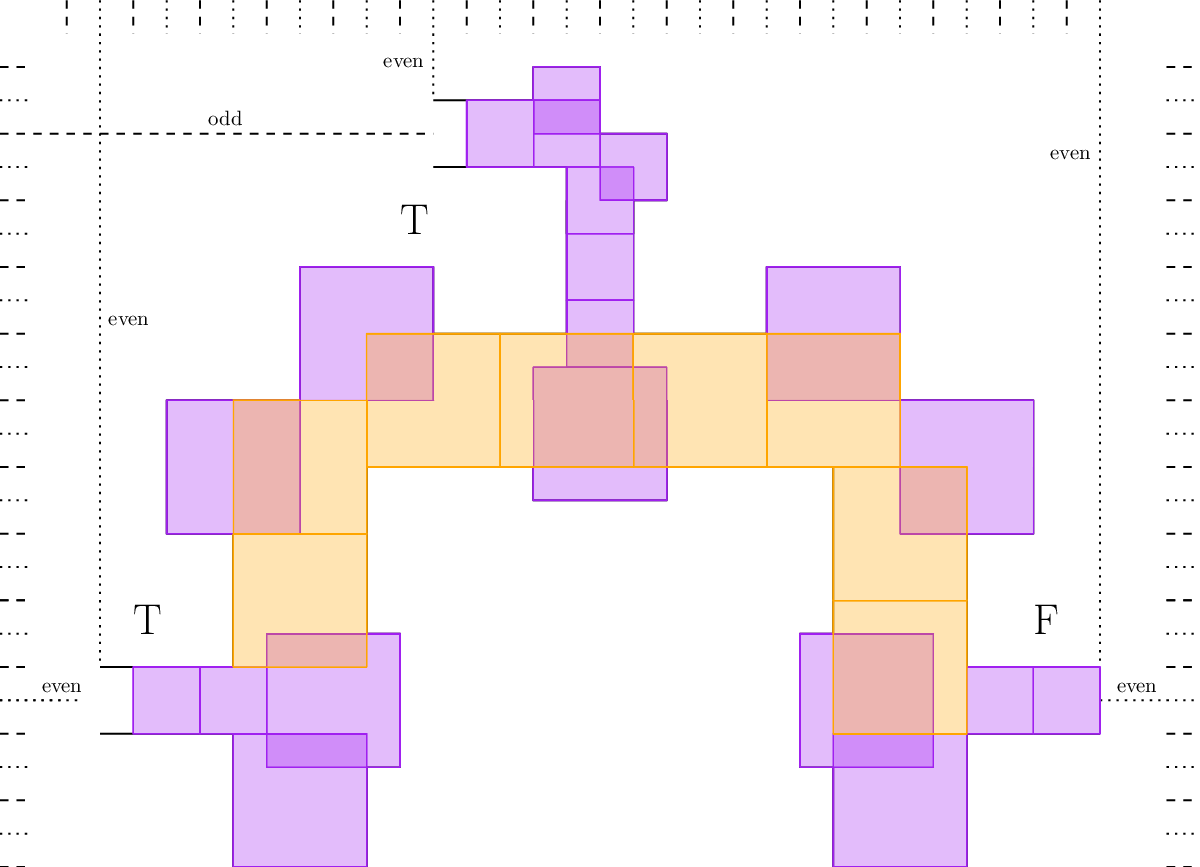}
        \caption{$28$ squares when literals are $(F, T, T)$ or $(T, F, F)$}
    \end{subfigure} \hfill \\
    \begin{subfigure}[t]{0.47\textwidth} 
        \centering 
        \includegraphics[width=65mm]{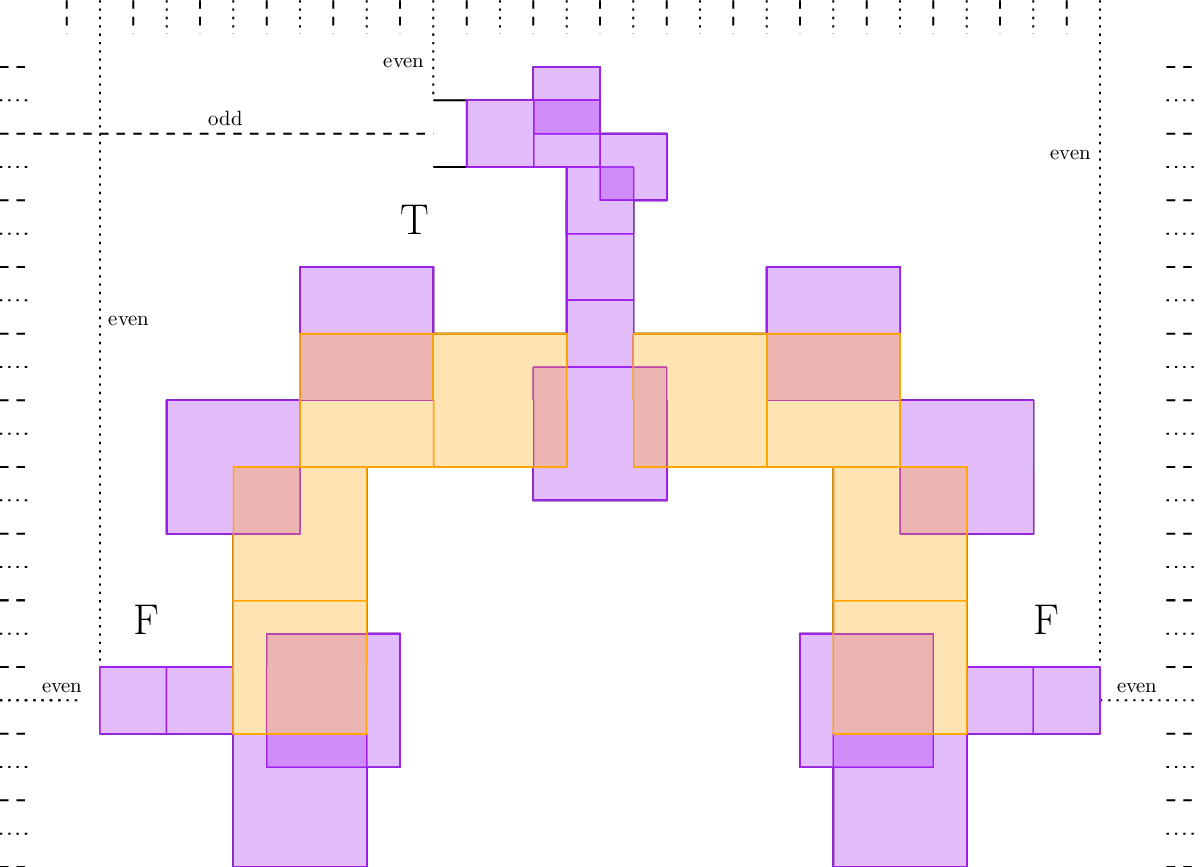}
        \caption{$28$ squares when literals are $(F,T,F)$}
    \end{subfigure} \hfill
    \begin{subfigure}[t]{0.47\textwidth}
        \centering
        \includegraphics[width=65mm]{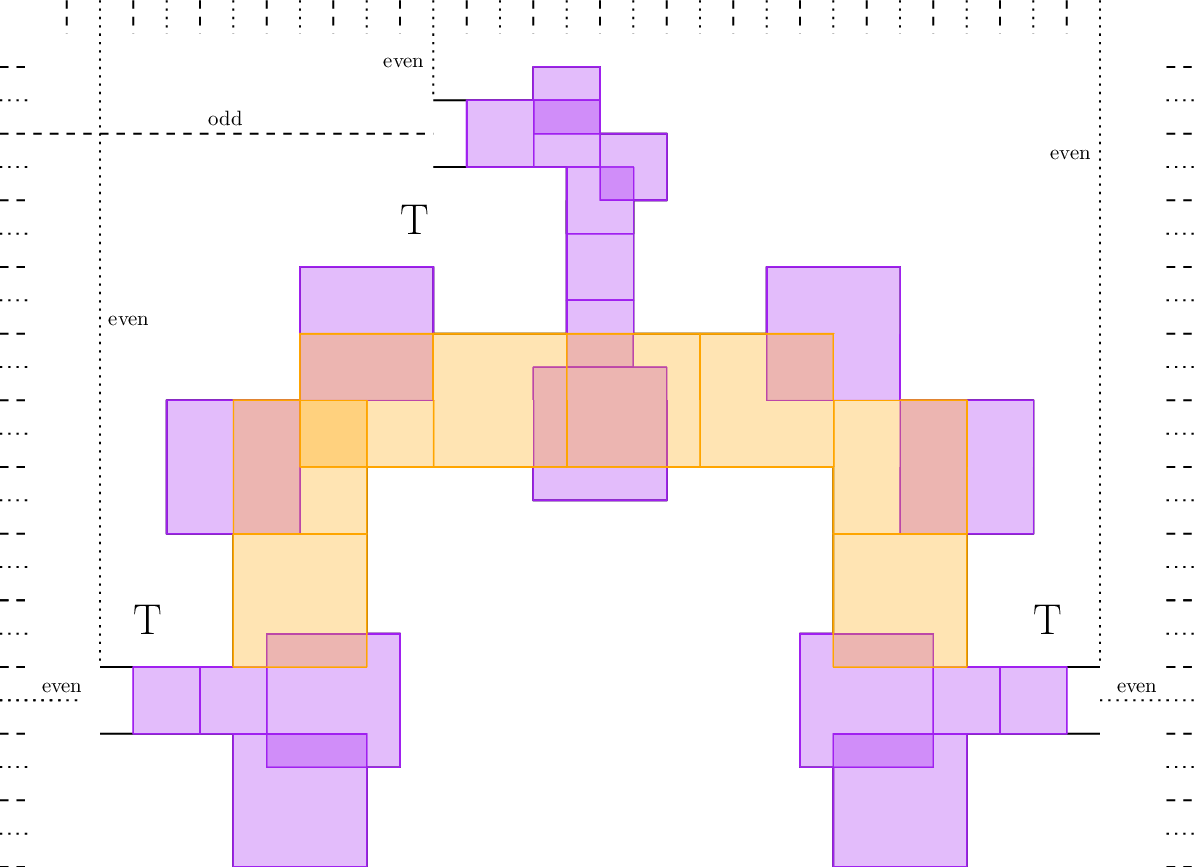}
        \caption{$28$ squares when literals are $(T, T, T)$}
    \end{subfigure} \hfill 
    \caption{Minimal covering of junctions in Figure~\ref{fig:jn-gadgets-b}}\label{fig:proof-jn-b} 
\end{figure}

\end{proof}

Now, we complete the reduction from \textsc{Planar 3-CNF}~\cite{doi:10.1137/0211025}.

\paragraph*{Completing the reduction.}

We consider a CNF boolean formula $\psi$, which is an instance of \textsc{Planar 3-CNF}. We construct variable gadgets for each variable in $\psi$, and set up junctions (\cref{fig:jn-gadgets}) for each clause in $\psi$ and connect them with wires. Let $V + W$ be the squares required to cover all variable gadgets and wires. Note that we are working with $\psi$ directly, instead of working with $\phi = \neg \psi$ as done by Aupperle et. al~\cite{aupperle1988covering}.

Now, we analyse the three cases like before. Let $j$ be the number of junctions. Note that, due to \cref{lem:new-jn}, the minimum number of squares required to cover the reduced instance is at least $V + W + 28j$.

\begin{itemize}

    \item \textbf{Case I, $\psi$ is {\sf true} for all assignments (satisfiable)}. Since $\psi$ is in CNF, all assignments must render all \textsc{or}-clauses as {\sf true}; \textsl{i.e.} at least one literal in each clause is {\sf true}. Therefore from \cref{lem:new-jn}, for all minimum covering of the variable gadgets, coverings \emph{all} junctions would require a minimum of only $28$ valid squares. This means the minimum number of squares needed to cover such an instance is \emph{equal to} $V + W + 28j$.

    \item \textbf{Case II, $\psi$ is {\sf false} for all assignments (not satisfiable)}. Since $\psi$ is in CNF, all assignments must render \emph{at least one} \textsc{or}-clause as {\sf false}; \textsl{i.e.} all literals in at least one clause is {\sf false}. Therefore from \cref{lem:new-jn}, \emph{at least one} junction requires a minimum of $29$ valid squares, whereas others require at least $28$ valid squares for being covered. This means the minimum number of squares needed to cover such an instance is \emph{strictly more than} $V + W + 28j$.

    \item \textbf{Case III, $\psi$ is {\sf true} for some assignments and {\sf false} for some (satisfiable)}. Consider any assignment such that $\psi$ is {\sf true}. Since $\psi$ is in CNF, this assignments must render \emph{all} \textsc{or}-clauses as {\sf true}; \textsl{i.e.} at least one literal in each clause is {\sf true}. Consider the square covering corresponding to this assignment. Due to \cref{lem:new-jn}, for such a minimum covering of the variable gadgets, coverings \emph{all} junctions would require exactly $28$ squares. This means the minimum number of squares needed to cover such an instance is \emph{equal to} $V + W + 28j$.
\end{itemize}

This implies that $\psi$ is satisfiable if and only if the reduced instance has a minimum covering with exactly $V + W + 28j$ squares. On the other hand, if $\psi$ is not satisfiable, then the minimum number of squares required to cover the reduced instance is strictly more than $V + Q + 28j$. Since \textsc{Planar 3-CNF} is NP-hard, \OPCSH must also be NP-hard. Finally, as any certificate can be verified in polynomial time (\cref{lem:check-algo}), \OPCSH is NP-complete.

\begin{theorem} \label{thm:opcsh-np-hard}
    The problem of covering orthogonal polygons with holes using minimum number of squares (\OPCSH) is NP-complete, where the input is the set of all $N$ lattice points inside the orthogonal polygon.
\end{theorem}

Moreover, as both the answer $V+W+28j$ and the numerical value of the coordinates of each vertex is linear in the number $N$ of lattice points inside the polygon, the problem is strongly NP-complete, and hence cannot have a fully polynomial time approximation scheme (FPTAS)~\cite{garey1979computers, vazirani1997approximation}.

\begin{corollary}
    \OPCSH is strongly NP-complete. Therefore, there is no FPTAS scheme for \OPCSH unless \pnp. That is, there is no $(1 + \varepsilon)$-approximation scheme for an arbitrarily small $\varepsilon$, such that runs in time $\mathcal O\left(\mathsf{poly}\left(\frac{1}{\varepsilon}, N\right)\right)$, where $N$ is the number of lattice points inside the polygon.
\end{corollary}

Due to the structure of the reduced instance, we can also infer the following.

\begin{corollary} The following problems are NP-complete:
    \begin{itemize}
        \item The problem of finding a minimum square covering of polygons with holes, where all squares are restricted to have a side-length of at most $\eta$ is NP-complete whenever $\eta \ge 4$.
        \item The problem of finding the minimum square covering of polygons, where the squares can take side lengths from a set $\Lambda \subseteq \mathbb{N}$ is NP-complete, even if $|\Lambda| = 2$.
    \end{itemize}

    \begin{proof}
        The results follow from the observation that the reduced instance from \textsc{Planar 3-CNF} only contains maximal squares of size $2\times2$ and $4\times4$, and from Observation~\ref{prop:max-sq-cov}.
    \end{proof}
\end{corollary} 

Another immediate corollary of \cref{thm:opcsh-np-hard} is that fact that \OPCSH is NP-hard when the $n$ vertices of the polygon constitute the input, instead of the $N$ lattice points lying inside. This is because the reduced instance had $n = \Theta(N)$. It turns out that \OPCSH is in fact in the class NP, with $n$ vertices as input. This gives us the following result.

\begin{corollary}
    \OPCSH is NP-complete when the input is the $n$ vertices. 
\end{corollary}
\begin{proof}
    \cref{thm:opcsh-np-hard} immediately gives us that \OPCSH is NP-hard when the input is the $n$ vertices, as the reduced instance of orthogonal polygon has $N = \Theta(n)$ lattice points inside.
    
    To prove that \OPCSH with $n$ vertices is in NP, we need to show a polynomial time verifier for a certificate. We consider the set of rec-packs to be a certificate, and we can do the following to check if these indeed form a valid covering:
    \begin{itemize}
        \item Check if each rec-pack lies inside the polygon. This can be checked in $\mathcal O(n)$ by checking if an arbitrary point in the rec-pack lies in the polygon, followed by ensuring that no polygon-edge intersects a rec-pack edge non-trivially.
        \item If all rec-packs lie in the polygon, we can use \cref{lem:check-algo} to check if they form a complete covering of the polygon. Note that \cref{lem:check-algo} generalizes to polygons with holes as well. 
    \end{itemize}
    This proves there there is in fact a polynomial time verifier for \OPCSH.
\end{proof}

\section{Conclusion}\label{sec:conclusion}

In this paper, we answer the open problem, initially posed by Aupperle et. al~\cite{aupperle1988covering}, asking whether \OPCS has an exact algorithm running in time polynomial in the number $n$ of vertices of the input orthogonal polygon without holes. Our exact algorithm for \OPCS runs in $\OO(n^{10})$ time. We further optimize the running time for orthogonal polygons with $n$ vertices and a small number of knobs $k$, by designing a recursive algorithm with running time $\OO(n^2 + k^{10} \cdot n)$. This gives us an $\OO(n^2)$ algorithm for solving \OPCS on orthogonally convex polygons. We also provide a correct proof for the NP-hardness of \OPCSH; this is a novel result as the proof claimed in the works of Aupperle et. al~\cite{aupperle1988covering} incorrectly reduce from a polynomial-time solvable problem. A natural future direction is to study the problem with respect to covering of orthogonal polygons with other geometric objects like triangles, line segments, orthogonally convex polygons etc. and try to obtain an algorithm whose running time is polynomial in the number of vertices of the input polygon and not on the total number of interior lattice points of the input polygon. Another interesting question directly related to our current work would be to find tight bounds for the number of rec-packs required to produce a minimum covering of an orthogonal polygon; our paper only proves an upper bound of $\OO(n^2)$ rec-packs required to produce a minimum covering of an entire polygon. 
\bibliography{refs}

\begin{thebibliography}{10}

\bibitem{aamand2023}
Anders Aamand, Mikkel Abrahamsen, Thomas~D. Ahle, and Peter M.~R. Rasmussen.
\newblock Tiling with squares and packing dominos in polynomial time.
\newblock {\em ACM Trans. Algorithms}, 19(3), July 2023.
\newblock \href {https://doi.org/10.1145/3597932} {\path{doi:10.1145/3597932}}.

\bibitem{doi:10.1137/0602026/albertson}
Michael~O. Albertson and Claire~J. O’Keefe.
\newblock Covering regions with squares.
\newblock {\em SIAM Journal on Algebraic Discrete Methods}, 2(3):240--243,
  1981.
\newblock \href {http://arxiv.org/abs/https://doi.org/10.1137/0602026}
  {\path{arXiv:https://doi.org/10.1137/0602026}}, \href
  {https://doi.org/10.1137/0602026} {\path{doi:10.1137/0602026}}.

\bibitem{10.1145/301250.301369Kumar}
V.~S. Anil~Kumar and H.~Ramesh.
\newblock Covering rectilinear polygons with axis-parallel rectangles.
\newblock In {\em Proceedings of the Thirty-First Annual ACM Symposium on
  Theory of Computing}, STOC '99, page 445–454, New York, NY, USA, 1999.
  Association for Computing Machinery.
\newblock \href {https://doi.org/10.1145/301250.301369}
  {\path{doi:10.1145/301250.301369}}.

\bibitem{aupperle1988covering}
L.~J. Aupperle, H.~E. Conn, J.~M. Keil, and Joseph O'Rourke.
\newblock Covering orthogonal polygons with squares.
\newblock {\em Proc. 26th Allerton Conf. Commun. Control Comput.}, 1988.
\newblock URL:
  \url{https://www.science.smith.edu/~jorourke/Papers/ConnJORsquares.pdf}.

\bibitem{BarYehuda2014CoveringPW}
Reuven Bar-Yehuda.
\newblock Covering polygons with squares.
\newblock 2014.
\newblock URL: \url{https://api.semanticscholar.org/CorpusID:50721774}.

\bibitem{bar1994}
Reuven Bar-Yehuda and Eyal Ben-Chanoch.
\newblock A linear-time algorithm for covering simple polygons with similar
  rectanges.
\newblock {\em International Journal of Computational Geometry \&
  Applications}, 06(01):79--102, 1996.
\newblock \href
  {http://arxiv.org/abs/https://doi.org/10.1142/S021819599600006X}
  {\path{arXiv:https://doi.org/10.1142/S021819599600006X}}, \href
  {https://doi.org/10.1142/S021819599600006X}
  {\path{doi:10.1142/S021819599600006X}}.

\bibitem{ben-or1983}
Michael Ben-Or.
\newblock Lower bounds for algebraic computation trees.
\newblock In {\em Proceedings of the Fifteenth Annual ACM Symposium on Theory
  of Computing}, STOC '83, page 80–86, New York, NY, USA, 1983. Association
  for Computing Machinery.
\newblock \href {https://doi.org/10.1145/800061.808735}
  {\path{doi:10.1145/800061.808735}}.

\bibitem{10.1007/978-1-4613-8369-7_1/chordal}
Jean R.~S. Blair and Barry Peyton.
\newblock An introduction to chordal graphs and clique trees.
\newblock In Alan George, John~R. Gilbert, and Joseph W.~H. Liu, editors, {\em
  Graph Theory and Sparse Matrix Computation}, pages 1--29, New York, NY, 1993.
  Springer New York.

\bibitem{bart1986}
Bart Braden.
\newblock The surveyor's area formula.
\newblock {\em The College Mathematics Journal}, 17(4):326--337, 1986.
\newblock \href
  {http://arxiv.org/abs/https://doi.org/10.1080/07468342.1986.11972974}
  {\path{arXiv:https://doi.org/10.1080/07468342.1986.11972974}}, \href
  {https://doi.org/10.1080/07468342.1986.11972974}
  {\path{doi:10.1080/07468342.1986.11972974}}.

\bibitem{Culberson21976}
J.C. Culberson and R.A. Reckhow.
\newblock Covering polygons is hard.
\newblock In {\em [Proceedings 1988] 29th Annual Symposium on Foundations of
  Computer Science}, pages 601--611, 1988.
\newblock \href {https://doi.org/10.1109/SFCS.1988.21976}
  {\path{doi:10.1109/SFCS.1988.21976}}.

\bibitem{garey1979computers}
Michael~R Garey and David~S Johnson.
\newblock {\em Computers and intractability}, volume 174.
\newblock freeman San Francisco, 1979.

\bibitem{Gavril1972AlgorithmsFM}
Fanica Gavril.
\newblock Algorithms for minimum coloring, maximum clique, minimum covering by
  cliques, and maximum independent set of a chordal graph.
\newblock {\em SIAM J. Comput.}, 1:180--187, 1972.
\newblock URL: \url{https://api.semanticscholar.org/CorpusID:1111479}.

\bibitem{KENYON1996272}
Richard Kenyon.
\newblock Tiling a rectangle with the fewest squares.
\newblock {\em Journal of Combinatorial Theory, Series A}, 76(2):272--291,
  1996.
\newblock URL:
  \url{https://www.sciencedirect.com/science/article/pii/S0097316596901041},
  \href {https://doi.org/https://doi.org/10.1006/jcta.1996.0104}
  {\path{doi:https://doi.org/10.1006/jcta.1996.0104}}.

\bibitem{Klee1977CanTM}
Victor Klee.
\newblock Can the measure of $\cup[ a_i , b_i ]$ be computed in less than $o(n
  \log n)$ steps?
\newblock {\em American Mathematical Monthly}, 84:284, 1977.
\newblock URL: \url{https://api.semanticscholar.org/CorpusID:124770860}.

\bibitem{doi:10.1137/0211025}
David Lichtenstein.
\newblock Planar formulae and their uses.
\newblock {\em SIAM Journal on Computing}, 11(2):329--343, 1982.
\newblock \href {https://doi.org/10.1137/0211025} {\path{doi:10.1137/0211025}}.

\bibitem{DBLP:journals/algorithmica/Moitra91}
Dipen Moitra.
\newblock Finding a minimal cover for binary images: An optimal parallel
  algorithm.
\newblock {\em Algorithmica}, 6(5):624--657, 1991.
\newblock \href {https://doi.org/10.1007/BF01759065}
  {\path{doi:10.1007/BF01759065}}.

\bibitem{o1987art}
J.~O'Rourke.
\newblock {\em Art Gallery Theorems and Algorithms}.
\newblock International series of monographs on computer science. Oxford
  University Press, 1987.
\newblock URL: \url{https://books.google.co.in/books?id=aPZQAAAAMAAJ}.

\bibitem{vazirani1997approximation}
Vijay~V Vazirani.
\newblock Approximation algorithms.
\newblock {\em Georgia Inst. Tech}, 1997.

\bibitem{WALTERS20092913}
Mark Walters.
\newblock Rectangles as sums of squares.
\newblock {\em Discrete Mathematics}, 309(9):2913--2921, 2009.
\newblock URL:
  \url{https://www.sciencedirect.com/science/article/pii/S0012365X08004780},
  \href {https://doi.org/https://doi.org/10.1016/j.disc.2008.07.028}
  {\path{doi:https://doi.org/10.1016/j.disc.2008.07.028}}.

\end{thebibliography}

\end{document}